\documentclass[11pt,a4paper]{article}

\usepackage{amsmath,amsthm,amssymb}
\usepackage{tikz}
\usepackage{subfig}

\usepackage[auth-sc,affil-it]{authblk}

\usepackage[utf8]{inputenc}
\usepackage[english]{babel}
\usepackage[T1]{fontenc}
\usepackage{lmodern}
\usepackage{microtype}
\usepackage[
    linktocpage=true,
    colorlinks=true,
    linkcolor=purple,
    citecolor=black,
    urlcolor=purple,
    bookmarksnumbered=true,
]{hyperref}

\usepackage[totalwidth=160mm, totalheight=247mm]{geometry}

\theoremstyle{plain}
\newtheorem{theo}{Theorem}[section]
\newtheorem{prop}[theo]{Proposition}
\newtheorem{lemm}[theo]{Lemma}

\theoremstyle{definition}
\newtheorem{defi}[theo]{Definition}
\newtheorem{exam}[theo]{Example}
\newtheorem{rema}[theo]{Remark}

\newtheorem{ntt}{Notation}

\renewcommand{\leq}{\leqslant}
\renewcommand{\geq}{\geqslant}

\newcommand\vect[1]{\mathbf{#1}}
\newcommand{\Ms}{{{\bfM}_\sigma}}
\newcommand{\Msinv}{{{\bfM}^{-1}_\sigma}}

\newcommand{\MFSi}{{\bfM^{\textup{FS}}_1}}
\newcommand{\MFSii}{{\bfM^{\textup{FS}}_2}}
\newcommand{\MFSiii}{{\bfM^{\textup{FS}}_3}}
\newcommand{\MFSx}{{\bfM^{\textup{FS}}_i}}
\newcommand{\MFS}{\bfM^{\textup{FS}}}

\newcommand{\sFS}{\sigma^{\textup{FS}}}
\newcommand{\sFSi}{{\sigma^{\textup{FS}}_1}}
\newcommand{\sFSii}{{\sigma^{\textup{FS}}_2}}
\newcommand{\sFSiii}{{\sigma^{\textup{FS}}_3}}
\newcommand{\sFSx}{{\sigma^{\textup{FS}}_i}}
\newcommand{\SFS}{\Sigma^{\textup{FS}}}
\newcommand{\SFSi}{{\Sigma^{\textup{FS}}_1}}
\newcommand{\SFSii}{{\Sigma^{\textup{FS}}_2}}
\newcommand{\SFSiii}{{\Sigma^{\textup{FS}}_3}}
\newcommand{\SFSx}{{\Sigma^{\textup{FS}}_i}}

\newcommand{\bbN}{{\mathbb N}}
\newcommand{\bbZ}{{\mathbb Z}}
\newcommand{\bbQ}{{\mathbb Q}}
\newcommand{\bbR}{{\mathbb R}}

\newcommand{\mcA}{\mathcal A}

\newcommand{\mcU}{\mathcal U}

\newcommand{\mcL}{\mathcal L}

\newcommand{\bfF}{{\mathbf F}}
\newcommand{\bfFS}{{\mathbf{FS}}}
\newcommand{\bfM}{{\mathbf M}}
\newcommand{\bfP}{{\mathbf P}}
\newcommand{\bfT}{{\mathbf T}}
\newcommand{\bfGa}{{\mathbf \Gamma}}

\newcommand{\frakP}{{\mathfrak P}}
\newcommand{\Otp}{{\mathcal O^+_3}}
\newcommand{\Otpd}[1]{{\mathcal O^+_{#1}}}
\newcommand{\Fthree}{{\mathcal F_3}}

\newcommand{\LFS}{\mathcal L_\textup{FS}}
\newcommand{\Ledge}{\mathcal L_\textup{edge}}

\newcommand{\bfe}{{\mathbf e}}
\newcommand{\bfu}{{\mathbf u}}
\newcommand{\bfv}{{\mathbf v}}
\newcommand{\bfx}{{\mathbf x}}
\newcommand{\bfy}{{\mathbf y}}
\newcommand{\bfz}{{\mathbf z}}

\newcommand{\Gv}{{\mathbf\bfGa_{\bfv}}}
\newcommand{\transp}[1]{{{}^\textup{t} #1}}
\newcommand{\EOS}{{{\mathbf E}_1^\star}}
\newcommand{\EOSS}{{{\mathbf E}_1^\star(\sigma)}}

\newcommand{\svect}[3]{%
\big(\begin{smallmatrix}%
#1 \\ #2 \\ #3%
\end{smallmatrix}\big)}

\newcommand{\myvcenter}[1]{\ensuremath{\vcenter{\hbox{#1}}}}

\AtBeginDocument{}
\AtBeginDocument{}
\usepackage{enumitem}
\setlist{noitemsep}

\newcommand{\TODO}{\colorbox{yellow}{\textcolor{red}{\Large \textbf{TODO}}}}

\title{\textbf{Critical connectedness of \\ thin arithmetical discrete planes}}

\author[1]{Val\'erie Berth\'e}
\author[2]{Damien Jamet}
\author[1,3]{Timo Jolivet}
\author[4]{Xavier Proven\c cal}

\affil[1]{%
    LIAFA, CNRS, Universit\'e Paris Diderot, France
}
\affil[2]{%
    LORIA, Universit\'e de Lorraine, France
}
\affil[3]{%
    Department of Mathematics, University of Turku, Finland
}
\affil[4]{%
    LAMA, Universit\'e de Savoie, France
}

\date{}

\begin{document}

\maketitle

\begin{abstract}
An arithmetical discrete plane is said to have critical connecting thickness if its thickness
is equal to the infimum of the set of values that preserve its $2$-connectedness.
This infimum thickness can be computed thanks to the fully subtractive algorithm.
This multidimensional continued fraction algorithm consists, in its linear form,
in subtracting the smallest entry to the other ones.
We provide a characterization of the discrete planes with critical thickness that have zero intercept
and that are $2$-connected.
Our tools rely on the notion of dual substitution which is a geometric
version of the usual notion of substitution acting on words.
We associate with the fully subtractive algorithm a set of substitutions whose incidence matrix
is provided by the matrices of the algorithm,
and prove that their geometric counterparts generate arithmetic discrete planes.
\end{abstract}

\section{Introduction}
This paper studies the connectedness of thin arithmetic discrete planes in the three-dimensional space.
We focus on the notion of $2$-connectedness, and we restrict ourselves to planes with zero intercept
that have critical thickness, that is, planes whose thickness
is the infimum of the set of all the $\omega \in \bbR_+ $ such that the plane of thickness $\omega$ is $2$-connected
(see Definitions~\ref{def:arith} and~\ref{def:connect}).
Let us recall that standard arithmetic discrete planes are known to be $2$-connected,
whereas naive ones are too thin to be $2$-connected.
We thus consider planes with a thickness that lies between the naive and the standard cases.

The problem of the computation of the critical thickness was raised in~\cite{BB04}.
It has been answered in~\cite{JT06,JT09,DJT09}, with an algorithm
(called the connecting thickness algorithm) that can be expressed
in terms of a multidimensional continued fraction algorithm,
namely the so-called fully subtractive algorithm.
The connecting thickness algorithm explicitly yields the value of the critical thickness when it halts,
and this value can be computed when the algorithm enters a loop (possibly infinite).
Furthermore, the set $\Fthree$ of vectors for which
the algorithm enters an infinite loop has Lebesgue measure zero,
as a consequence of results of~\cite{MN89} in the context of a percolation model
defined by rotations on the unit circle.
Our main goal is to provide a necessary and sufficient condition so that a
discrete plane with intercept zero and critical thickness
is $2$-connected when its normal vector belongs to $\Fthree$.

The tools we use here are combinatorial and are issued from numeration systems and combinatorics on words.
Our methods rely on a combinatorial generation method based on the notion of substitution
for the planes under study (see Section~\ref{subsec:sub} for more details).
In Section~\ref{sec:trans}, we construct a sequence of finite patterns $(\bfT_n)_n$
of the planes with critical thickness,
and we prove that these patterns are all $2$-connected when the parameters belong to $\Fthree$.
We then relate these finite patterns $\bfT_n$ with thinner patterns $\bfP_n$
that belong to the naive discrete plane with same parameters.
These pattern are generated in terms of a geometric interpretation of the fully
subtractive algorithm via the geometric formalism of dual substitutions.
Finally, in Section~\ref{sec:gen}, the thinner patterns $\bfP_n$ are proved
to generate the full naive plane.
This yields the $2$-connectedness of the critical plane (see Section~\ref{sec:main}).
In other words, we use the fact that the underlying naive plane
provides a relatively dense skeleton of the critical plane.
Theorem~\ref{thm::main} highlights the limit behavior of discrete plane with
critical thickness which is reminiscent of similar phenomena occurring in
percolation theory~\cite{Mee89}.

Note that in~\cite{DomV2012}, Domenjoud and Vuillon also studied the same patterns from the viewpoint
of geometrical closure (an analogue of palindromic closure in word combinatorics),
and used symmetries to build them. From this
approach, they  deduced topological results (including connectedness) and
showed how these patterns generalize Christoffel words to higher dimensions.
The present use of substitutions in order to address the problem of connectedness
provides an original viewpoint on these objects.
This paper is an extended version of~\cite{BJJP2013}. It also extends the study of~\cite{Jam12}
devoted to particular planes the parameters of which belong to the cubic extension generated
by the Tribonacci number.
Observe also that similar results in higher dimension have been proved in~\cite{DomProVui14}.

\section{Basic notions and notation}

\subsection{Discrete and stepped planes}

Let $\left(\bfe_1,\bfe_2,\bfe_3\right)$ be the canonical basis of $\bbR^3$,
and let $\langle \cdot , \cdot \rangle $ stand for the usual scalar product on $\bbR^3$.
Given $\bfv \in \bbR^3$ and $i \in \{1,2,3\}$, we let $\bfv_i = \langle \bfv, \bfe_i \rangle$
denote the $i$th coordinate of $\bfv$ in the basis $\left(\bfe_1,\bfe_2,\bfe_3\right)$.

\begin{defi}[Arithmetical discrete plane~\cite{Rev91,And03}]
\label{def:arith}
Given $\bfv \in \bbR^3_+$ and $\omega \in \bbR_+$,
the \emph{arithmetical discrete plane} with \emph{normal vector} $\bfv$
and \emph{thickness} $\omega$ is the set
$\frakP(\bfv,\omega)$ defined as follows:
\[
\frakP(\bfv,\omega) = \left \{\bfx \in \bbZ^3 : 0 \leq \langle \bfx,\bfv\rangle < \omega \right\}.
\]
If $\omega = \| \bfv \|_\infty = \max \{|\bfv_1|,|\bfv_2|,|\bfv_3|\}$
(resp. $\omega = \| \bfv \|_1 = |\bfv_1| + |\bfv_2| + |\bfv_3|$),
then $\frakP(\bfv,\omega)$ is said to be a \emph{naive arithmetical discrete plane}
(resp. \emph{standard arithmetical discrete plane}).
\end{defi}

Note that we consider here only planes with zero intercept, where the intercept
$\mu$ allows the more  general definition
$\left \{\bfx \in \bbZ^3 : 0 \leq \langle \bfx,\bfv\rangle +\mu < \omega \right\}.$
Even if any finite subset of a digitized plane can be represented
as a subset of an arithmetical discrete plane with integer parameters (by taking a suitable $\omega$ large enough),
we do not restrict ourselves here to finite sets
and we consider general arithmetical discrete plane with possibly non-integer parameters.

We will also deal with another discrete approximation of Euclidean planes, namely \emph{stepped planes}.
They can be considered as a more geometrical version,
in the sense that they consist of \emph{unit faces} instead of just points of $\bbZ^3$.

\definecolor{facecolor}{rgb}{0.8,0.8,0.8}
\begin{defi}[Unit faces, stepped planes]
\label{def:stepped}
A \emph{unit face} $[\bfx,i]^\star$ is defined as:
\[
\begin{array}{ccccc}
\lbrack \bfx, 1 \rbrack^\star & = & \{\bfx + \lambda \bfe_2 + \mu \bfe_3 : \lambda,\mu \in [0,1] \} & = &
    \myvcenter{\begin{tikzpicture}
    [x={(-0.216506cm,-0.125000cm)}, y={(0.216506cm,-0.125000cm)}, z={(0.000000cm,0.250000cm)}]
    \fill[fill=facecolor, draw=black, shift={(0,0,0)}]
    (0, 0, 0) -- (0, 1, 0) -- (0, 1, 1) -- (0, 0, 1) -- cycle;
    \node[circle,fill=black,draw=black,minimum size=1.2mm,inner sep=0pt] at (0,0,0) {};
    \end{tikzpicture}} \\
\lbrack \bfx, 2 \rbrack^\star & = & \{\bfx + \lambda \bfe_1 + \mu \bfe_3 : \lambda,\mu \in [0,1] \} & = &
    \myvcenter{\begin{tikzpicture}
    [x={(-0.216506cm,-0.125000cm)}, y={(0.216506cm,-0.125000cm)}, z={(0.000000cm,0.250000cm)}]
    \fill[fill=facecolor, draw=black, shift={(0,0,0)}]
    (0, 0, 0) -- (0, 0, 1) -- (1, 0, 1) -- (1, 0, 0) -- cycle;
    \node[circle,fill=black,draw=black,minimum size=1.2mm,inner sep=0pt] at (0,0,0) {};
    \end{tikzpicture}} \\
\lbrack \bfx, 3 \rbrack^\star & = & \{\bfx + \lambda \bfe_1 + \mu \bfe_2 : \lambda,\mu \in [0,1] \} & = &
    \myvcenter{\begin{tikzpicture}
    [x={(-0.216506cm,-0.125000cm)}, y={(0.216506cm,-0.125000cm)}, z={(0.000000cm,0.250000cm)}]
    \fill[fill=facecolor, draw=black, shift={(0,0,0)}]
    (0, 0, 0) -- (1, 0, 0) -- (1, 1, 0) -- (0, 1, 0) -- cycle;
    \node[circle,fill=black,draw=black,minimum size=1.2mm,inner sep=0pt] at (0,0,0) {};
    \end{tikzpicture}}
\end{array}
\]
where $i \in \{1,2,3\}$ is the \emph{type} of $[\bfx,i]^\star$,
and $\bfx \in \bbZ^3$ is the \emph{distinguished vertex} of $[\bfx,i]^\star$.
Let $\bfv \in \bbR^3_+$.
The \emph{stepped plane} $\Gv$ is the union of unit faces defined by:
\[
\Gv = \{ [\bfx,i]^\star : 0 \leq \langle \bfx,\bfv \rangle < \langle \bfe_i,\bfv \rangle \}.
\]
\end{defi}
The notation $\bfx+[\bfy,i]^\star$ stands for the unit face $[\bfx+\bfy,i]^\star$.
\begin{rema}
\label{rq:planes}
The set of distinguished vertices of $\Gv$ is the naive arithmetical discrete plane $ \frakP(\bfv,\|\bfv\|_\infty)$,
whereas the set of all vertices of the faces of $\Gv$
is the standard arithmetical discrete plane $\frakP(\bfv,\|\bfv\|_1)$ (see~\cite{ABI02,ABS04}).
\end{rema}


\subsection{Connecting thickness and the fully subtractive algorithm}
\begin{defi}[Adjacency, connectedness]
Two distinct elements $\bfx$ and $\bfy$ of $\bbZ^3$
are \emph{$2$-adjacent} if $\|\bfx-\bfy\|_1 = 1$.
A subset $A \subseteq \bbZ^3$ is $2$-connected
if it is not empty and if for every $\bfx,\bfy \in A$, there exist
$\bfx_1, \dots, \bfx_n \in A$
such that $\bfx_i$ and $\bfx_{i+1}$ are $2$-adjacent for all $i \in \{1,\dots,n-1\}$,
with $\bfx_1=\bfx$ and $\bfx_n=\bfy$.
\end{defi}
We stress the fact that a $2$-connected set is assumed here to be non-empty.
This will prove to be more convenient for the notion of connecting thickness.

\begin{defi}[Connecting thickness~\cite{JT06}]
\label{def:connect}
Given $\bfv \in \bbR^3_+$, the \emph{connecting thickness} $\Omega(\bfv)$ is defined by:
\begin{equation*}
\Omega(\bfv) = \inf \left\{\omega \in \bbR_+ : \frakP(\bfv,\omega) \text{ is $2$-connected} \right\}.
\end{equation*}
\end{defi}

The two above definitions are illustrated in Figure~\ref{fig:thicknesses}.
Note that these definitions focus on $2$-connectedness,
but similar definitions are also possible for $\kappa$-connectedness with $\kappa = 0$ or $1$ (in ${\mathbb Z}^3$).
However, the value of $\Omega(\bfv)$ for such alternative definitions
can be directly deduced from the value of $\Omega(\bfv)$ for $2$-connectedness (see~\cite{JT06}),
so it is natural to restrict to $2$-connectedness.

\begin{figure}[ht]%
\centering
\subfloat[][$\omega = 1$]{%
    \label{subfig:omega1}%
    \includegraphics[height=28mm]{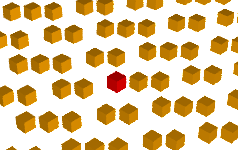}
} \hfil
\subfloat[][$\omega = 2.5$]{%
    \label{subfig:omega2}%
    \includegraphics[height=28mm]{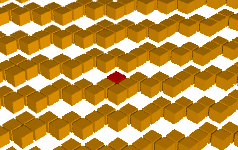}
}\\
\subfloat[][$\omega = 4$]{%
    \label{subfig:omega3}%
    \includegraphics[height=28mm]{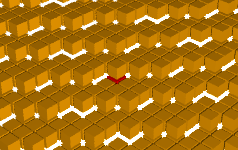}
} \hfil
\subfloat[][$\omega = 6$]{%
    \label{subfig:omega4}%
    \includegraphics[height=28mm]{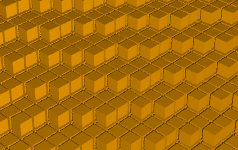}
}
\caption[]{
The arithmetical discrete plane $\frakP(\bfv,\omega)$
with $\bfv = (1,\sqrt{2},\pi)$ and varying thicknesses~$\omega$.
In \subref{subfig:omega1} and \subref{subfig:omega2},
$\frakP(\bfv,1)$ and $\frakP(\bfv,2.5)$ are not $2$-connected,
but in \subref{subfig:omega3} and \subref{subfig:omega4},
$\frakP(\bfv,4)$ and $\frakP(\bfv,6)$ are $2$-connected.
It follows that $2.5 \leq \Omega(\bfv) \leq 4$.
}%
\label{fig:thicknesses}%
\end{figure}

\paragraph{Computing the connecting thickness}

A first observation which follows easily from the definitions is that for every $\varepsilon > 0$,
the discrete plane $\frakP\left(\bfv,\Omega(\bfv) - \varepsilon\right)$ is not $2$-connected.
Moreover, a first approximation of $\Omega(\bfv) $ is provided by
$\| \bfv \|_\infty \leq \Omega(\bfv) \leq \|\bfv\|_1$ (see~\cite[Lemma~10]{AAS97}) and implies that,
for every $\varepsilon > 0$, $\frakP\left(\bfv,\Omega(\bfv) + \varepsilon\right)$ is $2$-connected.
Consequently, one deduces that the set
$\left\{\omega \in \bbR_+ : \frakP(\bfv,\omega) \text{ is $2$-connected} \right\}$ is an interval.

It is shown in~\cite{JT09} how to compute $\Omega(\bfv)$
from the expansion of the vector $\bfv$ by using the \emph{fully subtractive
algorithm}. Given a vector $\bfv = ( \bfv_1, \bfv_2, \bfv_3 )$ with $\bfv_1 =
\min(\bfv_1, \bfv_2, \bfv_3)$, this algorithm is defined as
$\bfFS(v) = (\bfv_1, \bfv_2 - \bfv_1, \bfv_3 - \bfv_1)$.
The fully subtractive algorithm is one of the many possible generalizations of the Euclid algorithm;
it consists in subtracting the smallest coordinate to all the other coordinates.
It yields a well-studied multidimensional continued fraction algorithm (see e.g.~\cite{KM95,FokKraaiNak,SCH,TijZam}).

In order to compute $\Omega(\bfv)$ we may assume without loss of generality that
$0 \leq \bfv_1 \leq \bfv_2 \leq \bfv_3$.
We thus restrict ourselves in the sequel to the set of parameters $\Otp =
\left\{\bfv \in \bbR^3 : 0 \leq \bfv_1 \leq \bfv_2 \leq \bfv_3 \right\}$ and
consider the \emph{ordered fully subtractive algorithm} $\bfF : \Otp \rightarrow
\Otp$ defined by:
\[
\bfF(\bfv) \ = \
\left\{
    \begin{array}{l}
        (\bfv_1, \bfv_2-\bfv_1, \bfv_3-\bfv_1)
        \quad \text{if} \ \bfv_1 \leq \bfv_2-\bfv_1 \leq \bfv_3-\bfv_1 \\
        (\bfv_2-\bfv_1, \bfv_1, \bfv_3-\bfv_1)
        \quad \text{if} \ \bfv_2-\bfv_1 < \bfv_1 \leq \bfv_3-\bfv_1 \\
        (\bfv_2-\bfv_1, \bfv_3-\bfv_1, \bfv_1)
        \quad \text{if} \ \bfv_2-\bfv_1 \leq \bfv_3-\bfv_1 < \bfv_1.
    \end{array}
\right.
\]
Iterating $\bfF$ on a vector $\bfv \in \Otp$ yields an infinite sequence of vectors
$(\bfv^{(n)})_{n\geq 0}$ defined by $\bfv^{(n)} = \bfF^n(\bfv)$ and $\bfv^{(0)} = \bfv$.
This can be rewritten in matrix form by
$\bfv = \MFS_{i_1} \dots \MFS_{i_n} \bfv^{(n)}$,
where the matrices $\MFSi$, $\MFSii$ and $\MFSiii$ are defined by
\[
\MFSi = \begin{bmatrix} 1 & 0 & 0 \\ 1 & 1 & 0 \\ 1 & 0 & 1 \end{bmatrix}, \quad
\MFSii  = \begin{bmatrix} 0 & 1 & 0 \\ 1 & 1 & 0 \\ 0 & 1 & 1 \end{bmatrix}, \quad
\MFSiii = \begin{bmatrix} 0 & 0 & 1 \\ 1 & 0 & 1 \\ 0 & 1 & 1 \end{bmatrix}.
\]
The \emph{$\bfF$-expansion} of $\bfv$
is the sequence $(i_n)_{n\geq 0} \in \{1,2,3\}^{\mathbb N}$ defined above,
that is, the sequence $(i_n)_{n\geq 0}$ such that
$\smash{\bfv^{(n)} = {\MFS_{i_n}}^{-1}\bfv^{(n-1)}}$ for all $n \geq 1$.


The link between the connecting thickness and the fully subtractive algorithm
$\bfF$ is provided by the following algorithm  that   we call the \emph{connecting thickness algorithm}:
\begin{quote}
\textup{\texttt{%
    def connecting\_thickness($\bfv$): \\
    \phantom{.} \quad if $\bfv_1 + \bfv_2 \leq \bfv_3$: return $\bfv_3$ \\
    \phantom{.} \quad else: return $\bfv_1$ + connecting\_thickness($\bfF(\bfv)$)
}}
\end{quote}

For some input vectors, the above algorithm never stops. Let $\Fthree$ be this
set of vectors:
\[
\Fthree = \left\{\bfv \in \Otp : \bfv^{(n)}_1 + \bfv^{(n)}_2 > \bfv^{(n)}_3 \text{ for all $n \in \bbN$} \right\}.
\]
This set has been studied in~\cite{Mee89},
and its properties are similar to that of the Rauzy gasket~\cite{ASt13}.
It will play a crucial role in the characterization stated in Theorem~\ref{thm::main}.

\begin{theo}[\cite{JT09},\cite{DJT09}]
\label{theo:algo_cc}
Let $\bfv \in \Otp$.
The arithmetical discrete plane $\frakP(\bfv,\omega)$ is $2$-connected
if and only if so is $\frakP(\bfF(\bfv),\omega-\bfv_1)$.
Consequently,  $\frakP(\bfv,\Omega(\bfv))$ is $2$-connected
if and only if so is $\frakP(\bfF(\bfv),\Omega(\bfF(\bfv))$,
\[
\Omega(\bfv) = \Omega(\bfF(\bfv))+\bfv_1,
\]
and,
\[
  \Omega(\bfv) = \left\{ \begin{array}{ll}
    \textup{\texttt{connecting\_thichness}}(\bfv) & \textrm{ if } \bfv \not\in \Fthree,\\
    \sum^\infty_{n=0}{\bfv^{(n)}_1} & \textrm{ if } \bfv \in \Fthree.
  \end{array} \right.
\]
Moreover, if $\bfv \in \Fthree$ then
\[
\Omega(\bfv) =\sum^\infty_{n=0}{\bfv^{(n)}_1} = \|\bfv\|_1/2.
\]
\end{theo}
\begin{proof}
According to  Theorem~\ref{theo:algo_cc}, we have:
$\bfv_1+\bfv_2+\bfv_3-2\Omega(\bfv) = \bfv^{(i)}_1+\bfv^{(i)}_2+\bfv^{(i)}_3-2\Omega(\bfv^{(i)})$
for all $i \in \{1,\ldots,n\}$.
Since $\Omega(\bfv^{(n)}) \leq \|\bfv^{(n)}\|_1$
and $\lim_{n\to\infty}{\bfv^{(n)}} = \mathbf 0$, then $\lim_{n\to\infty} \Omega(\bfv^{(n)}) = 0$
and the result follows.
\end{proof}
\begin{exam}
Let $\bfv = (1, \sqrt{13}, \sqrt{17})$.
Iterating $\bfF$ yields
\begin{align*}
\bfv^{(1)} &= (1, \sqrt{13} - 1, \sqrt{17} - 1) \\
\bfv^{(2)} &= (1, \sqrt{13} - 2, \sqrt{17} - 2) \\
\bfv^{(3)} &= (\sqrt{13} - 3, 1, \sqrt{17} - 3) \\
\bfv^{(4)} &= (4 - \sqrt{13}, \sqrt{17} - \sqrt{13}, \sqrt{13} - 3) \\
\bfv^{(5)} &= (\sqrt{17} - 4, 2\sqrt{13} - 7, 4 - \sqrt{13}).
\end{align*}
The connecting thickness algorithm stops because $\bfv^{(5)}_1 + \bfv^{(5)}_2 \leq \bfv^{(5)}_3$, so
$\Omega(\bfv) =
1 +
1 +
1 +
\sqrt{13} - 3 +
4 - \sqrt{13} +
4 - \sqrt{13}
= 8 - \sqrt{13}$.
Similarly, if $\bfv = (1, \sqrt[3]{10}, \pi)$,
then the algorithm stops after $19$ steps
and $\Omega(\bfv) = 2\pi - 98\sqrt[3]{10} + 208$.
It is also possible to exhibit some examples where the algorithm never stops,
for instance by choosing a right eigenvector of  a finite product  of the matrices
${{\bfM^{\textup{FS}}_i}}$, for $i=1,2,3$, of $\bfF$.
This is the case for example with the vector
$\bfv = (1, \alpha + 1, \alpha^2 + \alpha + 1) = (1, 1.54\ldots, 1.84\ldots)$,
where $\alpha = 0.54\ldots$ is the real root of $x^3 + x^2 + x + 1$.
\end{exam}

The next property  provides a characterization  with respect to  the $\bfF$-expansion of a vector $\bfv$
of  its belonging to $\Fthree$.

\begin{lemm}\label{lem:3infty}
We have $\bfv \in \Fthree$ if and only if the $\bfF$-expansion $(i_n)_n$ of $\bfv$  under the ordered
fully subtractive algorithm contains infinitely many occurrences of $3$.
\end{lemm}

\begin{proof}
Let $\bfv \in \Fthree$, and assume by contradiction  that  $(i_n)_{n \in {\mathbb N}}$
does not take the value $3$. One  thus checks that
$\lim_{n\to\infty}{\bfv_1^{(n)}}=\lim_{n\to\infty}{\bfv_2^{(n)}}= 0$, and hence,
$\lim_{n\to\infty}{\bfv_3^{(n)}}=0$.
Furthermore,
$\bfv_1^{(n+1)}+\bfv_2^{(n+1)}+\bfv_3^{(n+1)}+2\bfv_1^{(n)}=\bfv_1^{(n)}+\bfv_2^{(n)}+\bfv_3^{(n)}$, for  all $n$.
Hence
\[
\frac{\bfv_1+\bfv_2+\bfv_3}{2}=\sum _{n\geq 1} \bfv_1^{(n)}.
\]
Note that the expansion  of $(\bfv_1,\bfv_2,\bfv_1+\bfv_2)$ obtained by applying the
ordered fully subtractive  algorithm  $\bfF$ to  $(\bfv_1,\bfv_2,\bfv_1+\bfv_2)$
coincides on the first two  coordinates with   the   expansion of $(\bfv_1,\bfv_2, \bfv_3)$, that is,
$\bfF^n (\bfv_1,\bfv_2,\bfv_1+\bfv_2)=(\bfv_1^{(n)},\bfv_2^{(n)},\bfv_1^{(n)}+\bfv_2^{(n)})$ for all $n \geq 1$.
Consequently, here again  $\bfv_1+\bfv_2=\sum _{n\geq 1} \bfv_1^{(n)}$, which implies
$\bfv_3=\bfv_1+\bfv_2,$ a contradiction.
Hence, the sequence $(i_n)_n$ takes the value $3$ at least once,
and by repeating the argument, infinitely many times.

Conversely, assume that $\bfv \not \in \Fthree$.
If $\bfv^{(n)} _1 + \bfv^{(n)} _2 \leq \bfv^{(n)} _3$ for some $n$, then
$\bfv^{(m)} _1 + \bfv^{(m)} _2 \leq \bfv^{(m)} _3$ for all $m \geq n$, and in particular,
$\bfv^{(m)} _3-\bfv^{(m)} _1 \geq \bfv^{(m)} _2\geq \bfv^{(m)} _1$.
This implies that  the sequence $(i_m)_{m \geq n}$ will never take the value $3$.
\end{proof}


\begin{rema}
If $\bfv^{(n)} _1 + \bfv^{(n)} _2 < \bfv^{(n)} _3$ for some $n$, then
$\lim_{n\to\infty}{\bfv^{(n)}}\neq \mathbf 0$.
If $\bfv^{(n)} _1 + \bfv^{(n)} _2  = \bfv^{(n)} _3$ for some $n$,
we can say nothing concerning  the  fact that   $\lim_{n\to\infty}{\bfv^{(n)}}=\mathbf 0$.
Indeed, take $\bfv=(\bfv_1,\bfv_1,2 \bfv_1)$ for some $\bfv_1>0$.
Then $\lim_{n\to\infty}{\bfv^{(n)}}=(0,\bfv_1,\bfv_1)\neq \mathbf 0$.
Now take $\bfv=(1/\varphi^2,1/\varphi,1)$ with $1/\varphi+1/\varphi ^2 =1$ and $\varphi >0$.
One  checks that $\lim_{n\to\infty}{\bfv^{(n)}}=\mathbf 0$.
\end{rema}

\begin{prop}
\label{prop:F3}
If $\bfv \in \Fthree$, then $\dim_{\bbQ}(\bfv_1,\bfv_2,\bfv_3)=3$.
\end{prop}

\begin{proof}
The present  proof  is inspired by~\cite{Av-Del13}.
Let $\bfv \in \Fthree$  and  $\bfx \in {\mathbb Z}^3$ such that
$\langle \bfv, \bfx \rangle=0$.  We set $s_{\bfv}=(\sum_{i=1}^3 \bfv_i)/2 $.
One can rewrite  $\bfv$ as $\bfv=s_{\bfv}(\sum_{i=1}^3 \bfe_i)- \sum_{i=1}^3 (s_{\bfv}-\bfv_i) \bfe_i$,
which yields $\langle \bfv, \bfx \rangle=s_{\bfv}( \sum_{i=1}^3 \bfx_i) -\sum_{i=1}^3 (s_{\bfv}-\bfv_i) \bfx_i.$
The fact that $\langle \bfv, \bfx \rangle=0$ implies that
\[
\sum_{i=1}^3 \bfx_i =\sum_{i=1}^3 (1-\bfv_i/s_{\bfv}) \bfx_i.
\]

Furthermore, one checks that
$0 < 1- \bfv_i/s_{\bfv} < 1$, for $i=1,2,3$, since $\bfv \in \Fthree$.
The matrices
\[
{\mathbf A}^{\textup{FS}}_{1,\bfv} =
    \begin{bmatrix} 1-\bfv_1 /s_{\bfv}& 1-\bfv_2/s_{\bfv} & 1-\bfv_3/s_{\bfv} \\ 0 & 1 & 0 \\ 0& 0 & 1 \end{bmatrix},
    \quad
{\mathbf A}^{\textup{FS}}_{2,\bfv} =
    \begin{bmatrix} 0 & 1 & 0 \\ 1-\bfv_1/s_{\bfv} & 1 -\bfv_2/s_{\bfv}& 1-\bfv_3/s_{\bfv} \\ 0 & 0 & 1 \end{bmatrix},
\]
\[
{\mathbf A}^{\textup{FS}}_{3,\bfv} =
    \begin{bmatrix} 0 & 1 & 0 \\ 0 & 0 & 1 \\ 1-\bfv_1/s_{\bfv} & 1-\bfv_2/s_{\bfv} & 1-\bfv_3/s_{\bfv}\end{bmatrix}
\]
thus satisfy
\begin{equation}\label{eq:transp}
{\mathbf A}^{\textup{FS}}_{i,\bfv}\, \bfx =\transp{{\mathbf M}^{\textup{FS}}_{i}\, \bfx}, \mbox{ for all } i .
\end{equation}
for every $\bfx \in {\mathbb Z}^3$ such that $\langle \bfv, \bfx \rangle=0$.
We consider the matrix  norm induced  by  the  norm  $||\ ||_{\infty}$, that is,
$||{\mathbf A}||=\max _{i=1,2,3} \sum_{j=1}^{3} | a_{ij}|$.
The matrices ${\mathbf A}^{\textup{FS}}_{i,\bfv}$ are  nonnegative stochastic matrices.

Let  $(i_n)_n \in \{1,2,3\}^{\mathbb N}$ be the $\bfF$-expansion of $\bfv$.
One has
${\mathbf M}_n={\mathbf M}^{\textup{FS}}_{i_n}$ for all $n$.
One   sets  ${\mathbf A}_n={\mathbf A}^{\textup{FS}}_{i_n,\bfv^{(n-1)}}$, and
$\bfx ^{(n)}= \transp{{\mathbf M}_n} \cdots \transp{{\mathbf M}_1}\,\bfx$, for all $n$.
Note that the vectors  $\bfx^{(n)}$  take integer values.
Since  $\langle \bfv, \bfx \rangle=0$, we then have
\[
\langle \bfv^{(1)}, \bfx^{(1)} \rangle=\langle({\mathbf M}^{\textup{FS}}_{i_1,\bfv})^{-1} \, \bfv,
\transp{{\mathbf M}^{\textup{FS}}_{i_1,\bfv}}\, \bfx \rangle=\langle \bfv, \bfx \rangle= 0.
\]
More generally, one gets $\langle \bfv ^{(n)}, \bfx^{(n)} \rangle=0$, and (\ref{eq:transp}) extends to
\[
{\mathbf A}_{n+1}\,\bfx^{(n)}=\transp{{\mathbf M}_{n+1}}\,\bfx^{(n)}
\mbox{ and }
\bfx ^{(n)}={\mathbf A}_n\cdots  {\mathbf A}_1\bfx.
\]

The matrices   ${\mathbf A}_n\cdots  {\mathbf A}_1$  are stochastic matrices,
which yields that $|| \bfx^{(n)}||=|| {\mathbf A}^{(n)}\bfx||$ takes bounded  values.
Furthermore, the  vectors $\bfx^{(n)}$ take integer values.
There thus exist $k, \ell$ with $k < \ell$ such that $ \bfx^{(k)}= \bfx^{(\ell)}$.
Let us assume now $\bfx\neq 0$.
Then, $\bfx^{(k)}\neq 0$ (the matrices  ${\mathbf M}^{\textup{FS}}_{i,\bfv}$ are unimodular),
and $\bfx^{(k)}$ is  an eigenvector for ${\mathbf A}_{\ell}\cdots  {\mathbf A}_{k+1}$  for  the eigenvalue $1$.

Furthermore,  we assume that $\ell-k $ is large enough for $i_{k+1}\cdots i_{\ell}$ to
contain three times the letter $3$.  We use here the assumption that  $\bfv $ belongs to
$\Fthree$ together with  Lemma~\ref{lem:3infty}. One then checks that the matrix
${\mathbf A}_{\ell}\cdots {\mathbf A}_{k+1}$ is irreducible, that is,  for any  index $(i,j)$,
there exists  a power of the matrix for which the corresponding entry is positive. Indeed, as soon as a
matrix of type $3$ occurs at least  three times, one checks that the product of nonnative matrices
${\mathbf A}_{\ell}\cdots {\mathbf A}_{k+1}$ is irreducible.

By applying  the Perron-Frobenius theorem to the nonnegative stochastic matrix
${\mathbf A}_{\ell}\cdots  {\mathbf A}_{k+1}$,
one deduces  that  $\bfx^{(k)}$
is equal up to some multiplicative constant
to the vector  $\bfe_1+\bfe_2+\bfe_3$.
We deduce that  all the coordinates of  $\bfx^{(k)}$  are nonzero and have the same sign,
which   gives a  contradiction with $\langle \bfv ^{(k)}, \bfx^{(k)} \rangle=0$.
\end{proof}


\subsection{Substitutions and dual substitutions}
\label{subsec:sub}

Let $\mcA=\{1,2,3\}$ be a finite alphabet and $\mcA^\star$ be the set of finite words over $\mcA$.

\begin{defi}[Substitution]
A substitution over $\mcA$ is a  morphism of the free monoid $\mcA^\star$,
\emph{i.e.}, a function $\sigma : \mcA^\star \rightarrow \mcA^\star$
with
 $\sigma(uv) = \sigma(u)\sigma(v)$ for all words $u, v \in \mcA^\star$.
\end{defi}

Given a substitution $\sigma$ over $\mcA$,
the \emph{incidence matrix} $\Ms$ of $\sigma$ is the square matrix of size $3 \times 3$
defined by $\Ms = (m_{ij})$, where $m_{i,j}$ is the number of occurrences of the letter $i$ in $\sigma(j)$.
A substitution $\sigma$ is \emph{unimodular} if $\det \Ms = \pm 1$.

\begin{defi}[Dual substitution~\cite{AI01}]
\label{dfn::DS}
Let $\sigma : \{1,2,3\}^\star \longrightarrow \{1,2,3\}^\star$ be a unimodular substitution.
The \emph{dual substitution} $\EOSS$ is defined as
\[
\EOSS([\bfx, i]^\star) \ = \
\Msinv \bfx + \bigcup_{(p,j,s) \in \mcA^\star \times \mcA \times \mcA^\star \ : \ \sigma(j) = pis} [\Msinv \ell(s), j]^\star,
\]
where $\ell : w \mapsto (|w|_1,|w|_2,|w|_3) \in \bbZ^3$
is the \emph{Parikh map} counting the occurrences of each letter in a word $w$.
We extend the above definition to any union of unit faces: $\EOSS(P_1\cup P_2) = \EOSS(P_1) \cup \EOSS(P_2)$.
\end{defi}
Note that
$\EOS(\sigma \circ \sigma') = \EOS(\sigma') \circ \EOS(\sigma)$
for unimodular $\sigma$ and $\sigma'$ (see~\cite{AI01}).

\begin{prop}[\cite{AI01,Fer06}]
\label{prop::imgplane}
We have $\EOSS(\Gv) = \bfGa_{\transp{\Ms} \bfv}$
for every stepped plane $\Gv$ and unimodular substitution $\sigma$.
Furthermore, the images of two distinct faces of $\Gv$ have no common unit face.
\end{prop}

We now introduce the substitutions associated with the ordered fully subtractive algorithm,
which will be our main tool.  Let
\begin{equation*}
\begin{array}{ccccc}
\sFSi = \left\{\begin{array}{l}
1 \mapsto 1 \\ 2 \mapsto 21 \\3 \mapsto 31
\end{array}
\right. & \quad &
\sFSii = \left\{\begin{array}{l}
1 \mapsto 2 \\ 2 \mapsto 12 \\3 \mapsto 32
\end{array}
\right. & \quad &
\sFSiii = \left\{\begin{array}{l}
1 \mapsto 3 \\ 2 \mapsto 13 \\3 \mapsto 23.
\end{array}
\right.
\end{array}
\end{equation*}
The matrices occurring in the expansion of $\bfv$ according to the ordered fully subtractive algorithm
are the transposes of the matrices of incidence of the $\sFSx$,
that is, $\bfM_{\sFSx} = \transp{\MFSx}$
for $i \in \{1,2,3\}$.

We denote by $\SFSx$ the three dual substitutions $\EOS(\sFSx)$ for $i \in \{1,2,3\}$.
They can be represented as follows,
where the black dot respectively stands for the distinguished vector of a face and of  its image.
\[
\renewcommand{\tabcolsep}{1.5mm}
\renewcommand{\arraystretch}{1.2}
\centering
\SFSi : \left\{
    \begin{tabular}{rcl}%
    \myvcenter{%
    \begin{tikzpicture}
    [x={(-0.216506cm,-0.125000cm)}, y={(0.216506cm,-0.125000cm)}, z={(0.000000cm,0.250000cm)}]
    \fill[fill=facecolor, draw=black, shift={(0,0,0)}]
    (0, 0, 0) -- (0, 1, 0) -- (0, 1, 1) -- (0, 0, 1) -- cycle;
    \node[circle,fill=black,draw=black,minimum size=1.2mm,inner sep=0pt] at (0,0,0) {};
    \end{tikzpicture}}%
     & \myvcenter{$\mapsto$} &
    \myvcenter{%
    \begin{tikzpicture}
    [x={(-0.216506cm,-0.125000cm)}, y={(0.216506cm,-0.125000cm)}, z={(0.000000cm,0.250000cm)}]
    \fill[fill=facecolor, draw=black, shift={(0,0,0)}]
    (0, 0, 0) -- (0, 1, 0) -- (0, 1, 1) -- (0, 0, 1) -- cycle;
    \fill[fill=facecolor, draw=black, shift={(0,0,0)}]
    (0, 0, 0) -- (0, 0, 1) -- (1, 0, 1) -- (1, 0, 0) -- cycle;
    \fill[fill=facecolor, draw=black, shift={(0,0,0)}]
    (0, 0, 0) -- (1, 0, 0) -- (1, 1, 0) -- (0, 1, 0) -- cycle;
    \node[circle,fill=black,draw=black,minimum size=1.2mm,inner sep=0pt] at (0,0,0) {};
    \end{tikzpicture}} \\
    \myvcenter{%
    \begin{tikzpicture}
    [x={(-0.216506cm,-0.125000cm)}, y={(0.216506cm,-0.125000cm)}, z={(0.000000cm,0.250000cm)}]
    \fill[fill=facecolor, draw=black, shift={(0,0,0)}]
    (0, 0, 0) -- (0, 0, 1) -- (1, 0, 1) -- (1, 0, 0) -- cycle;
    \node[circle,fill=black,draw=black,minimum size=1.2mm,inner sep=0pt] at (0,0,0) {};
    \end{tikzpicture}}%
     & \myvcenter{$\mapsto$} &
    \myvcenter{%
    \begin{tikzpicture}
    [x={(-0.216506cm,-0.125000cm)}, y={(0.216506cm,-0.125000cm)}, z={(0.000000cm,0.250000cm)}]
    \draw[thick, densely dotted] (0,0,0) -- (1,0,0);
    \fill[fill=facecolor, draw=black, shift={(1,0,0)}]
    (0, 0, 0) -- (0, 0, 1) -- (1, 0, 1) -- (1, 0, 0) -- cycle;
    \node[circle,fill=black,draw=black,minimum size=1.2mm,inner sep=0pt] at (0,0,0) {};
    \end{tikzpicture}} \\
    \myvcenter{%
    \begin{tikzpicture}
    [x={(-0.216506cm,-0.125000cm)}, y={(0.216506cm,-0.125000cm)}, z={(0.000000cm,0.250000cm)}]
    \fill[fill=facecolor, draw=black, shift={(0,0,0)}]
    (0, 0, 0) -- (1, 0, 0) -- (1, 1, 0) -- (0, 1, 0) -- cycle;
    \node[circle,fill=black,draw=black,minimum size=1.2mm,inner sep=0pt] at (0,0,0) {};
    \end{tikzpicture}}%
     & \myvcenter{$\mapsto$} &
    \myvcenter{%
    \begin{tikzpicture}
    [x={(-0.216506cm,-0.125000cm)}, y={(0.216506cm,-0.125000cm)}, z={(0.000000cm,0.250000cm)}]
    \draw[thick, densely dotted] (0,0,0) -- (1,0,0);
    \fill[fill=facecolor, draw=black, shift={(1,0,0)}]
    (0, 0, 0) -- (1, 0, 0) -- (1, 1, 0) -- (0, 1, 0) -- cycle;
    \node[circle,fill=black,draw=black,minimum size=1.2mm,inner sep=0pt] at (0,0,0) {};
    \end{tikzpicture}}
    \end{tabular}
\right.
\quad
\SFSii : \left\{
    \begin{tabular}{rcl}%
    \myvcenter{%
    \begin{tikzpicture}
    [x={(-0.216506cm,-0.125000cm)}, y={(0.216506cm,-0.125000cm)}, z={(0.000000cm,0.250000cm)}]
    \fill[fill=facecolor, draw=black, shift={(0,0,0)}]
    (0, 0, 0) -- (0, 1, 0) -- (0, 1, 1) -- (0, 0, 1) -- cycle;
    \node[circle,fill=black,draw=black,minimum size=1.2mm,inner sep=0pt] at (0,0,0) {};
    \end{tikzpicture}}%
     & \myvcenter{$\mapsto$} &
    \myvcenter{%
    \begin{tikzpicture}
    [x={(-0.216506cm,-0.125000cm)}, y={(0.216506cm,-0.125000cm)}, z={(0.000000cm,0.250000cm)}]
    \draw[thick, densely dotted] (0,0,0) -- (1,0,0);
    \fill[fill=facecolor, draw=black, shift={(1,0,0)}]
    (0, 0, 0) -- (1, 0, 0) -- (1, 0, 1) -- (0, 0, 1) -- cycle;
    \node[circle,fill=black,draw=black,minimum size=1.2mm,inner sep=0pt] at (0,0,0) {};
    \end{tikzpicture}} \\
    \myvcenter{%
    \begin{tikzpicture}
    [x={(-0.216506cm,-0.125000cm)}, y={(0.216506cm,-0.125000cm)}, z={(0.000000cm,0.250000cm)}]
    \fill[fill=facecolor, draw=black, shift={(0,0,0)}]
    (0, 0, 0) -- (0, 0, 1) -- (1, 0, 1) -- (1, 0, 0) -- cycle;
    \node[circle,fill=black,draw=black,minimum size=1.2mm,inner sep=0pt] at (0,0,0) {};
    \end{tikzpicture}}%
     & \myvcenter{$\mapsto$} &
    \myvcenter{%
    \begin{tikzpicture}
    [x={(-0.216506cm,-0.125000cm)}, y={(0.216506cm,-0.125000cm)}, z={(0.000000cm,0.250000cm)}]
    \fill[fill=facecolor, draw=black, shift={(0,0,0)}]
    (0, 0, 0) -- (0, 1, 0) -- (0, 1, 1) -- (0, 0, 1) -- cycle;
    \fill[fill=facecolor, draw=black, shift={(0,0,0)}]
    (0, 0, 0) -- (0, 0, 1) -- (1, 0, 1) -- (1, 0, 0) -- cycle;
    \fill[fill=facecolor, draw=black, shift={(0,0,0)}]
    (0, 0, 0) -- (1, 0, 0) -- (1, 1, 0) -- (0, 1, 0) -- cycle;
    \node[circle,fill=black,draw=black,minimum size=1.2mm,inner sep=0pt] at (0,0,0) {};
    \end{tikzpicture}} \\
    \myvcenter{%
    \begin{tikzpicture}
    [x={(-0.216506cm,-0.125000cm)}, y={(0.216506cm,-0.125000cm)}, z={(0.000000cm,0.250000cm)}]
    \fill[fill=facecolor, draw=black, shift={(0,0,0)}]
    (0, 0, 0) -- (1, 0, 0) -- (1, 1, 0) -- (0, 1, 0) -- cycle;
    \node[circle,fill=black,draw=black,minimum size=1.2mm,inner sep=0pt] at (0,0,0) {};
    \end{tikzpicture}}%
     & \myvcenter{$\mapsto$} &
    \myvcenter{%
    \begin{tikzpicture}
    [x={(-0.216506cm,-0.125000cm)}, y={(0.216506cm,-0.125000cm)}, z={(0.000000cm,0.250000cm)}]
    \draw[thick, densely dotted] (0,0,0) -- (1,0,0);
    \fill[fill=facecolor, draw=black, shift={(1,0,0)}]
    (0, 0, 0) -- (1, 0, 0) -- (1, 1, 0) -- (0, 1, 0) -- cycle;
    \node[circle,fill=black,draw=black,minimum size=1.2mm,inner sep=0pt] at (0,0,0) {};
    \end{tikzpicture}}
    \end{tabular}
\right.
\quad
\SFSiii \ : \ \left\{
    \begin{tabular}{rcl}%
    \myvcenter{%
    \begin{tikzpicture}
    [x={(-0.216506cm,-0.125000cm)}, y={(0.216506cm,-0.125000cm)}, z={(0.000000cm,0.250000cm)}]
    \fill[fill=facecolor, draw=black, shift={(0,0,0)}]
    (0, 0, 0) -- (0, 1, 0) -- (0, 1, 1) -- (0, 0, 1) -- cycle;
    \node[circle,fill=black,draw=black,minimum size=1.2mm,inner sep=0pt] at (0,0,0) {};
    \end{tikzpicture}}%
     & \myvcenter{$\mapsto$} &
    \myvcenter{%
    \begin{tikzpicture}
    [x={(-0.216506cm,-0.125000cm)}, y={(0.216506cm,-0.125000cm)}, z={(0.000000cm,0.250000cm)}]
    \draw[thick, densely dotted] (0,0,0) -- (1,0,0);
    \fill[fill=facecolor, draw=black, shift={(1,0,0)}]
    (0, 0, 0) -- (1, 0, 0) -- (1, 0, 1) -- (0, 0, 1) -- cycle;
    \node[circle,fill=black,draw=black,minimum size=1.2mm,inner sep=0pt] at (0,0,0) {};
    \end{tikzpicture}} \\
    \myvcenter{%
    \begin{tikzpicture}
    [x={(-0.216506cm,-0.125000cm)}, y={(0.216506cm,-0.125000cm)}, z={(0.000000cm,0.250000cm)}]
    \fill[fill=facecolor, draw=black, shift={(0,0,0)}]
    (0, 0, 0) -- (0, 0, 1) -- (1, 0, 1) -- (1, 0, 0) -- cycle;
    \node[circle,fill=black,draw=black,minimum size=1.2mm,inner sep=0pt] at (0,0,0) {};
    \end{tikzpicture}}%
     & \myvcenter{$\mapsto$} &
    \myvcenter{%
    \begin{tikzpicture}
    [x={(-0.216506cm,-0.125000cm)}, y={(0.216506cm,-0.125000cm)}, z={(0.000000cm,0.250000cm)}]
    \draw[thick, densely dotted] (0,0,0) -- (1,0,0);
    \fill[fill=facecolor, draw=black, shift={(1,0,0)}]
    (0, 0, 0) -- (1, 0, 0) -- (1, 1, 0) -- (0, 1, 0) -- cycle;
    \node[circle,fill=black,draw=black,minimum size=1.2mm,inner sep=0pt] at (0,0,0) {};
    \end{tikzpicture}} \\
    \myvcenter{%
    \begin{tikzpicture}
    [x={(-0.216506cm,-0.125000cm)}, y={(0.216506cm,-0.125000cm)}, z={(0.000000cm,0.250000cm)}]
    \fill[fill=facecolor, draw=black, shift={(0,0,0)}]
    (0, 0, 0) -- (1, 0, 0) -- (1, 1, 0) -- (0, 1, 0) -- cycle;
    \node[circle,fill=black,draw=black,minimum size=1.2mm,inner sep=0pt] at (0,0,0) {};
    \end{tikzpicture}}%
     & \myvcenter{$\mapsto$} &
    \myvcenter{%
    \begin{tikzpicture}
    [x={(-0.216506cm,-0.125000cm)}, y={(0.216506cm,-0.125000cm)}, z={(0.000000cm,0.250000cm)}]
    \fill[fill=facecolor, draw=black, shift={(0,0,0)}]
    (0, 0, 0) -- (0, 1, 0) -- (0, 1, 1) -- (0, 0, 1) -- cycle;
    \fill[fill=facecolor, draw=black, shift={(0,0,0)}]
    (0, 0, 0) -- (0, 0, 1) -- (1, 0, 1) -- (1, 0, 0) -- cycle;
    \fill[fill=facecolor, draw=black, shift={(0,0,0)}]
    (0, 0, 0) -- (1, 0, 0) -- (1, 1, 0) -- (0, 1, 0) -- cycle;
    \node[circle,fill=black,draw=black,minimum size=1.2mm,inner sep=0pt] at (0,0,0) {};
    \end{tikzpicture}}
    \end{tabular}
\right.
\]

\begin{lemm}[Preimages of unit faces by $\Sigma_i$]
\label{lemm:preimages}
Let $\bfx =\svect{x}{y}{z} \in \bbZ^3$.
We have
\[
\Sigma_1^{-1}([\bfx, i]^\star) \ = \
\left\{
\begin{array}{ll}
  \big[ \svect{x+y+z}{y}{z}, 1 \big]^\star
    & \textup{if } i = 1 \\
  \big[ \svect{x+y+z}{y}{z}, 1 \big]^\star \cup \big[ \svect{x+y+z-1}{y}{z}, i \big]^\star
    & \textup{if } i = 2,3
\end{array}
\right.,
\]
\[
\Sigma_2^{-1}([\bfx, i]^\star) \ = \
\left\{
\begin{array}{ll}
  \big[ \svect{y}{x+y+z}{z}, 2 \big]^\star
    & \textup{if } i = 1 \\
  \big[ \svect{y}{x+y+z}{z}, 2 \big]^\star \cup \big[ \svect{y}{x+y+z-1}{z}, i \big]^\star
    & \textup{if } i = 2,3
\end{array}
\right.,
\]
\[
\Sigma_3^{-1}([\bfx, i]^\star) \ = \
\left\{
\begin{array}{ll}
  \big[ \svect{y}{z}{x+y+z}, 3 \big]^\star
    & \textup{if } i = 1 \\
  \big[ \svect{y}{z}{x+y+z}, 3 \big]^\star \cup \big[ \svect{y}{z}{x+y+z-1}, i \big]^\star
    & \textup{if } i = 2,3
\end{array}
\right..
\]
\end{lemm}


\section{Generating discrete planes by translations}
\label{sec:trans}


We introduce in this section  a sequence of  $2$-connected subsets  $(\bfT_n)_{n \in \bbN}$
of $\bbZ^3$ that are contained in  arithmetic discrete planes with connecting thickness.
\begin{defi}[Generation by translations]
Let $\bfv \in \Fthree$ be a vector with $\bfF$-expansion $(i_n)_{n \in \bbN}$.
Denote by $\bfM_n$ the matrix $\smash{\transp\bfM_{\sFS_{i_n}}}$, for all $n \geq 1$.
We define the sequence $(\bfT_n)_{n \in \bbN}$ of subsets of $\bbZ^3$  as follows for all $n \geq 0$:
\[
\bfT_0 = \{\mathbf 0\}, \quad
\bfT_1 = \{\mathbf 0,\bfe_1\}, \quad
\bfT_{n+1} = \bfT_n
    \cup \left( \bfT_n + \transp{(\bfM_1 \dots \bfM_n)}^{-1} \cdot \bfe_1 \right).
\]
\end{defi}

Note that the second initial condition $\bfT_1 = \{\mathbf 0,\bfe_1\}$ is consistent with
the usual convention that an empty product of matrices is equal to  the identity matrix.

\begin{prop}
\label{prop::T_n}
Let $\bfv \in \Fthree$. We have $\cup^\infty_{n=0} \bfT_n \subseteq \frakP(\bfv,\Omega(\bfv))$.
\end{prop}

\begin{proof}
Let us prove that for all $n \in \bbN$ and $\bfx \in \bfT_n$, we have
\begin{equation}
\label{eq::assertion}
\langle \bfx,\bfv \rangle < \sum^n_{i=0}{\bfv_1^{(i)}}.
\end{equation}
The case $n \in \{0,1\}$ can be checked easily.
Assume that Eq. (\ref{eq::assertion}) holds for some $n \geq 1$, and let
$\bfx \in \bfT_{n+1} = \bfT_n \cup \left( \bfT_n + \transp{(\bfM_1 \cdots \bfM_n )}^{-1} \cdot \bfe_1 \right)$.
Then, two cases can occur.
\begin{enumerate}
\item
If $\bfx \in \bfT_n$
then $\langle \bfx , \bfv \rangle < \displaystyle \sum^n_{i=0}{\bfv_1^{(i)}} < \sum^{n+1}_{i=0}{\bfv_1^{(i)}}$.
\item
If $\bfx \in \bfT_n + \transp{(\bfM_1 \cdots \bfM_n )}^{-1} \cdot \bfe_1$,
then let $\bfy \in \bfT_n$  be such that $\bfx = \bfy + \transp{(\bfM_1 \cdots \bfM_n )}^{-1} \cdot \bfe_1$.
We have
\begin{eqnarray*}
\langle \bfx, \bfv \rangle
     & = & \langle \bfy + \transp{(\bfM_1 \cdots \bfM_n )}^{-1} \cdot \bfe_1, \bfv \rangle \\
     & = & \langle \bfy,\bfv \rangle  + \langle \transp{(\bfM_1 \cdots \bfM_n )}^{-1} \cdot \bfe_1, \bfv \rangle \\
     & = & \langle \bfy,\bfv \rangle  + \langle \transp{(\bfM_1 \cdots \bfM_n )}^{-1} \cdot \bfe_1,
        \bfM_1 \dots \bfM_n \cdot \bfv^{(n)} \rangle \\
     & = & \langle \bfy,\bfv \rangle  + \langle \bfe_1, \bfv^{(n)} \rangle \\
     & = & \langle \bfy,\bfv \rangle  + \bfv^{(n)}_1 < \displaystyle \sum^n_{i=0}{\bfv_1^{(i)}} + \bfv^{(n)}_1 =  \sum^{n+1}_{i=0}{\bfv_1^{(i)}}.
\end{eqnarray*}
\end{enumerate}
\end{proof}

\begin{prop}
\label{prop::Tn_conn}
Let $\bfv \in \Fthree$.
For all $n \in \bbN$, the set $\bfT_n$ is $2$-connected.
\end{prop}

\begin{proof}
With the same arguments as in proof of Proposition~\ref{prop::T_n},
and by using Proposition~\ref{prop:F3},
we first get by induction that, for all $n \geq 1$:
\begin{equation*}
\bfT_n = \left\{ \bfx \in \bbZ^3 :
    \langle \bfx,\bfv \rangle = \sum^{n-1}_{i=0}{\varepsilon_i \bfv^{(i)}_1}
    \text{with $\varepsilon_i \in \{0,1\}$ for all $i$} \right\}.
\end{equation*}
Note that $\dim_\bbQ(\bfv_1,\bfv_2,\bfv_3) = 3$ implies that for all $x,y \in \bbZ^3$,
$\langle x, v\rangle = \langle y , v \rangle \iff x = y$.

Now, for all $n \in \bbN$, let $\bfx_n \in \bfT_n$
be such that $\langle \bfx_n, \bfv \rangle = \sum^{n-1}_{i=0}{ \bfv^{(i)}_1}$
(we set $\bfx_0= \mathbf 0$).
Let us prove  by induction the following property:
for all $n \geq 1,$
there exists $i_n \in \{1,2,3\}$ such that $\bfx_n - \bfe_{i_n} \in \bfT_{n-1}$.
This property implies that $\bfx_{n}$ is $2$-adjacent to $ \bfT_{n-1}$,
which implies the $2$-connectedness of $\bfT_n$.

The induction property is true for $n=1$ with $\bfx_1=\bfe_1$.
Let us now assume that the induction hypothesis holds for $n \geq 1$.
Let $u_1\cdots u_n\in\{1,2,3\}^{\bbN^\star}$ be such that $\MFS_{u_1} \cdots \MFS_{u_n} \bfv^{(n)} = \bfv$.
We have $\langle \bfx_{n+1}, \bfv \rangle = \langle \bfx_n, \bfv \rangle+ \bfv^{(n)}_1$,
and by definition of the fully subtractive algorithm $\bfF$:
\begin{equation*}
\bfv^{(n)}_1=
    \left\{
    \begin{array}{ll}
        \bfv^{(n-1)}_1,                  & \text{if $u_{n}=1$} \\
        \bfv^{(n-1)}_2 - \bfv^{(n-1)}_1, & \text{if $u_{n} \in \{2,3\}$.}
    \end{array}
    \right.
\end{equation*}
We distinguish several cases according to the values taken by $u_1 \cdots u_n$.

\noindent
\textbf{Case 1.}
If $u_n=1$, then, $\langle \bfx_{n+1} , \bfv \rangle = \langle \bfx_{n} , \bfv \rangle+ \bfv^{(n-1)}_1$, and
\begin{equation*}
\langle \bfx_{n+1} - \bfe_{i_n}, \bfv \rangle
    = \langle \underbrace{\bfx_{n} - \bfe_{i_n}}_{\in \bfT_{n-1}}, \bfv \rangle+ \bfv^{(n-1)}_1
    = \sum^{n-2}_{i=1}{\varepsilon_i \bfv^{(i)}_1} + \bfv^{(n-1)}_1,
\end{equation*}
where $\varepsilon_i \in \{0,1\}$ for $1 \leq i \leq n-2$,
which implies that $\bfx_{n+1} - \bfe_{i_n} \in \bfT_n$,
so taking $i_{n+1} = i_n$ yields the desired result.

%

\noindent
\textbf{Case 2.}
If $u_n \in \{2,3\}$ and $u_1 \cdots u_{n-1} = 1^k $, then
\begin{eqnarray*}
\langle \bfx_{n+1}, \bfv \rangle
    & = & \langle \bfx_n, \bfv \rangle + \bfv^{(n-1)}_2 - \bfv^{(n-1)}_1
        = \langle \bfx_{n-1}, \bfv \rangle + \bfv^{(n-1)}_2 \\
    & = & \langle \bfx_{n-2}, \bfv \rangle + \bfv^{(n-2)}_2
        = \cdots
        = \langle \bfx_{n-1-k}, \bfv \rangle+ \bfv^{(n-1-k)}_2
        = \bfv^{(0)}_2,
\end{eqnarray*}
which implies that $\bfx_{n+1} - \bfe_2 \in \bfT_n$.

\noindent
\textbf{Case 3.}
If $u_n \in \{2,3\}$ and $u_1 \cdots u_{n-1} = \cdots21^k $ with $0 \leq k \leq n-2$, then
\begin{eqnarray*}
\langle \bfx_{n+1} , \bfv \rangle
    & = & \langle \bfx_n, \bfv \rangle+ \bfv^{(n-1)}_2 - \bfv^{(n-1)}_1
        = \langle \bfx_{n-1}, \bfv \rangle+ \bfv^{(n-1)}_2 \\
    & = & \langle \bfx_{n-1-k}, \bfv \rangle+ \bfv^{(n-1-k)}_2
        = \langle \bfx_{n-1-k}, \bfv \rangle+ \bfv^{(n-2-k)}_1,
\end{eqnarray*}
so $\bfx_{n+1} - \bfe_{i_{n-1-k}} \in \bfT_{n-1-k} \subseteq \bfT_n$.



\noindent
\textbf{Case 4.}
If $u_n \in \{2,3\}$ and $u_1 \cdots u_{n-1} = w31^k$ with $w \in \{1,2\}^\ell$ and $k \geq 0$, then
\begin{eqnarray*}
\langle \bfx_{n+1}, \bfv \rangle
    & = & \langle \bfx_{n-1-k}, \bfv \rangle+ \bfv^{(n-1-k)}_2
        = \langle \bfx_{n-2-k}, \bfv \rangle+ \bfv^{(n-2-k)}_3 \\
    & = & \langle \bfx_{n-2-k-\ell}, \bfv \rangle+ \bfv^{(n-2-k-\ell)}_3 = \bfv^{(0)}_3,
\end{eqnarray*}
so $ \bfx_{n+1} - \bfe_3 \in \bfT_n$.


\noindent
\textbf{Case 5.}
If $u_n \in \{2,3\}$ and $u_1 \cdots u_{n-1} = \cdots 3w31^k$ with $w \in \{1,2\}^\ell$, $k \geq 0$, then
\begin{eqnarray*}
\langle \bfx_{n+1}, \bfv \rangle
    & = & \langle \bfx_{n-2-k-\ell}, \bfv \rangle+ \bfv^{(n-2-k-\ell)}_3 \\
    & = & \langle \bfx_{n-2-k-\ell}, \bfv \rangle+ \bfv^{(n-3-k-\ell)}_1
\end{eqnarray*}
so $ \bfx_{n+1} - \bfe_{i_{n-2-k-\ell}} \in \bfT_{n-2-k-\ell} \subseteq \bfT_n$.
\end{proof}





%
\section{Generation of naive planes with dual substitutions}
\label{sec:gen}
The aim of this section is to introduce a second sequence   $(\bfP_n)_{n\in {\mathbb N}} $
of patterns in ${\mathbb Z}^3$.  These patterns  are sub-patterns of the $\bfT_n$ and they are obtained
by applying generalized substitutions  according to the fully subtractive algorithm.
We will prove that the patterns $\bfP_n$
cover naive arithmetic discrete planes $ \frakP(\bfv, \|\bfv\|_\infty)$,
by showing that iterations of dual substitutions yield
concentric annuli (see Definition~\ref{defi:annulus}) with increasing radius.
The main result of this section is the following.
\begin{prop}
\label{prop:pn_union}
If $\bfv \in \Fthree$, then $\bigcup^\infty_{n=0} {\bfP_n} = \frakP(\bfv, \|\bfv\|_\infty)$.
\end{prop}

The proof will be given at the end of Section~\ref{subsec:annulus}.  The
remaining of this section is devoted to the development of specific tools used
in this proof.  Such tools have also been used in~\cite{BJS} to study other
multidimensional continued fraction algorithms.

\subsection{Definition of the patterns $\bfP_n$}

Let $\mcU$ stand for the lower half unit cube at the origin, that is,
$\mcU = [\mathbf 0,1]^\star \cup [\mathbf 0,2]^\star \cup [\mathbf 0,3]^\star = \myvcenter{%
    \begin{tikzpicture}
    [x={(-0.173205cm,-0.100000cm)}, y={(0.173205cm,-0.100000cm)}, z={(0.000000cm,0.200000cm)}]
    \fill[fill=facecolor, draw=black, shift={(0,0,0)}]
    (0, 0, 0) -- (0, 1, 0) -- (0, 1, 1) -- (0, 0, 1) -- cycle;
    \fill[fill=facecolor, draw=black, shift={(0,0,0)}]
    (0, 0, 0) -- (0, 0, 1) -- (1, 0, 1) -- (1, 0, 0) -- cycle;
    \fill[fill=facecolor, draw=black, shift={(0,0,0)}]
    (0, 0, 0) -- (1, 0, 0) -- (1, 1, 0) -- (0, 1, 0) -- cycle;
    \node[circle,fill=black,draw=black,minimum size=1mm,inner sep=0pt] at (0,0,0) {};
    \end{tikzpicture}}$.

\begin{defi}[Patterns $\bfP_n$]
Let $\bfv \in \Fthree$ be a vector with $\bfF$-expansion $(i_n)_{n \in \bbN}$.
We~define:
\begin{itemize}
\item
$P_n = \SFS_{i_1} \cdots \SFS_{i_n}(\mcU)$ for $n \geq 1$ and $P_0 = \mcU$;
\item
$\bfP_n = \{\bfx : [\bfx,i]^\star \in P_n\}$ for $n \geq 0$.
\end{itemize}
\end{defi}

\begin{prop}
\label{prop:pn_tn}
Let $\bfv \in \Fthree$.
For every $n \in \bbN$, we have  $\bfP_n \ \subseteq \ \bfT_n$.
\end{prop}

\begin{proof}
We first remark that
$\EOS(\sFSx)(\mcU) = \mcU \cup [\bfe_1,2]^\star \cup [\bfe_1,3]^\star = \myvcenter{%
    \begin{tikzpicture}
    [x={(-0.173205cm,-0.100000cm)}, y={(0.173205cm,-0.100000cm)}, z={(0.000000cm,0.200000cm)}]
    \fill[fill=facecolor, draw=black, shift={(0,0,0)}]
    (0, 0, 0) -- (0, 1, 0) -- (0, 1, 1) -- (0, 0, 1) -- cycle;
    \fill[fill=facecolor, draw=black, shift={(0,0,0)}]
    (0, 0, 0) -- (0, 0, 1) -- (1, 0, 1) -- (1, 0, 0) -- cycle;
    \fill[fill=facecolor, draw=black, shift={(0,0,0)}]
    (0, 0, 0) -- (1, 0, 0) -- (1, 1, 0) -- (0, 1, 0) -- cycle;
    \fill[fill=facecolor, draw=black, shift={(0,0,0)}]
    (1, 0, 0) -- (2, 0, 0) -- (2, 1, 0) -- (1, 1, 0) -- cycle;
    \fill[fill=facecolor, draw=black, shift={(1,0,0)}]
    (0, 0, 0) -- (0, 0, 1) -- (1, 0, 1) -- (1, 0, 0) -- cycle;
    \node[circle,fill=black,draw=black,minimum size=1mm,inner sep=0pt] at (0,0,0) {};
    \end{tikzpicture}}$ for all $i \in \{1,2,3\}$.
For $n \in \bbN$, we have
\begin{align*}
P_{n+1} &= \EOS(\sigma_{n+1} \circ \dots \circ \sigma_1)(\mcU) \\
        &= \EOS(\sigma_n \circ \dots \circ \sigma_1) \circ \EOS(\sigma_{n+1})  (\mcU) \\
        & = P_n \cup  \EOS(\sigma_n \circ \dots \circ \sigma_1) ([\bfe_1,2]^\star \cup [\bfe_1,3]^\star),
\end{align*}
which implies $P_n \subseteq P_{n+1} $.
Since $[\bfe_1,2]^\star \cup [\bfe_1,3]^\star \subseteq \bfe_1+ \mcU$, we have
$P_{n+1} \subseteq P_n \cup  \EOS(\sigma_n \circ \dots \circ \sigma_1)(\bfe_1+\mcU)$.
By Definition~\ref{dfn::DS}, we then have
\begin{align*}
\EOS(\sigma_n \circ \cdots \circ \sigma_1)( \bfe_1+\mcU )
    & = {\bfM}^{-1}_{\sigma_n \circ \cdots \circ \sigma_1} \cdot \bfe_1 +\EOS(\sigma_n \circ \cdots \circ \sigma_1)(\mcU)\\
    & = (\transp{\bfM}_n \cdots \transp{\bfM}_1 )^{-1} \cdot \bfe_1+ P_n \\
    & = \transp({\bfM_1 \cdots {}\bfM_n)}^{-1} \cdot \bfe_1+ P_n ,
\end{align*}
where $\bfM_n = \smash{\transp\bfM_{\sFS_{i_n}}}$, which proves that
$P_n \ \subseteq \ P_{n+1}
     \ \subseteq \ P_n \cup (P_n + \transp{(\bfM_1 \cdots \bfM_n)}^{-1} \cdot \bfe_1)$.
The result now follows by induction.
\end{proof}
As a direct consequence of Propositions~\ref{prop:pn_tn} and~\ref{prop:pn_union}, we obtain that
if $\bfv \in \Fthree$, then $\frakP(\bfv, \|\bfv\|_\infty) \subseteq \bigcup^\infty_{n=0}\bfT_n$,
\emph{i.e.}, the naive plane of normal vector $\bfv$ is included in $\bigcup^\infty_{n=0}\bfT_n$.

\subsection{Covering properties and annuli}\label{subsec:cov}
A \emph{pattern} is a union of unit faces.
In the rest of this section we will consider some sets of  connected patterns
($\mcL$, $\Ledge$ and $\LFS$) that will be needed in order to define (strong) coverings.
The patterns contained in these sets are considered up to translation only,
as it is all that matters for the definitions below (see
Figure~\ref{fig:lcover}).

\begin{defi}[$\mcL$-cover]
\label{defi:cover}
Let $\mcL$ be a set of patterns.
A pattern $P$ is \emph{$\mcL$-covered} if for all faces $e, f \in P$,
there exist $Q_1, \ldots, Q_n \in \mcL$ such that:
\begin{enumerate}
  \item $e \in Q_1$ and $f \in Q_n$;
  \item $Q_k \cap Q_{k+1}$ contains at least one face,
    for all $k \in \{1, \ldots, n-1\}$;
  \item $Q_k \subseteq P$ for all $k \in \{1, \ldots, n\}$.
\end{enumerate}
\end{defi}

\begin{lemm}[\cite{IO93}]
\label{lemm:coverprop}
Let $P$ be an $\mcL$-covered pattern,
$\Sigma$  a dual substitution
and $\mcL$ a set of patterns such that
$\Sigma(Q)$ is $\mcL$-covered for all $Q \in \mcL$.
Then $\Sigma(P)$ is $\mcL$-covered.
\end{lemm}

We will need \emph{strong} coverings to ensure that the image of an annulus is an annulus.
We denote by $\Ledge$  the set of all the twelve edge-connected two-face patterns (up to translation).

\begin{defi}[Strong $\mcL$-cover]
Let $\mcL$ be a set of edge-connected patterns.
A pattern $P$ is \emph{strongly $\mcL$-covered} if
\begin{enumerate}
\item $P$ is $\mcL$-covered;
\item for every pattern $X \in \Ledge$ such that $X \subseteq P$,
    there exists a pattern $Y \in \mcL$ such that $X \subseteq Y \subseteq P$.
\end{enumerate}
\end{defi}

The intuitive idea behind the notion of strong $\mcL$-covering is that every occurrence
of a pattern of $\Ledge$ in $P$ is required to be ``completed within $P$'' by a pattern of $\mcL$.

\begin{defi}[Annulus]
\label{defi:annulus}
Let $\mcL$ be a set of edge-connected patterns and $\bfGa$ be a stepped plane.
An \emph{$\mcL$-annulus} of a pattern $P \subseteq \bfGa$
is a pattern $A \subseteq \bfGa$ such that:
\begin{enumerate}
  \item $P$, $A \cup P$ and $\bfGa \setminus (A \cup P)$ are $\mcL$-covered;
    \label{defi:anneauprop1}
  \item $A$ is strongly $\mcL$-covered;
    \label{defi:anneauprop2}
  \item $A$ and $P$ have no face in common;
    \label{defi:anneauprop3}
  \item $P \cap \overline{\bfGa \setminus (P \cup A)} = \varnothing$.
    \label{defi:anneauprop4}
\end{enumerate}
\end{defi}
The notation $\overline{\bfGa \setminus (P \cup A)}$ stands for the
topological closure of $\bfGa \setminus (P \cup A)$.
Conditions~\ref{defi:anneauprop1} and~\ref{defi:anneauprop2} are
combinatorial properties that we will use in the proof of Lemma~\ref{lem:annulus_induction}
in order to prove that the image of an $\LFS$-annulus by a $\SFSx$ is an $\LFS$-annulus.
Conditions~\ref{defi:anneauprop3} and~\ref{defi:anneauprop4} are
properties of topological nature that we want annuli to satisfy.

\begin{figure}[ht]
\[
  \begin{array}{ccc}
    \begin{array}{c}
      \mcL = \left\{

        \begin{array}{c}
          \begin{tikzpicture}
            [x={(-0.173205cm,-0.100000cm)}, y={(0.173205cm,-0.100000cm)}, z={(0.000000cm,0.200000cm)}]
            \def\loza#1#2#3#4#5#6{
              \fill[fill=facecolor, draw=black, shift={(#1, #2, #3)}]
              (0, 0, 0) -- (0, 1, 0) -- (0, 1, 1) -- (0, 0, 1) -- cycle;
            }
            \def\lozb#1#2#3#4#5#6{
              \fill[fill=facecolor, draw=black, shift={(#1, #2, #3)}]
              (0, 0, 0) -- (0, 0, 1) -- (1, 0, 1) -- (1, 0, 0) -- cycle;
            }
            \def\lozc#1#2#3#4#5#6{
              \fill[fill=facecolor, draw=black, shift={(#1, #2, #3)}]
              (0, 0, 0) -- (1, 0, 0) -- (1, 1, 0) -- (0, 1, 0) -- cycle;
            }
            \loza{0}{0}{0}{0.50}{0.50}{0.50}
            \lozb{0}{0}{0}{0.50}{0.50}{0.50}
            \lozc{0}{0}{0}{0.50}{0.50}{0.50}
            \lozb{-1}{1}{0}{0.50}{0.50}{0.50}
            \loza{-1}{1}{0}{0.50}{0.50}{0.50}
            \lozc{-1}{1}{0}{0.50}{0.50}{0.50}
            \lozc{0}{1}{0}{0.50}{0.50}{0.50}
            \lozc{0}{1}{2}{0.50}{0.50}{0.50}
          \end{tikzpicture},
        \end{array}

        \begin{array}{c}
          \begin{tikzpicture}
            [x={(-0.173205cm,-0.100000cm)}, y={(0.173205cm,-0.100000cm)}, z={(0.000000cm,0.200000cm)}]
            \def\loza#1#2#3#4#5#6{
              \fill[fill=facecolor, draw=black, shift={(#1, #2, #3)}]
              (0, 0, 0) -- (0, 1, 0) -- (0, 1, 1) -- (0, 0, 1) -- cycle;
            }
            \def\lozb#1#2#3#4#5#6{
              \fill[fill=facecolor, draw=black, shift={(#1, #2, #3)}]
              (0, 0, 0) -- (0, 0, 1) -- (1, 0, 1) -- (1, 0, 0) -- cycle;
            }
            \def\lozc#1#2#3#4#5#6{
              \fill[fill=facecolor, draw=black, shift={(#1, #2, #3)}]
              (0, 0, 0) -- (1, 0, 0) -- (1, 1, 0) -- (0, 1, 0) -- cycle;
            }
            \loza{0}{0}{0}{0.50}{0.50}{0.50}
            \lozc{0}{0}{0}{0.50}{0.50}{0.50}
            \lozc{1}{0}{0}{0.50}{0.50}{0.50}
            \lozc{1}{-1}{0}{0.50}{0.50}{0.50}
            \lozc{1}{-2}{0}{0.50}{0.50}{0.50}
          \end{tikzpicture},
        \end{array}

        \begin{array}{c}
          \begin{tikzpicture}
            [x={(-0.173205cm,-0.100000cm)}, y={(0.173205cm,-0.100000cm)}, z={(0.000000cm,0.200000cm)}]
            \def\loza#1#2#3#4#5#6{
              \fill[fill=facecolor, draw=black, shift={(#1, #2, #3)}]
              (0, 0, 0) -- (0, 1, 0) -- (0, 1, 1) -- (0, 0, 1) -- cycle;
            }
            \def\lozb#1#2#3#4#5#6{
              \fill[fill=facecolor, draw=black, shift={(#1, #2, #3)}]
              (0, 0, 0) -- (0, 0, 1) -- (1, 0, 1) -- (1, 0, 0) -- cycle;
            }
            \def\lozc#1#2#3#4#5#6{
              \fill[fill=facecolor, draw=black, shift={(#1, #2, #3)}]
              (0, 0, 0) -- (1, 0, 0) -- (1, 1, 0) -- (0, 1, 0) -- cycle;
            }
            \loza{0}{0}{0}{0.50}{0.50}{0.50}
            \loza{0}{-1}{0}{0.50}{0.50}{0.50}
            \lozc{-1}{0}{1}{0.50}{0.50}{0.50}
            \lozc{-1}{-1}{1}{0.50}{0.50}{0.50}
            \lozb{-1}{1}{0}{0.50}{0.50}{0.50}

          \end{tikzpicture}
        \end{array}
      \right\}
      \\
      \\
      \begin{array}{c}
        P =
      \end{array}
      \begin{array}{c}
        \input{fig/exLcover}
        \definecolor{darkfacecolor}{rgb}{0.5,0.5,0.5}

        \fill[fill=darkfacecolor, draw=black, shift={(-1,1,0)}] (0, 0, 0) -- (1, 0, 0) -- (1, 1, 0) -- (0, 1, 0) -- cycle;
        \fill[fill=darkfacecolor, draw=black, shift={(-2,3,0)}] (0, 0, 0) -- (0, 1, 0) -- (0, 1, 1) -- (0, 0, 1) -- cycle;
        \fill[fill=darkfacecolor, draw=black, shift={(0,0,0)}] (0, 0, 0) -- (0, 0, 1) -- (1, 0, 1) -- (1, 0, 0) -- cycle;
        \draw[ ultra thick] (1,0,0) -- ++(0,0,1) -- ++(-2,0,0) -- ++(0,2,0) -- ++(0,0,-1) -- ++(2,0,0) -- ++(0,-2,0) -- cycle;

        \draw[ ultra thick] (0,1,0) -- ++(0,3,0) -- ++(-2,0,0) -- ++(0,0,1) -- ++(0,-1,0) -- ++(0,0,-1) -- ++(1,0,0) -- ++(0,-2,0) -- ++(1,0,0) -- cycle;

        \draw[ ultra thick] (-2,3,0) -- ++(0,2,0) -- ++(-1,0,0) -- ++(0,0,1) -- ++(0,-2,0) -- ++(1,0,0) -- cycle;
        \draw[ ultra thick] (1,-2,0) -- ++(0,2,0) -- ++(-1,0,0) -- ++(0,0,1) -- ++(0,-2,0) -- ++(1,0,0) -- cycle;

        \node[circle,fill=black,draw=black,minimum size=1mm,inner sep=0pt] at (1,-1.5,.5) {};
        \node[circle,fill=black,draw=black,minimum size=1mm,inner sep=0pt] at (-2.5,5,0.5) {};
      \end{tikzpicture}
    \end{array}

    \end{array}
    &
    \hspace{1cm}
    &
    \begin{array}{c}
      \begin{tikzpicture}
  [x={(-0.216506cm,-0.125000cm)}, y={(0.216506cm,-0.125000cm)}, z={(0.000000cm,0.250000cm)}]
  \def\loza#1#2#3#4#5#6{
    \definecolor{facecolor}{rgb}{#4,#5,#6}
    \fill[fill=facecolor, draw=black, shift={(#1, #2, #3)}]
    (0, 0, 0) -- (0, 1, 0) -- (0, 1, 1) -- (0, 0, 1) -- cycle;
  }
  \def\lozb#1#2#3#4#5#6{
    \definecolor{facecolor}{rgb}{#4,#5,#6}
    \fill[fill=facecolor, draw=black, shift={(#1, #2, #3)}]
    (0, 0, 0) -- (0, 0, 1) -- (1, 0, 1) -- (1, 0, 0) -- cycle;
  }
  \def\lozc#1#2#3#4#5#6{
    \definecolor{facecolor}{rgb}{#4,#5,#6}
    \fill[fill=facecolor, draw=black, shift={(#1, #2, #3)}]
    (0, 0, 0) -- (1, 0, 0) -- (1, 1, 0) -- (0, 1, 0) -- cycle;
  }
  \loza{4}{-2}{0}{0.80}{0.80}{0.80}
  \loza{6}{-1}{-2}{0.80}{0.80}{0.80}
  \loza{4}{-1}{-1}{0.80}{0.80}{0.80}
  \loza{7}{-3}{-1}{0.80}{0.80}{0.80}
  \loza{5}{-3}{0}{0.80}{0.80}{0.80}
  \loza{7}{-2}{-2}{0.80}{0.80}{0.80}
  \loza{5}{-2}{-1}{0.80}{0.80}{0.80}
  \lozb{8}{-3}{-1}{0.80}{0.80}{0.80}
  \lozb{6}{-3}{0}{0.80}{0.80}{0.80}
  \lozb{8}{-2}{-2}{0.80}{0.80}{0.80}
  \lozb{6}{-2}{-1}{0.80}{0.80}{0.80}
  \lozb{4}{-2}{0}{0.80}{0.80}{0.80}
  \lozb{6}{-1}{-2}{0.80}{0.80}{0.80}
  \lozb{4}{-1}{-1}{0.80}{0.80}{0.80}
  \lozb{7}{-3}{-1}{0.80}{0.80}{0.80}
  \lozb{5}{-3}{0}{0.80}{0.80}{0.80}
  \lozb{7}{-2}{-2}{0.80}{0.80}{0.80}
  \lozb{5}{-2}{-1}{0.80}{0.80}{0.80}
  \lozc{7}{-1}{-2}{0.80}{0.80}{0.80}
  \lozc{5}{-1}{-1}{0.80}{0.80}{0.80}
  \lozc{8}{-3}{-1}{0.80}{0.80}{0.80}
  \lozc{6}{-3}{0}{0.80}{0.80}{0.80}
  \lozc{8}{-2}{-2}{0.80}{0.80}{0.80}
  \lozc{6}{-2}{-1}{0.80}{0.80}{0.80}
  \lozc{4}{-2}{0}{0.80}{0.80}{0.80}
  \lozc{6}{-1}{-2}{0.80}{0.80}{0.80}
  \lozc{4}{-1}{-1}{0.80}{0.80}{0.80}
  \lozc{7}{-3}{-1}{0.80}{0.80}{0.80}
  \lozc{5}{-3}{0}{0.80}{0.80}{0.80}
  \lozc{7}{-2}{-2}{0.80}{0.80}{0.80}
  \lozc{5}{-2}{-1}{0.80}{0.80}{0.80}
  \loza{0}{3}{-2}{0.80}{0.80}{0.80}
  \loza{3}{1}{-2}{0.80}{0.80}{0.80}
  \loza{2}{4}{-4}{0.80}{0.80}{0.80}
  \loza{5}{2}{-4}{0.80}{0.80}{0.80}
  \loza{0}{4}{-3}{0.80}{0.80}{0.80}
  \loza{3}{2}{-3}{0.80}{0.80}{0.80}
  \loza{6}{0}{-3}{0.80}{0.80}{0.80}
  \loza{1}{2}{-2}{0.80}{0.80}{0.80}
  \loza{4}{0}{-2}{0.80}{0.80}{0.80}
  \loza{3}{3}{-4}{0.80}{0.80}{0.80}
  \loza{6}{1}{-4}{0.80}{0.80}{0.80}
  \loza{1}{3}{-3}{0.80}{0.80}{0.80}
  \loza{4}{1}{-3}{0.80}{0.80}{0.80}
  \lozb{7}{0}{-3}{0.80}{0.80}{0.80}
  \lozb{2}{2}{-2}{0.80}{0.80}{0.80}
  \lozb{5}{0}{-2}{0.80}{0.80}{0.80}
  \lozb{4}{3}{-4}{0.80}{0.80}{0.80}
  \lozb{7}{1}{-4}{0.80}{0.80}{0.80}
  \lozb{2}{3}{-3}{0.80}{0.80}{0.80}
  \lozb{5}{1}{-3}{0.80}{0.80}{0.80}
  \lozb{0}{3}{-2}{0.80}{0.80}{0.80}
  \lozb{3}{1}{-2}{0.80}{0.80}{0.80}
  \lozb{2}{4}{-4}{0.80}{0.80}{0.80}
  \lozb{5}{2}{-4}{0.80}{0.80}{0.80}
  \lozb{0}{4}{-3}{0.80}{0.80}{0.80}
  \lozb{3}{2}{-3}{0.80}{0.80}{0.80}
  \lozb{6}{0}{-3}{0.80}{0.80}{0.80}
  \lozb{1}{2}{-2}{0.80}{0.80}{0.80}
  \lozb{4}{0}{-2}{0.80}{0.80}{0.80}
  \lozb{3}{3}{-4}{0.80}{0.80}{0.80}
  \lozb{6}{1}{-4}{0.80}{0.80}{0.80}
  \lozb{1}{3}{-3}{0.80}{0.80}{0.80}
  \lozb{4}{1}{-3}{0.80}{0.80}{0.80}
  \lozc{3}{4}{-4}{0.80}{0.80}{0.80}
  \lozc{6}{2}{-4}{0.80}{0.80}{0.80}
  \lozc{1}{4}{-3}{0.80}{0.80}{0.80}
  \lozc{4}{2}{-3}{0.80}{0.80}{0.80}
  \lozc{7}{0}{-3}{0.80}{0.80}{0.80}
  \lozc{2}{2}{-2}{0.80}{0.80}{0.80}
  \lozc{5}{0}{-2}{0.80}{0.80}{0.80}
  \lozc{4}{3}{-4}{0.80}{0.80}{0.80}
  \lozc{7}{1}{-4}{0.80}{0.80}{0.80}
  \lozc{2}{3}{-3}{0.80}{0.80}{0.80}
  \lozc{5}{1}{-3}{0.80}{0.80}{0.80}
  \lozc{0}{3}{-2}{0.80}{0.80}{0.80}
  \lozc{3}{1}{-2}{0.80}{0.80}{0.80}
  \lozc{2}{4}{-4}{0.80}{0.80}{0.80}
  \lozc{5}{2}{-4}{0.80}{0.80}{0.80}
  \lozc{0}{4}{-3}{0.80}{0.80}{0.80}
  \lozc{3}{2}{-3}{0.80}{0.80}{0.80}
  \lozc{6}{0}{-3}{0.80}{0.80}{0.80}
  \lozc{1}{2}{-2}{0.80}{0.80}{0.80}
  \lozc{4}{0}{-2}{0.80}{0.80}{0.80}
  \lozc{3}{3}{-4}{0.80}{0.80}{0.80}
  \lozc{6}{1}{-4}{0.80}{0.80}{0.80}
  \lozc{1}{3}{-3}{0.80}{0.80}{0.80}
  \lozc{4}{1}{-3}{0.80}{0.80}{0.80}
  \loza{-3}{-1}{3}{0.80}{0.80}{0.80}
  \loza{-4}{2}{1}{0.80}{0.80}{0.80}
  \loza{0}{-3}{3}{0.80}{0.80}{0.80}
  \loza{-1}{0}{1}{0.30}{0.30}{0.30}
  \loza{-2}{3}{-1}{0.80}{0.80}{0.80}
  \loza{2}{-2}{1}{0.80}{0.80}{0.80}
  \loza{1}{1}{-1}{0.30}{0.30}{0.30}
  \loza{-3}{0}{2}{0.80}{0.80}{0.80}
  \loza{-4}{3}{0}{0.80}{0.80}{0.80}
  \loza{0}{-2}{2}{0.80}{0.80}{0.80}
  \loza{-1}{1}{0}{0.30}{0.30}{0.30}
  \loza{3}{-4}{2}{0.80}{0.80}{0.80}
  \loza{2}{-1}{0}{0.30}{0.30}{0.30}
  \loza{-2}{-2}{3}{0.80}{0.80}{0.80}
  \loza{-3}{1}{1}{0.80}{0.80}{0.80}
  \loza{1}{-4}{3}{0.80}{0.80}{0.80}
  \loza{0}{-1}{1}{0.30}{0.30}{0.30}
  \loza{-1}{2}{-1}{0.80}{0.80}{0.80}
  \loza{3}{-3}{1}{0.80}{0.80}{0.80}
  \loza{2}{0}{-1}{0.30}{0.30}{0.30}
  \loza{-2}{-1}{2}{0.80}{0.80}{0.80}
  \loza{-3}{2}{0}{0.80}{0.80}{0.80}
  \loza{1}{-3}{2}{0.80}{0.80}{0.80}
  \loza{0}{0}{0}{0.50}{0.50}{0.50}
  \lozb{4}{-4}{2}{0.80}{0.80}{0.80}
  \lozb{3}{-1}{0}{0.30}{0.30}{0.30}
  \lozb{-1}{-2}{3}{0.80}{0.80}{0.80}
  \lozb{-2}{1}{1}{0.80}{0.80}{0.80}
  \lozb{2}{-4}{3}{0.80}{0.80}{0.80}
  \lozb{1}{-1}{1}{0.30}{0.30}{0.30}
  \lozb{0}{2}{-1}{0.80}{0.80}{0.80}
  \lozb{4}{-3}{1}{0.80}{0.80}{0.80}
  \lozb{3}{0}{-1}{0.30}{0.30}{0.30}
  \lozb{-1}{-1}{2}{0.80}{0.80}{0.80}
  \lozb{-2}{2}{0}{0.80}{0.80}{0.80}
  \lozb{2}{-3}{2}{0.80}{0.80}{0.80}
  \lozb{1}{0}{0}{0.30}{0.30}{0.30}
  \lozb{-3}{-1}{3}{0.80}{0.80}{0.80}
  \lozb{-4}{2}{1}{0.80}{0.80}{0.80}
  \lozb{0}{-3}{3}{0.80}{0.80}{0.80}
  \lozb{-1}{0}{1}{0.30}{0.30}{0.30}
  \lozb{-2}{3}{-1}{0.80}{0.80}{0.80}
  \lozb{2}{-2}{1}{0.80}{0.80}{0.80}
  \lozb{1}{1}{-1}{0.30}{0.30}{0.30}
  \lozb{-3}{0}{2}{0.80}{0.80}{0.80}
  \lozb{-4}{3}{0}{0.80}{0.80}{0.80}
  \lozb{0}{-2}{2}{0.80}{0.80}{0.80}
  \lozb{-1}{1}{0}{0.30}{0.30}{0.30}
  \lozb{3}{-4}{2}{0.80}{0.80}{0.80}
  \lozb{2}{-1}{0}{0.30}{0.30}{0.30}
  \lozb{-2}{-2}{3}{0.80}{0.80}{0.80}
  \lozb{-3}{1}{1}{0.80}{0.80}{0.80}
  \lozb{1}{-4}{3}{0.80}{0.80}{0.80}
  \lozb{0}{-1}{1}{0.30}{0.30}{0.30}
  \lozb{-1}{2}{-1}{0.80}{0.80}{0.80}
  \lozb{3}{-3}{1}{0.80}{0.80}{0.80}
  \lozb{2}{0}{-1}{0.30}{0.30}{0.30}
  \lozb{-2}{-1}{2}{0.80}{0.80}{0.80}
  \lozb{-3}{2}{0}{0.80}{0.80}{0.80}
  \lozb{1}{-3}{2}{0.80}{0.80}{0.80}
  \lozb{0}{0}{0}{0.50}{0.50}{0.50}
  \lozc{-1}{3}{-1}{0.80}{0.80}{0.80}
  \lozc{3}{-2}{1}{0.80}{0.80}{0.80}
  \lozc{2}{1}{-1}{0.30}{0.30}{0.30}
  \lozc{-2}{0}{2}{0.80}{0.80}{0.80}
  \lozc{-3}{3}{0}{0.80}{0.80}{0.80}
  \lozc{1}{-2}{2}{0.80}{0.80}{0.80}
  \lozc{0}{1}{0}{0.30}{0.30}{0.30}
  \lozc{4}{-4}{2}{0.80}{0.80}{0.80}
  \lozc{3}{-1}{0}{0.30}{0.30}{0.30}
  \lozc{-1}{-2}{3}{0.80}{0.80}{0.80}
  \lozc{-2}{1}{1}{0.80}{0.80}{0.80}
  \lozc{2}{-4}{3}{0.80}{0.80}{0.80}
  \lozc{1}{-1}{1}{0.30}{0.30}{0.30}
  \lozc{0}{2}{-1}{0.80}{0.80}{0.80}
  \lozc{4}{-3}{1}{0.80}{0.80}{0.80}
  \lozc{3}{0}{-1}{0.30}{0.30}{0.30}
  \lozc{-1}{-1}{2}{0.80}{0.80}{0.80}
  \lozc{-2}{2}{0}{0.80}{0.80}{0.80}
  \lozc{2}{-3}{2}{0.80}{0.80}{0.80}
  \lozc{1}{0}{0}{0.30}{0.30}{0.30}
  \lozc{-3}{-1}{3}{0.80}{0.80}{0.80}
  \lozc{-4}{2}{1}{0.80}{0.80}{0.80}
  \lozc{0}{-3}{3}{0.80}{0.80}{0.80}
  \lozc{-1}{0}{1}{0.30}{0.30}{0.30}
  \lozc{-2}{3}{-1}{0.80}{0.80}{0.80}
  \lozc{2}{-2}{1}{0.80}{0.80}{0.80}
  \lozc{1}{1}{-1}{0.30}{0.30}{0.30}
  \lozc{-3}{0}{2}{0.80}{0.80}{0.80}
  \lozc{-4}{3}{0}{0.80}{0.80}{0.80}
  \lozc{0}{-2}{2}{0.80}{0.80}{0.80}
  \lozc{-1}{1}{0}{0.30}{0.30}{0.30}
  \lozc{3}{-4}{2}{0.80}{0.80}{0.80}
  \lozc{2}{-1}{0}{0.30}{0.30}{0.30}
  \lozc{-2}{-2}{3}{0.80}{0.80}{0.80}
  \lozc{-3}{1}{1}{0.80}{0.80}{0.80}
  \lozc{1}{-4}{3}{0.80}{0.80}{0.80}
  \lozc{0}{-1}{1}{0.30}{0.30}{0.30}
  \lozc{-1}{2}{-1}{0.80}{0.80}{0.80}
  \lozc{3}{-3}{1}{0.80}{0.80}{0.80}
  \lozc{2}{0}{-1}{0.30}{0.30}{0.30}
  \lozc{-2}{-1}{2}{0.80}{0.80}{0.80}
  \lozc{-3}{2}{0}{0.80}{0.80}{0.80}
  \lozc{1}{-3}{2}{0.80}{0.80}{0.80}
  \lozc{0}{0}{0}{0.50}{0.50}{0.50}
\end{tikzpicture}
    \end{array}
  \end{array}
  \]
  \caption{
    On the left, the pattern $P$ is $\mcL$-covered. Two faces of $P$ are
    connected via a sequence of patterns from $\mcL$.
    On the right, examples of $\LFS$-annulus. Patterns $P_0 \varsubsetneq P_4
    \varsubsetneq P_7$ are defined by $P_0$ = $\mcU$ and $P_{i+1} = \SFSiii(P_i)$. The
    lighter pattern $P_7 \setminus P_4$ is a $\LFS$-annulus of $P_4$ and the
    darker pattern $P_4 \setminus P_0$ is an $\LFS$-annulus of $P_0$.
    \label{fig:lcover}
  }
\end{figure}

\subsection{Covering properties for $\Sigma_1, \Sigma_2, \Sigma_3$}
Let $\LFS$ be the set of patterns containing
$\myvcenter{%
\begin{tikzpicture}
[x={(-0.173205cm,-0.100000cm)}, y={(0.173205cm,-0.100000cm)}, z={(0.000000cm,0.200000cm)}]
\definecolor{facecolor}{rgb}{0.8,0.8,0.8}
\fill[fill=facecolor, draw=black, shift={(0,0,0)}]
(0, 0, 0) -- (0, 0, 1) -- (1, 0, 1) -- (1, 0, 0) -- cycle;
\definecolor{facecolor}{rgb}{0.8,0.8,0.8}
\fill[fill=facecolor, draw=black, shift={(0,0,0)}]
(0, 0, 0) -- (0, 1, 0) -- (0, 1, 1) -- (0, 0, 1) -- cycle;
\end{tikzpicture}\,}%
,
\myvcenter{%
\begin{tikzpicture}
[x={(-0.173205cm,-0.100000cm)}, y={(0.173205cm,-0.100000cm)}, z={(0.000000cm,0.200000cm)}]
\definecolor{facecolor}{rgb}{0.8,0.8,0.8}
\fill[fill=facecolor, draw=black, shift={(0,0,0)}]
(0, 0, 0) -- (0, 1, 0) -- (0, 1, 1) -- (0, 0, 1) -- cycle;
\definecolor{facecolor}{rgb}{0.8,0.8,0.8}
\fill[fill=facecolor, draw=black, shift={(-1,1,0)}]
(0, 0, 0) -- (0, 0, 1) -- (1, 0, 1) -- (1, 0, 0) -- cycle;
\end{tikzpicture}\,}%
,
\myvcenter{%
\begin{tikzpicture}
[x={(-0.173205cm,-0.100000cm)}, y={(0.173205cm,-0.100000cm)}, z={(0.000000cm,0.200000cm)}]
\definecolor{facecolor}{rgb}{0.8,0.8,0.8}
\fill[fill=facecolor, draw=black, shift={(0,0,0)}]
(0, 0, 0) -- (0, 1, 0) -- (0, 1, 1) -- (0, 0, 1) -- cycle;
\definecolor{facecolor}{rgb}{0.8,0.8,0.8}
\fill[fill=facecolor, draw=black, shift={(0,0,0)}]
(0, 0, 0) -- (1, 0, 0) -- (1, 1, 0) -- (0, 1, 0) -- cycle;
\end{tikzpicture}\,}%
,
\myvcenter{%
\begin{tikzpicture}
[x={(-0.173205cm,-0.100000cm)}, y={(0.173205cm,-0.100000cm)}, z={(0.000000cm,0.200000cm)}]
\definecolor{facecolor}{rgb}{0.8,0.8,0.8}
\fill[fill=facecolor, draw=black, shift={(0,0,0)}]
(0, 0, 0) -- (0, 1, 0) -- (0, 1, 1) -- (0, 0, 1) -- cycle;
\definecolor{facecolor}{rgb}{0.8,0.8,0.8}
\fill[fill=facecolor, draw=black, shift={(-1,0,1)}]
(0, 0, 0) -- (1, 0, 0) -- (1, 1, 0) -- (0, 1, 0) -- cycle;
\end{tikzpicture}\,}%
,
\myvcenter{%
\begin{tikzpicture}
[x={(-0.173205cm,-0.100000cm)}, y={(0.173205cm,-0.100000cm)}, z={(0.000000cm,0.200000cm)}]
\definecolor{facecolor}{rgb}{0.8,0.8,0.8}
\fill[fill=facecolor, draw=black, shift={(0,0,0)}]
(0, 0, 0) -- (0, 0, 1) -- (1, 0, 1) -- (1, 0, 0) -- cycle;
\definecolor{facecolor}{rgb}{0.8,0.8,0.8}
\fill[fill=facecolor, draw=black, shift={(0,0,0)}]
(0, 0, 0) -- (1, 0, 0) -- (1, 1, 0) -- (0, 1, 0) -- cycle;
\end{tikzpicture}\,}%
,
\myvcenter{%
\begin{tikzpicture}
[x={(-0.173205cm,-0.100000cm)}, y={(0.173205cm,-0.100000cm)}, z={(0.000000cm,0.200000cm)}]
\definecolor{facecolor}{rgb}{0.8,0.8,0.8}
\fill[fill=facecolor, draw=black, shift={(0,-1,1)}]
(0, 0, 0) -- (1, 0, 0) -- (1, 1, 0) -- (0, 1, 0) -- cycle;
\definecolor{facecolor}{rgb}{0.8,0.8,0.8}
\fill[fill=facecolor, draw=black, shift={(0,0,0)}]
(0, 0, 0) -- (0, 0, 1) -- (1, 0, 1) -- (1, 0, 0) -- cycle;
\end{tikzpicture}\,}%
 \in \Ledge$
and
\[
\renewcommand{\arraystretch}{1.4}
\begin{array}{ccc}
\myvcenter{%
\begin{tikzpicture}
[x={(-0.173205cm,-0.100000cm)}, y={(0.173205cm,-0.100000cm)}, z={(0.000000cm,0.200000cm)}]
\definecolor{facecolor}{rgb}{0.800,0.800,0.800}
\fill[fill=facecolor, draw=black, shift={(0,0,0)}]
(0, 0, 0) -- (0, 0, 1) -- (1, 0, 1) -- (1, 0, 0) -- cycle;
\fill[fill=facecolor, draw=black, shift={(0,0,0)}]
(0, 0, 0) -- (1, 0, 0) -- (1, 1, 0) -- (0, 1, 0) -- cycle;
\fill[fill=facecolor, draw=black, shift={(1,0,0)}]
(0, 0, 0) -- (0, 0, 1) -- (1, 0, 1) -- (1, 0, 0) -- cycle;
\end{tikzpicture}\,}%
 & = & [\mathbf 0, 2]^\star \cup [(1,0,0), 2]^\star \cup [\mathbf 0, 3]^\star, \\
\myvcenter{%
\begin{tikzpicture}
[x={(-0.173205cm,-0.100000cm)}, y={(0.173205cm,-0.100000cm)}, z={(0.000000cm,0.200000cm)}]
\definecolor{facecolor}{rgb}{0.800,0.800,0.800}
\fill[fill=facecolor, draw=black, shift={(0,0,0)}]
(0, 0, 0) -- (0, 0, 1) -- (1, 0, 1) -- (1, 0, 0) -- cycle;
\fill[fill=facecolor, draw=black, shift={(1,0,0)}]
(0, 0, 0) -- (1, 0, 0) -- (1, 1, 0) -- (0, 1, 0) -- cycle;
\fill[fill=facecolor, draw=black, shift={(0,0,0)}]
(0, 0, 0) -- (1, 0, 0) -- (1, 1, 0) -- (0, 1, 0) -- cycle;
\end{tikzpicture}\,}%
 & = & [\mathbf 0, 3]^\star \cup [(1,0,0), 3]^\star \cup [\mathbf 0, 2]^\star, \\
\myvcenter{%
\begin{tikzpicture}
[x={(-0.173205cm,-0.100000cm)}, y={(0.173205cm,-0.100000cm)}, z={(0.000000cm,0.200000cm)}]
\definecolor{facecolor}{rgb}{0.800,0.800,0.800}
\fill[fill=facecolor, draw=black, shift={(0,0,0)}]
(0, 0, 0) -- (0, 1, 0) -- (0, 1, 1) -- (0, 0, 1) -- cycle;
\fill[fill=facecolor, draw=black, shift={(0,1,0)}]
(0, 0, 0) -- (1, 0, 0) -- (1, 1, 0) -- (0, 1, 0) -- cycle;
\fill[fill=facecolor, draw=black, shift={(0,0,0)}]
(0, 0, 0) -- (1, 0, 0) -- (1, 1, 0) -- (0, 1, 0) -- cycle;
\end{tikzpicture}\,}%
 & = & [\mathbf 0, 3]^\star \cup [(0,1,0), 3]^\star \cup [\mathbf 0, 1]^\star.
\end{array}
\]

\begin{lemm}[$\LFS$-covering]
\label{lemm:covFS}
Let $P$ be an $\LFS$-covered pattern.
Then the pattern $\Sigma_i(P)$ is $\LFS$-covered for every $i \in \{1,2,3\}$.
\end{lemm}

\begin{proof}
The proof relies on Lemma~\ref{lemm:coverprop} because
$\Sigma_i(P)$ is $\LFS$-covered for all $P \in \LFS$,
which can be checked by inspection of the twenty-seven cases below.
\begin{center}
\input{fig/FScov.tex}
\end{center}

\end{proof}

\begin{lemm}
\label{lemm:forbidden}
Let $\bfGa$ be a stepped plane that does not contain any translate of one of the patterns
$\myvcenter{%
\begin{tikzpicture}
[x={(-0.173205cm,-0.100000cm)}, y={(0.173205cm,-0.100000cm)}, z={(0.000000cm,0.200000cm)}]
\definecolor{facecolor}{rgb}{0.8,0.8,0.8}
\fill[fill=facecolor, draw=black, shift={(0,0,0)}]
(0, 0, 0) -- (0, 1, 0) -- (0, 1, 1) -- (0, 0, 1) -- cycle;
\fill[fill=facecolor, draw=black, shift={(0,1,0)}]
(0, 0, 0) -- (0, 1, 0) -- (0, 1, 1) -- (0, 0, 1) -- cycle;
\end{tikzpicture}\,}%
,
\myvcenter{%
\begin{tikzpicture}
[x={(-0.173205cm,-0.100000cm)}, y={(0.173205cm,-0.100000cm)}, z={(0.000000cm,0.200000cm)}]
\definecolor{facecolor}{rgb}{0.8,0.8,0.8}
\fill[fill=facecolor, draw=black, shift={(0,0,0)}]
(0, 0, 0) -- (0, 1, 0) -- (0, 1, 1) -- (0, 0, 1) -- cycle;
\fill[fill=facecolor, draw=black, shift={(0,0,1)}]
(0, 0, 0) -- (0, 1, 0) -- (0, 1, 1) -- (0, 0, 1) -- cycle;
\end{tikzpicture}\,}%
,
\myvcenter{%
\begin{tikzpicture}
[x={(-0.173205cm,-0.100000cm)}, y={(0.173205cm,-0.100000cm)}, z={(0.000000cm,0.200000cm)}]
\definecolor{facecolor}{rgb}{0.8,0.8,0.8}
\fill[fill=facecolor, draw=black, shift={(0,0,1)}]
(0, 0, 0) -- (0, 0, 1) -- (1, 0, 1) -- (1, 0, 0) -- cycle;
\fill[fill=facecolor, draw=black, shift={(0,0,0)}]
(0, 0, 0) -- (0, 0, 1) -- (1, 0, 1) -- (1, 0, 0) -- cycle;
\end{tikzpicture}\,}%
 \in \Ledge$
and
$\myvcenter{%
\begin{tikzpicture}
[x={(-0.173205cm,-0.100000cm)}, y={(0.173205cm,-0.100000cm)}, z={(0.000000cm,0.200000cm)}]
\definecolor{facecolor}{rgb}{0.8,0.8,0.8}
\fill[fill=facecolor, draw=black, shift={(0,0,0)}]
(0, 0, 0) -- (1, 0, 0) -- (1, 1, 0) -- (0, 1, 0) -- cycle;
\fill[fill=facecolor, draw=black, shift={(1,1,0)}]
(0, 0, 0) -- (1, 0, 0) -- (1, 1, 0) -- (0, 1, 0) -- cycle;
\end{tikzpicture}\,}%
 = [\mathbf 0, 3]^\star \cup [(1,1,0),3]^\star$.
Then no translate of any of these four patterns appears in $\Sigma_i(\bfGa)$.
\end{lemm}

\begin{proof}
The patterns \myvcenter{%
\begin{tikzpicture}
[x={(-0.173205cm,-0.100000cm)}, y={(0.173205cm,-0.100000cm)}, z={(0.000000cm,0.200000cm)}]
\definecolor{facecolor}{rgb}{0.8,0.8,0.8}
\fill[fill=facecolor, draw=black, shift={(0,0,0)}]
(0, 0, 0) -- (0, 1, 0) -- (0, 1, 1) -- (0, 0, 1) -- cycle;
\fill[fill=facecolor, draw=black, shift={(0,1,0)}]
(0, 0, 0) -- (0, 1, 0) -- (0, 1, 1) -- (0, 0, 1) -- cycle;
\end{tikzpicture}\,}%
 and \myvcenter{%
\begin{tikzpicture}
[x={(-0.173205cm,-0.100000cm)}, y={(0.173205cm,-0.100000cm)}, z={(0.000000cm,0.200000cm)}]
\definecolor{facecolor}{rgb}{0.8,0.8,0.8}
\fill[fill=facecolor, draw=black, shift={(0,0,0)}]
(0, 0, 0) -- (0, 1, 0) -- (0, 1, 1) -- (0, 0, 1) -- cycle;
\fill[fill=facecolor, draw=black, shift={(0,0,1)}]
(0, 0, 0) -- (0, 1, 0) -- (0, 1, 1) -- (0, 0, 1) -- cycle;
\end{tikzpicture}\,}%
 admit no preimage
that is included in a discrete plane,
as can be checked using Lemma~\ref{lemm:preimages} and Definition~\ref{def:stepped},
so they cannot appear in $\Sigma_i(\bfGa)$.

For the two other cases, below are listed
their possible preimages (in light gray) obtained by applying Lemma~\ref{lemm:preimages},
together with their only possible ``completion'' within a discrete plane (in dark gray)
following from Definition~\ref{def:stepped}.
\begin{center}
\renewcommand{\arraystretch}{1.5}
\begin{tabular}{c|c}
$i$ & $\Sigma_i^{-1}(\,\myvcenter{\myvcenter{%
\begin{tikzpicture}
[x={(-0.173205cm,-0.100000cm)}, y={(0.173205cm,-0.100000cm)}, z={(0.000000cm,0.200000cm)}]
\definecolor{facecolor}{rgb}{0.8,0.8,0.8}
\fill[fill=facecolor, draw=black, shift={(0,0,1)}]
(0, 0, 0) -- (0, 0, 1) -- (1, 0, 1) -- (1, 0, 0) -- cycle;
\fill[fill=facecolor, draw=black, shift={(0,0,0)}]
(0, 0, 0) -- (0, 0, 1) -- (1, 0, 1) -- (1, 0, 0) -- cycle;
\end{tikzpicture}\,}%
})$ \\
\hline
1 &
\myvcenter{%
\begin{tikzpicture}
[x={(-0.173205cm,-0.100000cm)}, y={(0.173205cm,-0.100000cm)}, z={(0.000000cm,0.200000cm)}]
\definecolor{facecolor}{rgb}{0.800,0.800,0.800}
\fill[fill=facecolor, draw=black, shift={(0,0,1)}]
(0, 0, 0) -- (0, 0, 1) -- (1, 0, 1) -- (1, 0, 0) -- cycle;
\fill[fill=facecolor, draw=black, shift={(0,0,0)}]
(0, 0, 0) -- (0, 1, 0) -- (0, 1, 1) -- (0, 0, 1) -- cycle;
\definecolor{facecolor}{rgb}{0.35,0.35,0.35}
\fill[fill=facecolor, draw=black, shift={(0,0,0)}]
(0, 0, 0) -- (0, 0, 1) -- (1, 0, 1) -- (1, 0, 0) -- cycle;
\end{tikzpicture}}
\ or \
\myvcenter{%
\begin{tikzpicture}
[x={(-0.173205cm,-0.100000cm)}, y={(0.173205cm,-0.100000cm)}, z={(0.000000cm,0.200000cm)}]
\definecolor{facecolor}{rgb}{0.800,0.800,0.800}
\fill[fill=facecolor, draw=black, shift={(-1,0,0)}]
(0, 0, 0) -- (0, 0, 1) -- (1, 0, 1) -- (1, 0, 0) -- cycle;
\fill[fill=facecolor, draw=black, shift={(0,0,1)}]
(0, 0, 0) -- (0, 0, 1) -- (1, 0, 1) -- (1, 0, 0) -- cycle;
\definecolor{facecolor}{rgb}{0.35,0.35,0.35}
\fill[fill=facecolor, draw=black, shift={(0,0,0)}]
(0, 0, 0) -- (0, 0, 1) -- (1, 0, 1) -- (1, 0, 0) -- cycle;
\end{tikzpicture}}
  \\
2 &
\myvcenter{%
\begin{tikzpicture}
[x={(-0.173205cm,-0.100000cm)}, y={(0.173205cm,-0.100000cm)}, z={(0.000000cm,0.200000cm)}]
\definecolor{facecolor}{rgb}{0.800,0.800,0.800}
\fill[fill=facecolor, draw=black, shift={(0,0,0)}]
(0, 0, 0) -- (0, 0, 1) -- (1, 0, 1) -- (1, 0, 0) -- cycle;
\fill[fill=facecolor, draw=black, shift={(0,0,1)}]
(0, 0, 0) -- (0, 1, 0) -- (0, 1, 1) -- (0, 0, 1) -- cycle;
\definecolor{facecolor}{rgb}{0.35,0.35,0.35}
\fill[fill=facecolor, draw=black, shift={(0,0,0)}]
(0, 0, 0) -- (0, 1, 0) -- (0, 1, 1) -- (0, 0, 1) -- cycle;
\end{tikzpicture}}
\ or \
\myvcenter{%
\begin{tikzpicture}
[x={(-0.173205cm,-0.100000cm)}, y={(0.173205cm,-0.100000cm)}, z={(0.000000cm,0.200000cm)}]
\definecolor{facecolor}{rgb}{0.800,0.800,0.800}
\fill[fill=facecolor, draw=black, shift={(0,-1,0)}]
(0, 0, 0) -- (0, 1, 0) -- (0, 1, 1) -- (0, 0, 1) -- cycle;
\fill[fill=facecolor, draw=black, shift={(0,0,1)}]
(0, 0, 0) -- (0, 1, 0) -- (0, 1, 1) -- (0, 0, 1) -- cycle;
\definecolor{facecolor}{rgb}{0.35,0.35,0.35}
\fill[fill=facecolor, draw=black, shift={(0,0,0)}]
(0, 0, 0) -- (0, 1, 0) -- (0, 1, 1) -- (0, 0, 1) -- cycle;
\end{tikzpicture}}
  \\
3 &
\myvcenter{%
\begin{tikzpicture}
[x={(-0.173205cm,-0.100000cm)}, y={(0.173205cm,-0.100000cm)}, z={(0.000000cm,0.200000cm)}]
\definecolor{facecolor}{rgb}{0.800,0.800,0.800}
\fill[fill=facecolor, draw=black, shift={(0,0,0)}]
(0, 0, 0) -- (1, 0, 0) -- (1, 1, 0) -- (0, 1, 0) -- cycle;
\fill[fill=facecolor, draw=black, shift={(0,1,0)}]
(0, 0, 0) -- (0, 1, 0) -- (0, 1, 1) -- (0, 0, 1) -- cycle;
\definecolor{facecolor}{rgb}{0.35,0.35,0.35}
\fill[fill=facecolor, draw=black, shift={(0,0,0)}]
(0, 0, 0) -- (0, 1, 0) -- (0, 1, 1) -- (0, 0, 1) -- cycle;
\end{tikzpicture}}
\ or \
\myvcenter{%
\begin{tikzpicture}
[x={(-0.173205cm,-0.100000cm)}, y={(0.173205cm,-0.100000cm)}, z={(0.000000cm,0.200000cm)}]
\definecolor{facecolor}{rgb}{0.800,0.800,0.800}
\fill[fill=facecolor, draw=black, shift={(0,0,-1)}]
(0, 0, 0) -- (0, 1, 0) -- (0, 1, 1) -- (0, 0, 1) -- cycle;
\fill[fill=facecolor, draw=black, shift={(0,1,0)}]
(0, 0, 0) -- (0, 1, 0) -- (0, 1, 1) -- (0, 0, 1) -- cycle;
\definecolor{facecolor}{rgb}{0.35,0.35,0.35}
\fill[fill=facecolor, draw=black, shift={(0,0,0)}]
(0, 0, 0) -- (0, 1, 0) -- (0, 1, 1) -- (0, 0, 1) -- cycle;
\end{tikzpicture}}
\end{tabular}
\qquad
\begin{tabular}{c|c}
$i$ & $\Sigma_i^{-1}(\,\myvcenter{\myvcenter{%
\begin{tikzpicture}
[x={(-0.173205cm,-0.100000cm)}, y={(0.173205cm,-0.100000cm)}, z={(0.000000cm,0.200000cm)}]
\definecolor{facecolor}{rgb}{0.8,0.8,0.8}
\fill[fill=facecolor, draw=black, shift={(0,0,0)}]
(0, 0, 0) -- (1, 0, 0) -- (1, 1, 0) -- (0, 1, 0) -- cycle;
\fill[fill=facecolor, draw=black, shift={(1,1,0)}]
(0, 0, 0) -- (1, 0, 0) -- (1, 1, 0) -- (0, 1, 0) -- cycle;
\end{tikzpicture}\,}%
})$ \\
\hline
1 &
\myvcenter{%
\begin{tikzpicture}
[x={(-0.173205cm,-0.100000cm)}, y={(0.173205cm,-0.100000cm)}, z={(0.000000cm,0.200000cm)}]
\definecolor{facecolor}{rgb}{0.800000,0.800000,0.800000}
\fill[fill=facecolor, draw=black, shift={(0,0,0)}]
(0, 0, 0) -- (0, 1, 0) -- (0, 1, 1) -- (0, 0, 1) -- cycle;
\definecolor{facecolor}{rgb}{0.800000,0.800000,0.800000}
\fill[fill=facecolor, draw=black, shift={(1,1,0)}]
(0, 0, 0) -- (1, 0, 0) -- (1, 1, 0) -- (0, 1, 0) -- cycle;
\definecolor{facecolor}{rgb}{0.35,0.35,0.35}
\fill[fill=facecolor, draw=black, shift={(0,0,0)}]
(0, 0, 0) -- (1, 0, 0) -- (1, 1, 0) -- (0, 1, 0) -- cycle;
\end{tikzpicture}}
\ or \
\myvcenter{%
\begin{tikzpicture}
[x={(-0.173205cm,-0.100000cm)}, y={(0.173205cm,-0.100000cm)}, z={(0.000000cm,0.200000cm)}]
\definecolor{facecolor}{rgb}{0.800000,0.800000,0.800000}
\fill[fill=facecolor, draw=black, shift={(-1,0,0)}]
(0, 0, 0) -- (1, 0, 0) -- (1, 1, 0) -- (0, 1, 0) -- cycle;
\definecolor{facecolor}{rgb}{0.800000,0.800000,0.800000}
\fill[fill=facecolor, draw=black, shift={(1,1,0)}]
(0, 0, 0) -- (1, 0, 0) -- (1, 1, 0) -- (0, 1, 0) -- cycle;
\definecolor{facecolor}{rgb}{0.35,0.35,0.35}
\fill[fill=facecolor, draw=black, shift={(0,0,0)}]
(0, 0, 0) -- (1, 0, 0) -- (1, 1, 0) -- (0, 1, 0) -- cycle;
\end{tikzpicture}}
  \\
2 &
\myvcenter{%
\begin{tikzpicture}
[x={(-0.173205cm,-0.100000cm)}, y={(0.173205cm,-0.100000cm)}, z={(0.000000cm,0.200000cm)}]
\definecolor{facecolor}{rgb}{0.800000,0.800000,0.800000}
\fill[fill=facecolor, draw=black, shift={(0,0,0)}]
(0, 0, 0) -- (0, 0, 1) -- (1, 0, 1) -- (1, 0, 0) -- cycle;
\definecolor{facecolor}{rgb}{0.800000,0.800000,0.800000}
\fill[fill=facecolor, draw=black, shift={(1,1,0)}]
(0, 0, 0) -- (1, 0, 0) -- (1, 1, 0) -- (0, 1, 0) -- cycle;
\definecolor{facecolor}{rgb}{0.35,0.35,0.35}
\fill[fill=facecolor, draw=black, shift={(0,0,0)}]
(0, 0, 0) -- (1, 0, 0) -- (1, 1, 0) -- (0, 1, 0) -- cycle;
\end{tikzpicture}}
\ or \
\myvcenter{%
\begin{tikzpicture}
[x={(-0.173205cm,-0.100000cm)}, y={(0.173205cm,-0.100000cm)}, z={(0.000000cm,0.200000cm)}]
\definecolor{facecolor}{rgb}{0.800000,0.800000,0.800000}
\fill[fill=facecolor, draw=black, shift={(0,-1,0)}]
(0, 0, 0) -- (1, 0, 0) -- (1, 1, 0) -- (0, 1, 0) -- cycle;
\definecolor{facecolor}{rgb}{0.800000,0.800000,0.800000}
\fill[fill=facecolor, draw=black, shift={(1,1,0)}]
(0, 0, 0) -- (1, 0, 0) -- (1, 1, 0) -- (0, 1, 0) -- cycle;
\definecolor{facecolor}{rgb}{0.35,0.35,0.35}
\fill[fill=facecolor, draw=black, shift={(0,0,0)}]
(0, 0, 0) -- (1, 0, 0) -- (1, 1, 0) -- (0, 1, 0) -- cycle;
\end{tikzpicture}}
  \\
3 &
\myvcenter{%
\begin{tikzpicture}
[x={(-0.173205cm,-0.100000cm)}, y={(0.173205cm,-0.100000cm)}, z={(0.000000cm,0.200000cm)}]
\definecolor{facecolor}{rgb}{0.800000,0.800000,0.800000}
\fill[fill=facecolor, draw=black, shift={(0,0,0)}]
(0, 0, 0) -- (1, 0, 0) -- (1, 1, 0) -- (0, 1, 0) -- cycle;
\fill[fill=facecolor, draw=black, shift={(1,0,1)}]
(0, 0, 0) -- (0, 0, 1) -- (1, 0, 1) -- (1, 0, 0) -- cycle;
\definecolor{facecolor}{rgb}{0.35,0.35,0.35}
\fill[fill=facecolor, draw=black, shift={(1,0,0)}]
(0, 0, 0) -- (0, 0, 1) -- (1, 0, 1) -- (1, 0, 0) -- cycle;
\end{tikzpicture}}
\ or \
\myvcenter{%
\begin{tikzpicture}
[x={(-0.173205cm,-0.100000cm)}, y={(0.173205cm,-0.100000cm)}, z={(0.000000cm,0.200000cm)}]
\definecolor{facecolor}{rgb}{0.800000,0.800000,0.800000}
\fill[fill=facecolor, draw=black, shift={(0,0,-1)}]
(0, 0, 0) -- (0, 0, 1) -- (1, 0, 1) -- (1, 0, 0) -- cycle;
\fill[fill=facecolor, draw=black, shift={(1,0,1)}]
(0, 0, 0) -- (0, 0, 1) -- (1, 0, 1) -- (1, 0, 0) -- cycle;
\definecolor{facecolor}{rgb}{0.35,0.35,0.35}
\fill[fill=facecolor, draw=black, shift={(1,0,0)}]
(0, 0, 0) -- (0, 0, 1) -- (1, 0, 1) -- (1, 0, 0) -- cycle;
\end{tikzpicture}}
\end{tabular}
\end{center}
These two tables show that if one of these two patterns appears in $\Sigma(\bfGa)$,
then one of the four patterns must appear in $\bfGa$, which concludes the proof.
\end{proof}

\begin{lemm}[Strong $\LFS$-covering]
\label{lemm:strongcovFS}
Let $P$ be a strongly $\LFS$-covered pattern
which is contained in a stepped plane that avoids
\myvcenter{%
\begin{tikzpicture}
[x={(-0.173205cm,-0.100000cm)}, y={(0.173205cm,-0.100000cm)}, z={(0.000000cm,0.200000cm)}]
\definecolor{facecolor}{rgb}{0.8,0.8,0.8}
\fill[fill=facecolor, draw=black, shift={(0,0,0)}]
(0, 0, 0) -- (0, 1, 0) -- (0, 1, 1) -- (0, 0, 1) -- cycle;
\fill[fill=facecolor, draw=black, shift={(0,1,0)}]
(0, 0, 0) -- (0, 1, 0) -- (0, 1, 1) -- (0, 0, 1) -- cycle;
\end{tikzpicture}\,}%
,
\myvcenter{%
\begin{tikzpicture}
[x={(-0.173205cm,-0.100000cm)}, y={(0.173205cm,-0.100000cm)}, z={(0.000000cm,0.200000cm)}]
\definecolor{facecolor}{rgb}{0.8,0.8,0.8}
\fill[fill=facecolor, draw=black, shift={(0,0,0)}]
(0, 0, 0) -- (0, 1, 0) -- (0, 1, 1) -- (0, 0, 1) -- cycle;
\fill[fill=facecolor, draw=black, shift={(0,0,1)}]
(0, 0, 0) -- (0, 1, 0) -- (0, 1, 1) -- (0, 0, 1) -- cycle;
\end{tikzpicture}\,}%
,
\myvcenter{%
\begin{tikzpicture}
[x={(-0.173205cm,-0.100000cm)}, y={(0.173205cm,-0.100000cm)}, z={(0.000000cm,0.200000cm)}]
\definecolor{facecolor}{rgb}{0.8,0.8,0.8}
\fill[fill=facecolor, draw=black, shift={(0,0,1)}]
(0, 0, 0) -- (0, 0, 1) -- (1, 0, 1) -- (1, 0, 0) -- cycle;
\fill[fill=facecolor, draw=black, shift={(0,0,0)}]
(0, 0, 0) -- (0, 0, 1) -- (1, 0, 1) -- (1, 0, 0) -- cycle;
\end{tikzpicture}\,}%
and \myvcenter{%
\begin{tikzpicture}
[x={(-0.173205cm,-0.100000cm)}, y={(0.173205cm,-0.100000cm)}, z={(0.000000cm,0.200000cm)}]
\definecolor{facecolor}{rgb}{0.8,0.8,0.8}
\fill[fill=facecolor, draw=black, shift={(0,0,0)}]
(0, 0, 0) -- (1, 0, 0) -- (1, 1, 0) -- (0, 1, 0) -- cycle;
\fill[fill=facecolor, draw=black, shift={(1,1,0)}]
(0, 0, 0) -- (1, 0, 0) -- (1, 1, 0) -- (0, 1, 0) -- cycle;
\end{tikzpicture}\,}%
.
Then $\Sigma_i(P)$ is strongly $\LFS$-covered for every $i \in \{1,2,3\}$.
\end{lemm}

\begin{proof}
Let $i \in \{1,2,3\}$.
The pattern $\Sigma(P)$ is $\LFS$-covered thanks to Lemma~\ref{lemm:covFS}.
Now, let $X \subseteq \Sigma_i(P)$ be an element of $\Ledge$.
We must prove that there exists $Y \in \LFS$ such that $X \subseteq Y \subseteq \Sigma(P)$.

If $X$ is a translation of one of the six patterns
\myvcenter{%
\begin{tikzpicture}
[x={(-0.173205cm,-0.100000cm)}, y={(0.173205cm,-0.100000cm)}, z={(0.000000cm,0.200000cm)}]
\definecolor{facecolor}{rgb}{0.8,0.8,0.8}
\fill[fill=facecolor, draw=black, shift={(0,0,0)}]
(0, 0, 0) -- (0, 0, 1) -- (1, 0, 1) -- (1, 0, 0) -- cycle;
\definecolor{facecolor}{rgb}{0.8,0.8,0.8}
\fill[fill=facecolor, draw=black, shift={(0,0,0)}]
(0, 0, 0) -- (0, 1, 0) -- (0, 1, 1) -- (0, 0, 1) -- cycle;
\end{tikzpicture}\,}%
,
\myvcenter{%
\begin{tikzpicture}
[x={(-0.173205cm,-0.100000cm)}, y={(0.173205cm,-0.100000cm)}, z={(0.000000cm,0.200000cm)}]
\definecolor{facecolor}{rgb}{0.8,0.8,0.8}
\fill[fill=facecolor, draw=black, shift={(0,0,0)}]
(0, 0, 0) -- (0, 1, 0) -- (0, 1, 1) -- (0, 0, 1) -- cycle;
\definecolor{facecolor}{rgb}{0.8,0.8,0.8}
\fill[fill=facecolor, draw=black, shift={(-1,1,0)}]
(0, 0, 0) -- (0, 0, 1) -- (1, 0, 1) -- (1, 0, 0) -- cycle;
\end{tikzpicture}\,}%
,
\myvcenter{%
\begin{tikzpicture}
[x={(-0.173205cm,-0.100000cm)}, y={(0.173205cm,-0.100000cm)}, z={(0.000000cm,0.200000cm)}]
\definecolor{facecolor}{rgb}{0.8,0.8,0.8}
\fill[fill=facecolor, draw=black, shift={(0,0,0)}]
(0, 0, 0) -- (0, 1, 0) -- (0, 1, 1) -- (0, 0, 1) -- cycle;
\definecolor{facecolor}{rgb}{0.8,0.8,0.8}
\fill[fill=facecolor, draw=black, shift={(0,0,0)}]
(0, 0, 0) -- (1, 0, 0) -- (1, 1, 0) -- (0, 1, 0) -- cycle;
\end{tikzpicture}\,}%
,
\myvcenter{%
\begin{tikzpicture}
[x={(-0.173205cm,-0.100000cm)}, y={(0.173205cm,-0.100000cm)}, z={(0.000000cm,0.200000cm)}]
\definecolor{facecolor}{rgb}{0.8,0.8,0.8}
\fill[fill=facecolor, draw=black, shift={(0,0,0)}]
(0, 0, 0) -- (0, 1, 0) -- (0, 1, 1) -- (0, 0, 1) -- cycle;
\definecolor{facecolor}{rgb}{0.8,0.8,0.8}
\fill[fill=facecolor, draw=black, shift={(-1,0,1)}]
(0, 0, 0) -- (1, 0, 0) -- (1, 1, 0) -- (0, 1, 0) -- cycle;
\end{tikzpicture}\,}%
,
\myvcenter{%
\begin{tikzpicture}
[x={(-0.173205cm,-0.100000cm)}, y={(0.173205cm,-0.100000cm)}, z={(0.000000cm,0.200000cm)}]
\definecolor{facecolor}{rgb}{0.8,0.8,0.8}
\fill[fill=facecolor, draw=black, shift={(0,0,0)}]
(0, 0, 0) -- (0, 0, 1) -- (1, 0, 1) -- (1, 0, 0) -- cycle;
\definecolor{facecolor}{rgb}{0.8,0.8,0.8}
\fill[fill=facecolor, draw=black, shift={(0,0,0)}]
(0, 0, 0) -- (1, 0, 0) -- (1, 1, 0) -- (0, 1, 0) -- cycle;
\end{tikzpicture}\,}%
,
\myvcenter{%
\begin{tikzpicture}
[x={(-0.173205cm,-0.100000cm)}, y={(0.173205cm,-0.100000cm)}, z={(0.000000cm,0.200000cm)}]
\definecolor{facecolor}{rgb}{0.8,0.8,0.8}
\fill[fill=facecolor, draw=black, shift={(0,-1,1)}]
(0, 0, 0) -- (1, 0, 0) -- (1, 1, 0) -- (0, 1, 0) -- cycle;
\definecolor{facecolor}{rgb}{0.8,0.8,0.8}
\fill[fill=facecolor, draw=black, shift={(0,0,0)}]
(0, 0, 0) -- (0, 0, 1) -- (1, 0, 1) -- (1, 0, 0) -- cycle;
\end{tikzpicture}\,}%
,
then the strong covering condition is trivially verified by taking $X = Y$.
By assumption and by Lemma~\ref{lemm:forbidden},
$\Sigma(P)$ does not contain any translate of
\myvcenter{%
\begin{tikzpicture}
[x={(-0.173205cm,-0.100000cm)}, y={(0.173205cm,-0.100000cm)}, z={(0.000000cm,0.200000cm)}]
\definecolor{facecolor}{rgb}{0.8,0.8,0.8}
\fill[fill=facecolor, draw=black, shift={(0,0,0)}]
(0, 0, 0) -- (0, 1, 0) -- (0, 1, 1) -- (0, 0, 1) -- cycle;
\fill[fill=facecolor, draw=black, shift={(0,1,0)}]
(0, 0, 0) -- (0, 1, 0) -- (0, 1, 1) -- (0, 0, 1) -- cycle;
\end{tikzpicture}\,}%
, \myvcenter{%
\begin{tikzpicture}
[x={(-0.173205cm,-0.100000cm)}, y={(0.173205cm,-0.100000cm)}, z={(0.000000cm,0.200000cm)}]
\definecolor{facecolor}{rgb}{0.8,0.8,0.8}
\fill[fill=facecolor, draw=black, shift={(0,0,0)}]
(0, 0, 0) -- (0, 1, 0) -- (0, 1, 1) -- (0, 0, 1) -- cycle;
\fill[fill=facecolor, draw=black, shift={(0,0,1)}]
(0, 0, 0) -- (0, 1, 0) -- (0, 1, 1) -- (0, 0, 1) -- cycle;
\end{tikzpicture}\,}%
 or \myvcenter{%
\begin{tikzpicture}
[x={(-0.173205cm,-0.100000cm)}, y={(0.173205cm,-0.100000cm)}, z={(0.000000cm,0.200000cm)}]
\definecolor{facecolor}{rgb}{0.8,0.8,0.8}
\fill[fill=facecolor, draw=black, shift={(0,0,1)}]
(0, 0, 0) -- (0, 0, 1) -- (1, 0, 1) -- (1, 0, 0) -- cycle;
\fill[fill=facecolor, draw=black, shift={(0,0,0)}]
(0, 0, 0) -- (0, 0, 1) -- (1, 0, 1) -- (1, 0, 0) -- cycle;
\end{tikzpicture}\,}%
,
so we ignore these cases for $X$.

It remains to treat the cases $X = \myvcenter{%
\begin{tikzpicture}
[x={(-0.173205cm,-0.100000cm)}, y={(0.173205cm,-0.100000cm)}, z={(0.000000cm,0.200000cm)}]
\definecolor{facecolor}{rgb}{0.8,0.8,0.8}
\fill[fill=facecolor, draw=black, shift={(0,0,0)}]
(0, 0, 0) -- (0, 0, 1) -- (1, 0, 1) -- (1, 0, 0) -- cycle;
\fill[fill=facecolor, draw=black, shift={(1,0,0)}]
(0, 0, 0) -- (0, 0, 1) -- (1, 0, 1) -- (1, 0, 0) -- cycle;
\end{tikzpicture}\,}%
, \myvcenter{%
\begin{tikzpicture}
[x={(-0.173205cm,-0.100000cm)}, y={(0.173205cm,-0.100000cm)}, z={(0.000000cm,0.200000cm)}]
\definecolor{facecolor}{rgb}{0.8,0.8,0.8}
\fill[fill=facecolor, draw=black, shift={(1,0,0)}]
(0, 0, 0) -- (1, 0, 0) -- (1, 1, 0) -- (0, 1, 0) -- cycle;
\fill[fill=facecolor, draw=black, shift={(0,0,0)}]
(0, 0, 0) -- (1, 0, 0) -- (1, 1, 0) -- (0, 1, 0) -- cycle;
\end{tikzpicture}\,}%
$
or \myvcenter{%
\begin{tikzpicture}
[x={(-0.173205cm,-0.100000cm)}, y={(0.173205cm,-0.100000cm)}, z={(0.000000cm,0.200000cm)}]
\definecolor{facecolor}{rgb}{0.8,0.8,0.8}
\fill[fill=facecolor, draw=black, shift={(0,1,0)}]
(0, 0, 0) -- (1, 0, 0) -- (1, 1, 0) -- (0, 1, 0) -- cycle;
\fill[fill=facecolor, draw=black, shift={(0,0,0)}]
(0, 0, 0) -- (1, 0, 0) -- (1, 1, 0) -- (0, 1, 0) -- cycle;
\end{tikzpicture}\,}%
.
We have $X \subseteq \Sigma(P)$, so it is sufficient to check that
for every pattern $Q \subseteq P$ such that $X \subseteq \Sigma(Q) $,
there exists $Y \in \LFS$ such that $X \subseteq Y \subseteq \Sigma(Q)$.
Moreover, since $X$ is a two-face pattern,
we can restrict to the case where $Q$ consists of two faces only,
which leaves a finite number of cases to check for $Q$
(using Lemma~\ref{lemm:preimages}):
\begin{center}
\renewcommand{\arraystretch}{1.2}
\begin{tabular}[h]{c|c|c}
$i$ & $Q$ & $\Sigma_i(Q)$ \\
\hline
$1$ & \myvcenter{%
\begin{tikzpicture}
[x={(-0.173205cm,-0.100000cm)}, y={(0.173205cm,-0.100000cm)}, z={(0.000000cm,0.200000cm)}]
\definecolor{facecolor}{rgb}{0.8,0.8,0.8}
\fill[fill=facecolor, draw=black, shift={(0,0,0)}]
(0, 0, 0) -- (0, 0, 1) -- (1, 0, 1) -- (1, 0, 0) -- cycle;
\definecolor{facecolor}{rgb}{0.8,0.8,0.8}
\fill[fill=facecolor, draw=black, shift={(0,0,0)}]
(0, 0, 0) -- (0, 1, 0) -- (0, 1, 1) -- (0, 0, 1) -- cycle;
\end{tikzpicture}\,}%
    & \myvcenter{%
        \begin{tikzpicture}
        [x={(-0.173205cm,-0.100000cm)}, y={(0.173205cm,-0.100000cm)}, z={(0.000000cm,0.200000cm)}]
        \definecolor{facecolor}{rgb}{0.800,0.800,0.800}
        \fill[fill=facecolor, draw=black, shift={(0,0,0)}]
        (0, 0, 0) -- (0, 0, 1) -- (1, 0, 1) -- (1, 0, 0) -- cycle;
        \fill[fill=facecolor, draw=black, shift={(0,0,0)}]
        (0, 0, 0) -- (0, 1, 0) -- (0, 1, 1) -- (0, 0, 1) -- cycle;
        \fill[fill=facecolor, draw=black, shift={(1,0,0)}]
        (0, 0, 0) -- (0, 0, 1) -- (1, 0, 1) -- (1, 0, 0) -- cycle;
        \fill[fill=facecolor, draw=black, shift={(0,0,0)}]
        (0, 0, 0) -- (1, 0, 0) -- (1, 1, 0) -- (0, 1, 0) -- cycle;
        \end{tikzpicture}} \\
$1$ & \myvcenter{%
\begin{tikzpicture}
[x={(-0.173205cm,-0.100000cm)}, y={(0.173205cm,-0.100000cm)}, z={(0.000000cm,0.200000cm)}]
\definecolor{facecolor}{rgb}{0.8,0.8,0.8}
\fill[fill=facecolor, draw=black, shift={(0,0,0)}]
(0, 0, 0) -- (0, 0, 1) -- (1, 0, 1) -- (1, 0, 0) -- cycle;
\fill[fill=facecolor, draw=black, shift={(1,0,0)}]
(0, 0, 0) -- (0, 0, 1) -- (1, 0, 1) -- (1, 0, 0) -- cycle;
\end{tikzpicture}\,}%
    & \myvcenter{%
\begin{tikzpicture}
[x={(-0.173205cm,-0.100000cm)}, y={(0.173205cm,-0.100000cm)}, z={(0.000000cm,0.200000cm)}]
\definecolor{facecolor}{rgb}{0.8,0.8,0.8}
\fill[fill=facecolor, draw=black, shift={(0,0,0)}]
(0, 0, 0) -- (0, 0, 1) -- (1, 0, 1) -- (1, 0, 0) -- cycle;
\fill[fill=facecolor, draw=black, shift={(1,0,0)}]
(0, 0, 0) -- (0, 0, 1) -- (1, 0, 1) -- (1, 0, 0) -- cycle;
\end{tikzpicture}\,}%
 \\
$2$ & \myvcenter{%
\begin{tikzpicture}
[x={(-0.173205cm,-0.100000cm)}, y={(0.173205cm,-0.100000cm)}, z={(0.000000cm,0.200000cm)}]
\definecolor{facecolor}{rgb}{0.8,0.8,0.8}
\fill[fill=facecolor, draw=black, shift={(0,0,0)}]
(0, 0, 0) -- (0, 0, 1) -- (1, 0, 1) -- (1, 0, 0) -- cycle;
\definecolor{facecolor}{rgb}{0.8,0.8,0.8}
\fill[fill=facecolor, draw=black, shift={(0,0,0)}]
(0, 0, 0) -- (0, 1, 0) -- (0, 1, 1) -- (0, 0, 1) -- cycle;
\end{tikzpicture}\,}%
    & \myvcenter{%
        \begin{tikzpicture}
        [x={(-0.173205cm,-0.100000cm)}, y={(0.173205cm,-0.100000cm)}, z={(0.000000cm,0.200000cm)}]
        \definecolor{facecolor}{rgb}{0.800,0.800,0.800}
        \fill[fill=facecolor, draw=black, shift={(0,0,0)}]
        (0, 0, 0) -- (0, 0, 1) -- (1, 0, 1) -- (1, 0, 0) -- cycle;
        \fill[fill=facecolor, draw=black, shift={(0,0,0)}]
        (0, 0, 0) -- (0, 1, 0) -- (0, 1, 1) -- (0, 0, 1) -- cycle;
        \fill[fill=facecolor, draw=black, shift={(0,0,0)}]
        (0, 0, 0) -- (1, 0, 0) -- (1, 1, 0) -- (0, 1, 0) -- cycle;
        \fill[fill=facecolor, draw=black, shift={(1,0,0)}]
        (0, 0, 0) -- (0, 0, 1) -- (1, 0, 1) -- (1, 0, 0) -- cycle;
        \end{tikzpicture}} \\
$2$ & \myvcenter{%
\begin{tikzpicture}
[x={(-0.173205cm,-0.100000cm)}, y={(0.173205cm,-0.100000cm)}, z={(0.000000cm,0.200000cm)}]
\definecolor{facecolor}{rgb}{0.8,0.8,0.8}
\fill[fill=facecolor, draw=black, shift={(0,0,0)}]
(0, 0, 0) -- (0, 1, 0) -- (0, 1, 1) -- (0, 0, 1) -- cycle;
\fill[fill=facecolor, draw=black, shift={(0,1,0)}]
(0, 0, 0) -- (0, 1, 0) -- (0, 1, 1) -- (0, 0, 1) -- cycle;
\end{tikzpicture}\,}%
    & \myvcenter{%
\begin{tikzpicture}
[x={(-0.173205cm,-0.100000cm)}, y={(0.173205cm,-0.100000cm)}, z={(0.000000cm,0.200000cm)}]
\definecolor{facecolor}{rgb}{0.8,0.8,0.8}
\fill[fill=facecolor, draw=black, shift={(0,0,0)}]
(0, 0, 0) -- (0, 0, 1) -- (1, 0, 1) -- (1, 0, 0) -- cycle;
\fill[fill=facecolor, draw=black, shift={(1,0,0)}]
(0, 0, 0) -- (0, 0, 1) -- (1, 0, 1) -- (1, 0, 0) -- cycle;
\end{tikzpicture}\,}%
 \\
$3$ & \myvcenter{%
\begin{tikzpicture}
[x={(-0.173205cm,-0.100000cm)}, y={(0.173205cm,-0.100000cm)}, z={(0.000000cm,0.200000cm)}]
\definecolor{facecolor}{rgb}{0.8,0.8,0.8}
\fill[fill=facecolor, draw=black, shift={(0,0,0)}]
(0, 0, 0) -- (0, 1, 0) -- (0, 1, 1) -- (0, 0, 1) -- cycle;
\definecolor{facecolor}{rgb}{0.8,0.8,0.8}
\fill[fill=facecolor, draw=black, shift={(0,0,0)}]
(0, 0, 0) -- (1, 0, 0) -- (1, 1, 0) -- (0, 1, 0) -- cycle;
\end{tikzpicture}\,}%
    & \myvcenter{%
        \begin{tikzpicture}
        [x={(-0.173205cm,-0.100000cm)}, y={(0.173205cm,-0.100000cm)}, z={(0.000000cm,0.200000cm)}]
        \definecolor{facecolor}{rgb}{0.800,0.800,0.800}
        \fill[fill=facecolor, draw=black, shift={(0,0,0)}]
        (0, 0, 0) -- (0, 0, 1) -- (1, 0, 1) -- (1, 0, 0) -- cycle;
        \fill[fill=facecolor, draw=black, shift={(0,0,0)}]
        (0, 0, 0) -- (0, 1, 0) -- (0, 1, 1) -- (0, 0, 1) -- cycle;
        \fill[fill=facecolor, draw=black, shift={(1,0,0)}]
        (0, 0, 0) -- (0, 0, 1) -- (1, 0, 1) -- (1, 0, 0) -- cycle;
        \fill[fill=facecolor, draw=black, shift={(0,0,0)}]
        (0, 0, 0) -- (1, 0, 0) -- (1, 1, 0) -- (0, 1, 0) -- cycle;
        \end{tikzpicture}} \\
$3$ & \myvcenter{%
\begin{tikzpicture}
[x={(-0.173205cm,-0.100000cm)}, y={(0.173205cm,-0.100000cm)}, z={(0.000000cm,0.200000cm)}]
\definecolor{facecolor}{rgb}{0.8,0.8,0.8}
\fill[fill=facecolor, draw=black, shift={(0,0,0)}]
(0, 0, 0) -- (0, 1, 0) -- (0, 1, 1) -- (0, 0, 1) -- cycle;
\fill[fill=facecolor, draw=black, shift={(0,0,1)}]
(0, 0, 0) -- (0, 1, 0) -- (0, 1, 1) -- (0, 0, 1) -- cycle;
\end{tikzpicture}\,}%
    & \myvcenter{%
\begin{tikzpicture}
[x={(-0.173205cm,-0.100000cm)}, y={(0.173205cm,-0.100000cm)}, z={(0.000000cm,0.200000cm)}]
\definecolor{facecolor}{rgb}{0.8,0.8,0.8}
\fill[fill=facecolor, draw=black, shift={(0,0,0)}]
(0, 0, 0) -- (0, 0, 1) -- (1, 0, 1) -- (1, 0, 0) -- cycle;
\fill[fill=facecolor, draw=black, shift={(1,0,0)}]
(0, 0, 0) -- (0, 0, 1) -- (1, 0, 1) -- (1, 0, 0) -- cycle;
\end{tikzpicture}\,}%
\end{tabular}
\qquad
\begin{tabular}[h]{c|c|c}
$i$ & $Q$ & $\Sigma_i(Q)$ \\
\hline
$1$ & \myvcenter{%
\begin{tikzpicture}
[x={(-0.173205cm,-0.100000cm)}, y={(0.173205cm,-0.100000cm)}, z={(0.000000cm,0.200000cm)}]
\definecolor{facecolor}{rgb}{0.8,0.8,0.8}
\fill[fill=facecolor, draw=black, shift={(0,0,0)}]
(0, 0, 0) -- (0, 1, 0) -- (0, 1, 1) -- (0, 0, 1) -- cycle;
\definecolor{facecolor}{rgb}{0.8,0.8,0.8}
\fill[fill=facecolor, draw=black, shift={(0,0,0)}]
(0, 0, 0) -- (1, 0, 0) -- (1, 1, 0) -- (0, 1, 0) -- cycle;
\end{tikzpicture}\,}%
    & \myvcenter{%
        \begin{tikzpicture}
        [x={(-0.173205cm,-0.100000cm)}, y={(0.173205cm,-0.100000cm)}, z={(0.000000cm,0.200000cm)}]
        \definecolor{facecolor}{rgb}{0.800,0.800,0.800}
        \fill[fill=facecolor, draw=black, shift={(1,0,0)}]
        (0, 0, 0) -- (1, 0, 0) -- (1, 1, 0) -- (0, 1, 0) -- cycle;
        \fill[fill=facecolor, draw=black, shift={(0,0,0)}]
        (0, 0, 0) -- (0, 1, 0) -- (0, 1, 1) -- (0, 0, 1) -- cycle;
        \fill[fill=facecolor, draw=black, shift={(0,0,0)}]
        (0, 0, 0) -- (0, 0, 1) -- (1, 0, 1) -- (1, 0, 0) -- cycle;
        \fill[fill=facecolor, draw=black, shift={(0,0,0)}]
        (0, 0, 0) -- (1, 0, 0) -- (1, 1, 0) -- (0, 1, 0) -- cycle;
        \end{tikzpicture}} \\
$1$ & \myvcenter{%
\begin{tikzpicture}
[x={(-0.173205cm,-0.100000cm)}, y={(0.173205cm,-0.100000cm)}, z={(0.000000cm,0.200000cm)}]
\definecolor{facecolor}{rgb}{0.8,0.8,0.8}
\fill[fill=facecolor, draw=black, shift={(1,0,0)}]
(0, 0, 0) -- (1, 0, 0) -- (1, 1, 0) -- (0, 1, 0) -- cycle;
\fill[fill=facecolor, draw=black, shift={(0,0,0)}]
(0, 0, 0) -- (1, 0, 0) -- (1, 1, 0) -- (0, 1, 0) -- cycle;
\end{tikzpicture}\,}%
    & \myvcenter{%
\begin{tikzpicture}
[x={(-0.173205cm,-0.100000cm)}, y={(0.173205cm,-0.100000cm)}, z={(0.000000cm,0.200000cm)}]
\definecolor{facecolor}{rgb}{0.8,0.8,0.8}
\fill[fill=facecolor, draw=black, shift={(1,0,0)}]
(0, 0, 0) -- (1, 0, 0) -- (1, 1, 0) -- (0, 1, 0) -- cycle;
\fill[fill=facecolor, draw=black, shift={(0,0,0)}]
(0, 0, 0) -- (1, 0, 0) -- (1, 1, 0) -- (0, 1, 0) -- cycle;
\end{tikzpicture}\,}%
 \\
$2$ & \myvcenter{%
\begin{tikzpicture}
[x={(-0.173205cm,-0.100000cm)}, y={(0.173205cm,-0.100000cm)}, z={(0.000000cm,0.200000cm)}]
\definecolor{facecolor}{rgb}{0.8,0.8,0.8}
\fill[fill=facecolor, draw=black, shift={(0,0,0)}]
(0, 0, 0) -- (0, 0, 1) -- (1, 0, 1) -- (1, 0, 0) -- cycle;
\definecolor{facecolor}{rgb}{0.8,0.8,0.8}
\fill[fill=facecolor, draw=black, shift={(0,0,0)}]
(0, 0, 0) -- (1, 0, 0) -- (1, 1, 0) -- (0, 1, 0) -- cycle;
\end{tikzpicture}\,}%
    & \myvcenter{%
        \begin{tikzpicture}
        [x={(-0.173205cm,-0.100000cm)}, y={(0.173205cm,-0.100000cm)}, z={(0.000000cm,0.200000cm)}]
        \definecolor{facecolor}{rgb}{0.800,0.800,0.800}
        \fill[fill=facecolor, draw=black, shift={(1,0,0)}]
        (0, 0, 0) -- (1, 0, 0) -- (1, 1, 0) -- (0, 1, 0) -- cycle;
        \fill[fill=facecolor, draw=black, shift={(0,0,0)}]
        (0, 0, 0) -- (0, 1, 0) -- (0, 1, 1) -- (0, 0, 1) -- cycle;
        \fill[fill=facecolor, draw=black, shift={(0,0,0)}]
        (0, 0, 0) -- (0, 0, 1) -- (1, 0, 1) -- (1, 0, 0) -- cycle;
        \fill[fill=facecolor, draw=black, shift={(0,0,0)}]
        (0, 0, 0) -- (1, 0, 0) -- (1, 1, 0) -- (0, 1, 0) -- cycle;
        \end{tikzpicture}} \\
$2$ & \myvcenter{%
\begin{tikzpicture}
[x={(-0.173205cm,-0.100000cm)}, y={(0.173205cm,-0.100000cm)}, z={(0.000000cm,0.200000cm)}]
\definecolor{facecolor}{rgb}{0.8,0.8,0.8}
\fill[fill=facecolor, draw=black, shift={(0,1,0)}]
(0, 0, 0) -- (1, 0, 0) -- (1, 1, 0) -- (0, 1, 0) -- cycle;
\fill[fill=facecolor, draw=black, shift={(0,0,0)}]
(0, 0, 0) -- (1, 0, 0) -- (1, 1, 0) -- (0, 1, 0) -- cycle;
\end{tikzpicture}\,}%
    & \myvcenter{%
\begin{tikzpicture}
[x={(-0.173205cm,-0.100000cm)}, y={(0.173205cm,-0.100000cm)}, z={(0.000000cm,0.200000cm)}]
\definecolor{facecolor}{rgb}{0.8,0.8,0.8}
\fill[fill=facecolor, draw=black, shift={(1,0,0)}]
(0, 0, 0) -- (1, 0, 0) -- (1, 1, 0) -- (0, 1, 0) -- cycle;
\fill[fill=facecolor, draw=black, shift={(0,0,0)}]
(0, 0, 0) -- (1, 0, 0) -- (1, 1, 0) -- (0, 1, 0) -- cycle;
\end{tikzpicture}\,}%
 \\
$3$ & \myvcenter{%
\begin{tikzpicture}
[x={(-0.173205cm,-0.100000cm)}, y={(0.173205cm,-0.100000cm)}, z={(0.000000cm,0.200000cm)}]
\definecolor{facecolor}{rgb}{0.8,0.8,0.8}
\fill[fill=facecolor, draw=black, shift={(0,0,0)}]
(0, 0, 0) -- (0, 0, 1) -- (1, 0, 1) -- (1, 0, 0) -- cycle;
\definecolor{facecolor}{rgb}{0.8,0.8,0.8}
\fill[fill=facecolor, draw=black, shift={(0,0,0)}]
(0, 0, 0) -- (1, 0, 0) -- (1, 1, 0) -- (0, 1, 0) -- cycle;
\end{tikzpicture}\,}%
    & \myvcenter{%
        \begin{tikzpicture}
        [x={(-0.173205cm,-0.100000cm)}, y={(0.173205cm,-0.100000cm)}, z={(0.000000cm,0.200000cm)}]
        \definecolor{facecolor}{rgb}{0.800,0.800,0.800}
        \fill[fill=facecolor, draw=black, shift={(1,0,0)}]
        (0, 0, 0) -- (1, 0, 0) -- (1, 1, 0) -- (0, 1, 0) -- cycle;
        \fill[fill=facecolor, draw=black, shift={(0,0,0)}]
        (0, 0, 0) -- (0, 1, 0) -- (0, 1, 1) -- (0, 0, 1) -- cycle;
        \fill[fill=facecolor, draw=black, shift={(0,0,0)}]
        (0, 0, 0) -- (0, 0, 1) -- (1, 0, 1) -- (1, 0, 0) -- cycle;
        \fill[fill=facecolor, draw=black, shift={(0,0,0)}]
        (0, 0, 0) -- (1, 0, 0) -- (1, 1, 0) -- (0, 1, 0) -- cycle;
        \end{tikzpicture}} \\
$3$ & \myvcenter{%
\begin{tikzpicture}
[x={(-0.173205cm,-0.100000cm)}, y={(0.173205cm,-0.100000cm)}, z={(0.000000cm,0.200000cm)}]
\definecolor{facecolor}{rgb}{0.8,0.8,0.8}
\fill[fill=facecolor, draw=black, shift={(0,0,1)}]
(0, 0, 0) -- (0, 0, 1) -- (1, 0, 1) -- (1, 0, 0) -- cycle;
\fill[fill=facecolor, draw=black, shift={(0,0,0)}]
(0, 0, 0) -- (0, 0, 1) -- (1, 0, 1) -- (1, 0, 0) -- cycle;
\end{tikzpicture}\,}%
    & \myvcenter{%
\begin{tikzpicture}
[x={(-0.173205cm,-0.100000cm)}, y={(0.173205cm,-0.100000cm)}, z={(0.000000cm,0.200000cm)}]
\definecolor{facecolor}{rgb}{0.8,0.8,0.8}
\fill[fill=facecolor, draw=black, shift={(1,0,0)}]
(0, 0, 0) -- (1, 0, 0) -- (1, 1, 0) -- (0, 1, 0) -- cycle;
\fill[fill=facecolor, draw=black, shift={(0,0,0)}]
(0, 0, 0) -- (1, 0, 0) -- (1, 1, 0) -- (0, 1, 0) -- cycle;
\end{tikzpicture}\,}%
\end{tabular}
\qquad
\begin{tabular}[h]{c|c|c}
$i$ & $Q$ & $\Sigma_i(Q)$ \\
\hline
$1$ & \myvcenter{%
        \begin{tikzpicture}
        [x={(-0.173205cm,-0.100000cm)}, y={(0.173205cm,-0.100000cm)}, z={(0.000000cm,0.200000cm)}]
        \definecolor{facecolor}{rgb}{0.800,0.800,0.800}
        \fill[fill=facecolor, draw=black, shift={(0,0,0)}]
        (0, 0, 0) -- (0, 1, 0) -- (0, 1, 1) -- (0, 0, 1) -- cycle;
        \fill[fill=facecolor, draw=black, shift={(0,1,0)}]
        (0, 0, 0) -- (1, 0, 0) -- (1, 1, 0) -- (0, 1, 0) -- cycle;
        \end{tikzpicture}}
    & \myvcenter{%
        \begin{tikzpicture}
        [x={(-0.173205cm,-0.100000cm)}, y={(0.173205cm,-0.100000cm)}, z={(0.000000cm,0.200000cm)}]
        \definecolor{facecolor}{rgb}{0.800,0.800,0.800}
        \fill[fill=facecolor, draw=black, shift={(0,0,0)}]
        (0, 0, 0) -- (0, 0, 1) -- (1, 0, 1) -- (1, 0, 0) -- cycle;
        \fill[fill=facecolor, draw=black, shift={(0,0,0)}]
        (0, 0, 0) -- (0, 1, 0) -- (0, 1, 1) -- (0, 0, 1) -- cycle;
        \fill[fill=facecolor, draw=black, shift={(0,1,0)}]
        (0, 0, 0) -- (1, 0, 0) -- (1, 1, 0) -- (0, 1, 0) -- cycle;
        \fill[fill=facecolor, draw=black, shift={(0,0,0)}]
        (0, 0, 0) -- (1, 0, 0) -- (1, 1, 0) -- (0, 1, 0) -- cycle;
        \end{tikzpicture}} \\
$1$ & \myvcenter{%
\begin{tikzpicture}
[x={(-0.173205cm,-0.100000cm)}, y={(0.173205cm,-0.100000cm)}, z={(0.000000cm,0.200000cm)}]
\definecolor{facecolor}{rgb}{0.8,0.8,0.8}
\fill[fill=facecolor, draw=black, shift={(0,0,0)}]
(0, 0, 0) -- (1, 0, 0) -- (1, 1, 0) -- (0, 1, 0) -- cycle;
\fill[fill=facecolor, draw=black, shift={(1,1,0)}]
(0, 0, 0) -- (1, 0, 0) -- (1, 1, 0) -- (0, 1, 0) -- cycle;
\end{tikzpicture}\,}%
    & \myvcenter{%
\begin{tikzpicture}
[x={(-0.173205cm,-0.100000cm)}, y={(0.173205cm,-0.100000cm)}, z={(0.000000cm,0.200000cm)}]
\definecolor{facecolor}{rgb}{0.8,0.8,0.8}
\fill[fill=facecolor, draw=black, shift={(0,1,0)}]
(0, 0, 0) -- (1, 0, 0) -- (1, 1, 0) -- (0, 1, 0) -- cycle;
\fill[fill=facecolor, draw=black, shift={(0,0,0)}]
(0, 0, 0) -- (1, 0, 0) -- (1, 1, 0) -- (0, 1, 0) -- cycle;
\end{tikzpicture}\,}%
 \\
$2$ & \myvcenter{%
        \begin{tikzpicture}
        [x={(-0.173205cm,-0.100000cm)}, y={(0.173205cm,-0.100000cm)}, z={(0.000000cm,0.200000cm)}]
        \definecolor{facecolor}{rgb}{0.800,0.800,0.800}
        \fill[fill=facecolor, draw=black, shift={(1,0,0)}]
        (0, 0, 0) -- (1, 0, 0) -- (1, 1, 0) -- (0, 1, 0) -- cycle;
        \fill[fill=facecolor, draw=black, shift={(0,0,0)}]
        (0, 0, 0) -- (0, 0, 1) -- (1, 0, 1) -- (1, 0, 0) -- cycle;
        \end{tikzpicture}}
    & \myvcenter{%
        \begin{tikzpicture}
        [x={(-0.173205cm,-0.100000cm)}, y={(0.173205cm,-0.100000cm)}, z={(0.000000cm,0.200000cm)}]
        \definecolor{facecolor}{rgb}{0.800,0.800,0.800}
        \fill[fill=facecolor, draw=black, shift={(0,0,0)}]
        (0, 0, 0) -- (0, 0, 1) -- (1, 0, 1) -- (1, 0, 0) -- cycle;
        \fill[fill=facecolor, draw=black, shift={(0,0,0)}]
        (0, 0, 0) -- (0, 1, 0) -- (0, 1, 1) -- (0, 0, 1) -- cycle;
        \fill[fill=facecolor, draw=black, shift={(0,1,0)}]
        (0, 0, 0) -- (1, 0, 0) -- (1, 1, 0) -- (0, 1, 0) -- cycle;
        \fill[fill=facecolor, draw=black, shift={(0,0,0)}]
        (0, 0, 0) -- (1, 0, 0) -- (1, 1, 0) -- (0, 1, 0) -- cycle;
        \end{tikzpicture}} \\
$2$ & \myvcenter{%
\begin{tikzpicture}
[x={(-0.173205cm,-0.100000cm)}, y={(0.173205cm,-0.100000cm)}, z={(0.000000cm,0.200000cm)}]
\definecolor{facecolor}{rgb}{0.8,0.8,0.8}
\fill[fill=facecolor, draw=black, shift={(0,0,0)}]
(0, 0, 0) -- (1, 0, 0) -- (1, 1, 0) -- (0, 1, 0) -- cycle;
\fill[fill=facecolor, draw=black, shift={(1,1,0)}]
(0, 0, 0) -- (1, 0, 0) -- (1, 1, 0) -- (0, 1, 0) -- cycle;
\end{tikzpicture}\,}%
    & \myvcenter{%
\begin{tikzpicture}
[x={(-0.173205cm,-0.100000cm)}, y={(0.173205cm,-0.100000cm)}, z={(0.000000cm,0.200000cm)}]
\definecolor{facecolor}{rgb}{0.8,0.8,0.8}
\fill[fill=facecolor, draw=black, shift={(0,1,0)}]
(0, 0, 0) -- (1, 0, 0) -- (1, 1, 0) -- (0, 1, 0) -- cycle;
\fill[fill=facecolor, draw=black, shift={(0,0,0)}]
(0, 0, 0) -- (1, 0, 0) -- (1, 1, 0) -- (0, 1, 0) -- cycle;
\end{tikzpicture}\,}%
 \\
$3$ & \myvcenter{%
        \begin{tikzpicture}
        [x={(-0.173205cm,-0.100000cm)}, y={(0.173205cm,-0.100000cm)}, z={(0.000000cm,0.200000cm)}]
        \definecolor{facecolor}{rgb}{0.800,0.800,0.800}
        \fill[fill=facecolor, draw=black, shift={(1,0,0)}]
        (0, 0, 0) -- (0, 0, 1) -- (1, 0, 1) -- (1, 0, 0) -- cycle;
        \fill[fill=facecolor, draw=black, shift={(0,0,0)}]
        (0, 0, 0) -- (1, 0, 0) -- (1, 1, 0) -- (0, 1, 0) -- cycle;
        \end{tikzpicture}}
    & \myvcenter{%
        \begin{tikzpicture}
        [x={(-0.173205cm,-0.100000cm)}, y={(0.173205cm,-0.100000cm)}, z={(0.000000cm,0.200000cm)}]
        \definecolor{facecolor}{rgb}{0.800,0.800,0.800}
        \fill[fill=facecolor, draw=black, shift={(0,0,0)}]
        (0, 0, 0) -- (0, 0, 1) -- (1, 0, 1) -- (1, 0, 0) -- cycle;
        \fill[fill=facecolor, draw=black, shift={(0,0,0)}]
        (0, 0, 0) -- (0, 1, 0) -- (0, 1, 1) -- (0, 0, 1) -- cycle;
        \fill[fill=facecolor, draw=black, shift={(0,1,0)}]
        (0, 0, 0) -- (1, 0, 0) -- (1, 1, 0) -- (0, 1, 0) -- cycle;
        \fill[fill=facecolor, draw=black, shift={(0,0,0)}]
        (0, 0, 0) -- (1, 0, 0) -- (1, 1, 0) -- (0, 1, 0) -- cycle;
        \end{tikzpicture}} \\
$3$ & \myvcenter{%
\begin{tikzpicture}
[x={(-0.173205cm,-0.100000cm)}, y={(0.173205cm,-0.100000cm)}, z={(0.000000cm,0.200000cm)}]
\definecolor{facecolor}{rgb}{0.8,0.8,0.8}
\fill[fill=facecolor, draw=black, shift={(0,0,0)}]
(0, 0, 0) -- (0, 0, 1) -- (1, 0, 1) -- (1, 0, 0) -- cycle;
\fill[fill=facecolor, draw=black, shift={(1,0,1)}]
(0, 0, 0) -- (0, 0, 1) -- (1, 0, 1) -- (1, 0, 0) -- cycle;
\end{tikzpicture}\,}%
    & \myvcenter{%
\begin{tikzpicture}
[x={(-0.173205cm,-0.100000cm)}, y={(0.173205cm,-0.100000cm)}, z={(0.000000cm,0.200000cm)}]
\definecolor{facecolor}{rgb}{0.8,0.8,0.8}
\fill[fill=facecolor, draw=black, shift={(0,1,0)}]
(0, 0, 0) -- (1, 0, 0) -- (1, 1, 0) -- (0, 1, 0) -- cycle;
\fill[fill=facecolor, draw=black, shift={(0,0,0)}]
(0, 0, 0) -- (1, 0, 0) -- (1, 1, 0) -- (0, 1, 0) -- cycle;
\end{tikzpicture}\,}%
\end{tabular}
\end{center}
The case of the first, third and fifth rows of each table is settled:
\begin{itemize}
\item
if $X = \myvcenter{%
\begin{tikzpicture}
[x={(-0.173205cm,-0.100000cm)}, y={(0.173205cm,-0.100000cm)}, z={(0.000000cm,0.200000cm)}]
\definecolor{facecolor}{rgb}{0.8,0.8,0.8}
\fill[fill=facecolor, draw=black, shift={(0,0,0)}]
(0, 0, 0) -- (0, 0, 1) -- (1, 0, 1) -- (1, 0, 0) -- cycle;
\fill[fill=facecolor, draw=black, shift={(1,0,0)}]
(0, 0, 0) -- (0, 0, 1) -- (1, 0, 1) -- (1, 0, 0) -- cycle;
\end{tikzpicture}\,}%
$, then $Y = \myvcenter{%
\begin{tikzpicture}
[x={(-0.173205cm,-0.100000cm)}, y={(0.173205cm,-0.100000cm)}, z={(0.000000cm,0.200000cm)}]
\definecolor{facecolor}{rgb}{0.800,0.800,0.800}
\fill[fill=facecolor, draw=black, shift={(0,0,0)}]
(0, 0, 0) -- (0, 0, 1) -- (1, 0, 1) -- (1, 0, 0) -- cycle;
\fill[fill=facecolor, draw=black, shift={(0,0,0)}]
(0, 0, 0) -- (1, 0, 0) -- (1, 1, 0) -- (0, 1, 0) -- cycle;
\fill[fill=facecolor, draw=black, shift={(1,0,0)}]
(0, 0, 0) -- (0, 0, 1) -- (1, 0, 1) -- (1, 0, 0) -- cycle;
\end{tikzpicture}\,}%
 \in \LFS$ works;
\item
if $X = \myvcenter{%
\begin{tikzpicture}
[x={(-0.173205cm,-0.100000cm)}, y={(0.173205cm,-0.100000cm)}, z={(0.000000cm,0.200000cm)}]
\definecolor{facecolor}{rgb}{0.8,0.8,0.8}
\fill[fill=facecolor, draw=black, shift={(1,0,0)}]
(0, 0, 0) -- (1, 0, 0) -- (1, 1, 0) -- (0, 1, 0) -- cycle;
\fill[fill=facecolor, draw=black, shift={(0,0,0)}]
(0, 0, 0) -- (1, 0, 0) -- (1, 1, 0) -- (0, 1, 0) -- cycle;
\end{tikzpicture}\,}%
$, then $Y = \myvcenter{%
\begin{tikzpicture}
[x={(-0.173205cm,-0.100000cm)}, y={(0.173205cm,-0.100000cm)}, z={(0.000000cm,0.200000cm)}]
\definecolor{facecolor}{rgb}{0.800,0.800,0.800}
\fill[fill=facecolor, draw=black, shift={(0,0,0)}]
(0, 0, 0) -- (0, 0, 1) -- (1, 0, 1) -- (1, 0, 0) -- cycle;
\fill[fill=facecolor, draw=black, shift={(1,0,0)}]
(0, 0, 0) -- (1, 0, 0) -- (1, 1, 0) -- (0, 1, 0) -- cycle;
\fill[fill=facecolor, draw=black, shift={(0,0,0)}]
(0, 0, 0) -- (1, 0, 0) -- (1, 1, 0) -- (0, 1, 0) -- cycle;
\end{tikzpicture}\,}%
 \in \LFS$ works;
\item
if $X = \myvcenter{%
\begin{tikzpicture}
[x={(-0.173205cm,-0.100000cm)}, y={(0.173205cm,-0.100000cm)}, z={(0.000000cm,0.200000cm)}]
\definecolor{facecolor}{rgb}{0.8,0.8,0.8}
\fill[fill=facecolor, draw=black, shift={(0,1,0)}]
(0, 0, 0) -- (1, 0, 0) -- (1, 1, 0) -- (0, 1, 0) -- cycle;
\fill[fill=facecolor, draw=black, shift={(0,0,0)}]
(0, 0, 0) -- (1, 0, 0) -- (1, 1, 0) -- (0, 1, 0) -- cycle;
\end{tikzpicture}\,}%
$, then $Y = \myvcenter{%
\begin{tikzpicture}
[x={(-0.173205cm,-0.100000cm)}, y={(0.173205cm,-0.100000cm)}, z={(0.000000cm,0.200000cm)}]
\definecolor{facecolor}{rgb}{0.800,0.800,0.800}
\fill[fill=facecolor, draw=black, shift={(0,0,0)}]
(0, 0, 0) -- (0, 1, 0) -- (0, 1, 1) -- (0, 0, 1) -- cycle;
\fill[fill=facecolor, draw=black, shift={(0,1,0)}]
(0, 0, 0) -- (1, 0, 0) -- (1, 1, 0) -- (0, 1, 0) -- cycle;
\fill[fill=facecolor, draw=black, shift={(0,0,0)}]
(0, 0, 0) -- (1, 0, 0) -- (1, 1, 0) -- (0, 1, 0) -- cycle;
\end{tikzpicture}\,}%
 \in \LFS$ works.
\end{itemize}
For the second row of the first table with $Q = \myvcenter{%
\begin{tikzpicture}
[x={(-0.173205cm,-0.100000cm)}, y={(0.173205cm,-0.100000cm)}, z={(0.000000cm,0.200000cm)}]
\definecolor{facecolor}{rgb}{0.8,0.8,0.8}
\fill[fill=facecolor, draw=black, shift={(0,0,0)}]
(0, 0, 0) -- (0, 1, 0) -- (0, 1, 1) -- (0, 0, 1) -- cycle;
\fill[fill=facecolor, draw=black, shift={(0,1,0)}]
(0, 0, 0) -- (0, 1, 0) -- (0, 1, 1) -- (0, 0, 1) -- cycle;
\end{tikzpicture}\,}%
$,
$P$ is strongly $\LFS$-covered
so we have $Q \subseteq Y_0 \subseteq P$
with $Y_0 = \myvcenter{%
\begin{tikzpicture}
[x={(-0.173205cm,-0.100000cm)}, y={(0.173205cm,-0.100000cm)}, z={(0.000000cm,0.200000cm)}]
\definecolor{facecolor}{rgb}{0.800,0.800,0.800}
\fill[fill=facecolor, draw=black, shift={(0,0,0)}]
(0, 0, 0) -- (0, 0, 1) -- (1, 0, 1) -- (1, 0, 0) -- cycle;
\fill[fill=facecolor, draw=black, shift={(0,0,0)}]
(0, 0, 0) -- (1, 0, 0) -- (1, 1, 0) -- (0, 1, 0) -- cycle;
\fill[fill=facecolor, draw=black, shift={(1,0,0)}]
(0, 0, 0) -- (0, 0, 1) -- (1, 0, 1) -- (1, 0, 0) -- cycle;
\end{tikzpicture}\,}%
 \in \LFS$.
It follows that $X \subseteq \Sigma_1(Y_0) = \myvcenter{%
\begin{tikzpicture}
[x={(-0.173205cm,-0.100000cm)}, y={(0.173205cm,-0.100000cm)}, z={(0.000000cm,0.200000cm)}]
\definecolor{facecolor}{rgb}{0.800,0.800,0.800}
\fill[fill=facecolor, draw=black, shift={(0,0,0)}]
(0, 0, 0) -- (0, 0, 1) -- (1, 0, 1) -- (1, 0, 0) -- cycle;
\fill[fill=facecolor, draw=black, shift={(0,0,0)}]
(0, 0, 0) -- (1, 0, 0) -- (1, 1, 0) -- (0, 1, 0) -- cycle;
\fill[fill=facecolor, draw=black, shift={(1,0,0)}]
(0, 0, 0) -- (0, 0, 1) -- (1, 0, 1) -- (1, 0, 0) -- cycle;
\end{tikzpicture}\,}%
$,
so taking $Y = \Sigma_1(Y_0) \in \LFS$ works.
The cases of the second and fourth rows of the second table can be dealt with similarly.

In the last row of the third table with we have $Q = \myvcenter{%
\begin{tikzpicture}
[x={(-0.173205cm,-0.100000cm)}, y={(0.173205cm,-0.100000cm)}, z={(0.000000cm,0.200000cm)}]
\definecolor{facecolor}{rgb}{0.8,0.8,0.8}
\fill[fill=facecolor, draw=black, shift={(0,0,0)}]
(0, 0, 0) -- (0, 0, 1) -- (1, 0, 1) -- (1, 0, 0) -- cycle;
\fill[fill=facecolor, draw=black, shift={(1,0,1)}]
(0, 0, 0) -- (0, 0, 1) -- (1, 0, 1) -- (1, 0, 0) -- cycle;
\end{tikzpicture}\,}%
$,
so \myvcenter{%
\begin{tikzpicture}
[x={(-0.173205cm,-0.100000cm)}, y={(0.173205cm,-0.100000cm)}, z={(0.000000cm,0.200000cm)}]
\definecolor{facecolor}{rgb}{0.8,0.8,0.8}
\fill[fill=facecolor, draw=black, shift={(0,0,1)}]
(0, 0, 0) -- (0, 0, 1) -- (1, 0, 1) -- (1, 0, 0) -- cycle;
\fill[fill=facecolor, draw=black, shift={(0,0,0)}]
(0, 0, 0) -- (0, 0, 1) -- (1, 0, 1) -- (1, 0, 0) -- cycle;
\end{tikzpicture}\,}%
 must appear in $\bfGa$,
which is forbidden by assumption.
This can be seen by using Definition~\ref{def:stepped}
to compute the only possible ``completion'' of $Q$ within $\bfGa$ (shown in dark gray):
\myvcenter{\begin{tikzpicture}
[x={(-0.216506cm,-0.125000cm)}, y={(0.216506cm,-0.125000cm)}, z={(0.000000cm,0.250000cm)}]
\definecolor{facecolor}{rgb}{0.000,1.000,0.000}
\definecolor{facecolor}{rgb}{0.800,0.800,0.800}
\fill[fill=facecolor, draw=black, shift={(0,0,0)}]
(0, 0, 0) -- (0, 0, 1) -- (1, 0, 1) -- (1, 0, 0) -- cycle;
\fill[fill=facecolor, draw=black, shift={(1,0,1)}]
(0, 0, 0) -- (0, 0, 1) -- (1, 0, 1) -- (1, 0, 0) -- cycle;
\definecolor{facecolor}{rgb}{0.3500,0.3500,0.3500}
\fill[fill=facecolor, draw=black, shift={(0,0,1)}]
(0, 0, 0) -- (0, 0, 1) -- (1, 0, 1) -- (1, 0, 0) -- cycle;
\end{tikzpicture}}\,.
In all the remaining cases, $Q$ is a pattern which is not allowed in $\bfGa$ by assumption,
so they can be ignored.
\end{proof}

\subsection{Annuli and dual substitutions}
\label{subsec:annulus}

The proof of the following proposition (by induction) relies on Lemma~\ref{lem:annulus_base} (base case) and on Lemma~\ref{lem:annulus_induction} (induction step).
We recall that
$\mcU = [\mathbf 0,1]^\star \cup [\mathbf 0,2]^\star \cup [\mathbf 0,3]^\star =  \myvcenter{%
    \begin{tikzpicture}
    [x={(-0.173205cm,-0.100000cm)}, y={(0.173205cm,-0.100000cm)}, z={(0.000000cm,0.200000cm)}]
    \fill[fill=facecolor, draw=black, shift={(0,0,0)}]
    (0, 0, 0) -- (0, 1, 0) -- (0, 1, 1) -- (0, 0, 1) -- cycle;
    \fill[fill=facecolor, draw=black, shift={(0,0,0)}]
    (0, 0, 0) -- (0, 0, 1) -- (1, 0, 1) -- (1, 0, 0) -- cycle;
    \fill[fill=facecolor, draw=black, shift={(0,0,0)}]
    (0, 0, 0) -- (1, 0, 0) -- (1, 1, 0) -- (0, 1, 0) -- cycle;
    \node[circle,fill=black,draw=black,minimum size=1mm,inner sep=0pt] at (0,0,0) {};
\end{tikzpicture}}$.

\begin{prop}
\label{prop:annulus}
Let $(\Sigma_i)_{n\in\bbN}$ be a sequence with values in $\{\SFSi,\SFSii,\SFSiii\}$
such that $\SFSiii$ occurs infinitely often,
and let $k \in \bbN$ be such that $(\Sigma_1, \ldots, \Sigma_k)$ contains $\SFSiii$ at least four times.

Then for every $\ell \geq 1$,
$\Sigma_1 \cdots\Sigma_{k+\ell}(\mcU) \setminus \Sigma_1 \cdots \Sigma_\ell(\mcU)$
is an $\LFS$-annulus of $\Sigma_1 \cdots \Sigma_\ell(\mcU)$
in the stepped plane $\Sigma_1 \cdots \Sigma_{k+\ell}(\bfGa_{(1,1,1)})$.
\end{prop}

\begin{proof}
We prove the result by induction on $\ell$.
The case $\ell=0$ (\emph{i.e.}, $\Sigma_1 \cdots\Sigma_k(\mcU) \setminus \mcU$
is an annulus of $\mcU$) is settled by Lemma~\ref{lem:annulus_base}.
Now, assume that the induction property holds for some $\ell \in \bbN$.
The pattern $\Sigma_1 \cdots \Sigma_{k + \ell}(\mcU)$
is contained in the stepped plane $\Sigma_{k + \ell}(\bfGa_{(1,1,1)})$,
so it does not contain any of the patterns forbidden by Lemma~\ref{lemm:forbidden}.
We can then apply Lemma~\ref{lem:annulus_induction}
to deduce that $\Sigma_1 \cdots \Sigma_{k + \ell + 1}(\mcU) \setminus \Sigma_1 \cdots \Sigma_{\ell + 1}(\mcU)$
is an $\LFS$-annulus of $\Sigma_1 \cdots \Sigma_{\ell + 1}(\mcU)$.
\end{proof}

\begin{lemm}
\label{lem:annulus_base}
Let $\Sigma$ be a product of $\Sigma_1$, $\Sigma_2$ and $\Sigma_3$
such that $\Sigma_3$ appears at least four times.
Then $\Sigma(\mcU) \setminus \mcU$ is an $\LFS$-annulus of $\mcU$
in $\Sigma(\bfGa_{(1,1,1)})$.
\end{lemm}

\begin{proof}
Below, ``$P \stackrel{i}{\rightarrow} Q$'' means that
$Q \subseteq \Sigma_i(P)$ so the result follows.
\vspace{-0.75em}
\begin{center}
\begin{tikzpicture}[x=1cm,y=0.85cm]
\node at (0,0) (0) {\includegraphics[scale=0.85]{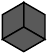}};
\node at (2,0) (1) {\includegraphics[scale=0.85]{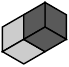}};
\node at (4,0) (2) {\includegraphics[scale=0.85]{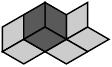}};
\node at (6,0) (3) {\includegraphics[scale=0.85]{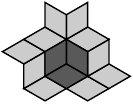}};
\node at (8,0) (4) {\includegraphics[scale=0.85]{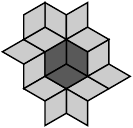}};

\path[->, thick] (0) edge node [above] {\small$1,2,3$} (1);
\path[->, thick] (1) edge node [above] {\small$3$} (2);
\path[->, thick] (2) edge node [above] {\small$3$} (3);
\path[->, thick] (3) edge node [above] {\small$3$} (4);

\draw[->, thick] (1) .. controls +(120:10mm) and +(60:10mm) .. node [above] {\small$1, 2$} (1);
\draw[->, thick] (2) .. controls +(120:10mm) and +(60:10mm) .. node [above] {\small$1, 2$} (2);
\draw[->, thick] (3) .. controls +(120:12mm) and +(60:12mm) .. node [above] {\small$1, 2$} (3);
\draw[->, thick] (4) .. controls +(120:13mm) and +(60:13mm) .. node [above] {\small$1, 2, 3$} (4);
\end{tikzpicture}
\end{center}
\vspace{-2.5em}
\end{proof}

\begin{lemm}
\label{lem:annulus_induction}
Let $\bfGa$ be a stepped plane that avoids
\myvcenter{%
\begin{tikzpicture}
[x={(-0.173205cm,-0.100000cm)}, y={(0.173205cm,-0.100000cm)}, z={(0.000000cm,0.200000cm)}]
\definecolor{facecolor}{rgb}{0.8,0.8,0.8}
\fill[fill=facecolor, draw=black, shift={(0,0,0)}]
(0, 0, 0) -- (0, 1, 0) -- (0, 1, 1) -- (0, 0, 1) -- cycle;
\fill[fill=facecolor, draw=black, shift={(0,1,0)}]
(0, 0, 0) -- (0, 1, 0) -- (0, 1, 1) -- (0, 0, 1) -- cycle;
\end{tikzpicture}\,}%
,
\myvcenter{%
\begin{tikzpicture}
[x={(-0.173205cm,-0.100000cm)}, y={(0.173205cm,-0.100000cm)}, z={(0.000000cm,0.200000cm)}]
\definecolor{facecolor}{rgb}{0.8,0.8,0.8}
\fill[fill=facecolor, draw=black, shift={(0,0,0)}]
(0, 0, 0) -- (0, 1, 0) -- (0, 1, 1) -- (0, 0, 1) -- cycle;
\fill[fill=facecolor, draw=black, shift={(0,0,1)}]
(0, 0, 0) -- (0, 1, 0) -- (0, 1, 1) -- (0, 0, 1) -- cycle;
\end{tikzpicture}\,}%
,
\myvcenter{%
\begin{tikzpicture}
[x={(-0.173205cm,-0.100000cm)}, y={(0.173205cm,-0.100000cm)}, z={(0.000000cm,0.200000cm)}]
\definecolor{facecolor}{rgb}{0.8,0.8,0.8}
\fill[fill=facecolor, draw=black, shift={(0,0,1)}]
(0, 0, 0) -- (0, 0, 1) -- (1, 0, 1) -- (1, 0, 0) -- cycle;
\fill[fill=facecolor, draw=black, shift={(0,0,0)}]
(0, 0, 0) -- (0, 0, 1) -- (1, 0, 1) -- (1, 0, 0) -- cycle;
\end{tikzpicture}\,}%
and \myvcenter{%
\begin{tikzpicture}
[x={(-0.173205cm,-0.100000cm)}, y={(0.173205cm,-0.100000cm)}, z={(0.000000cm,0.200000cm)}]
\definecolor{facecolor}{rgb}{0.8,0.8,0.8}
\fill[fill=facecolor, draw=black, shift={(0,0,0)}]
(0, 0, 0) -- (1, 0, 0) -- (1, 1, 0) -- (0, 1, 0) -- cycle;
\fill[fill=facecolor, draw=black, shift={(1,1,0)}]
(0, 0, 0) -- (1, 0, 0) -- (1, 1, 0) -- (0, 1, 0) -- cycle;
\end{tikzpicture}\,}%
.
Let $A \subseteq \bfGa$ be an $\LFS$-annulus of a pattern $P \subseteq \bfGa$,
and let $\Sigma = \Sigma_i$ for some $i \in \{1,2,3\}$.
Then $\Sigma(A)$ is an $\LFS$-annulus of $\Sigma(P)$ in the stepped plane $\Sigma(\bfGa)$.
\end{lemm}

\begin{proof}
We must prove the following:
\begin{enumerate}
  \item $\Sigma(P)$, $\Sigma(A) \cup \Sigma(P)$ and $\bfGa \setminus (\Sigma(A) \cup \Sigma(P))$ are $\LFS$-covered;
  \item $\Sigma(A)$ is strongly $\LFS$-covered;
  \item $\Sigma(A)$ and $\Sigma(P)$ have no face in common;
  \item $\Sigma(P) \cap \overline{\Sigma(\bfGa) \setminus (\Sigma(P) \cup \Sigma(A))} = \varnothing$.
\end{enumerate}
Conditions~\ref{defi:anneauprop1} and~\ref{defi:anneauprop3} hold
thanks to Lemma~\ref{lemm:coverprop} and Proposition~\ref{prop::imgplane} respectively,
and~\ref{defi:anneauprop2} holds thanks to Lemma~\ref{lemm:strongcovFS}.
It remains to prove that~\ref{defi:anneauprop4} holds.

Suppose that~\ref{defi:anneauprop4} does not hold.
This implies that there exist faces
$
f \in P,
g \in \bfGa \setminus (A \cup P),
f' \in \Sigma(f)$
and $g' \in \Sigma(g)$
such that $f'$ and $g'$ have a nonempty intersection.
Also, $f \cup g$ must be disconnected because $P$ and $\overline{\bfGa \setminus (P \cup A)}$
have empty intersection by hypothesis.


The strategy of the proof is as follows:
we check all the possible patterns $f \cup g$ and $f' \cup g'$ as above,
and for each case we derive a contradiction.
This can be done by inspection of a finite number of cases.
Indeed, there are $36$ possibilities for $f' \cup g'$ up to translation
(the number of connected two-face patterns that share a vertex or an edge),
and each of these patterns has a finite number of two-face preimages.

The first patterns $f' \cup g'$ which have disconnected preimages are
$f' \cup g' = [\mathbf 0, 3]^\star \cup [(1, 1, 0), 3]^\star$
or $[\mathbf 0, 2]^\star \cup [(1, -1, 1), 1]^\star$
or $[\mathbf 0, 2]^\star \cup [(1, 0, 1), 2]^\star$.
These cases can be ignored thanks to Lemma~\ref{lemm:forbidden}:
the first case (\,\myvcenter{%
\begin{tikzpicture}
[x={(-0.173205cm,-0.100000cm)}, y={(0.173205cm,-0.100000cm)}, z={(0.000000cm,0.200000cm)}]
\definecolor{facecolor}{rgb}{0.8,0.8,0.8}
\fill[fill=facecolor, draw=black, shift={(0,0,0)}]
(0, 0, 0) -- (1, 0, 0) -- (1, 1, 0) -- (0, 1, 0) -- cycle;
\fill[fill=facecolor, draw=black, shift={(1,1,0)}]
(0, 0, 0) -- (1, 0, 0) -- (1, 1, 0) -- (0, 1, 0) -- cycle;
\end{tikzpicture}\,}%
) is forbidden by assumption.
In the second case, Definition~\ref{def:stepped} implies that if a stepped plane contains
$f' \cup g'$, then it contains the face $[(0,0,1),2]^\star$ shown in dark gray
\myvcenter{\begin{tikzpicture}
[x={(-0.173205cm,-0.100000cm)}, y={(0.173205cm,-0.100000cm)}, z={(0.000000cm,0.200000cm)}]
\definecolor{facecolor}{rgb}{0.8,0.8,0.8}
\fill[fill=facecolor, draw=black, shift={(1,-1,1)}]
(0, 0, 0) -- (0, 1, 0) -- (0, 1, 1) -- (0, 0, 1) -- cycle;
\fill[fill=facecolor, draw=black, shift={(0,0,0)}]
(0, 0, 0) -- (0, 0, 1) -- (1, 0, 1) -- (1, 0, 0) -- cycle;
\definecolor{facecolor}{rgb}{0.35,0.35,0.35}
\fill[fill=facecolor, draw=black, shift={(0,0,0)}]
(0, 0, 1) -- (0, 0, 2) -- (1, 0, 2) -- (1, 0, 1) -- cycle;
\end{tikzpicture}}\,.
This contains a pattern ruled out by Lemma~\ref{lemm:forbidden},
which settles this case.
The third case can be treated in the same way.

Another possibility is
$f' \cup g' = [\mathbf 0, 2]^\star \cup [(1, -1, 1), 3]^\star$,
which admits six disconnected preimages (two for each $\Sigma_i$).
They are shown below (in light gray),
together with their only possible completion within a stepped plane
(in dark gray), which can be deduced from Definition~\ref{def:stepped}:
\[
\Sigma_1 :
\myvcenter{\begin{tikzpicture}
[x={(-0.173205cm,-0.100000cm)}, y={(0.173205cm,-0.100000cm)}, z={(0.000000cm,0.200000cm)}]
\definecolor{facecolor}{rgb}{0.800,0.800,0.800}
\fill[fill=facecolor, draw=black, shift={(1,-1,1)}]
(0, 0, 0) -- (0, 1, 0) -- (0, 1, 1) -- (0, 0, 1) -- cycle;
\fill[fill=facecolor, draw=black, shift={(0,0,0)}]
(0, 0, 0) -- (0, 1, 0) -- (0, 1, 1) -- (0, 0, 1) -- cycle;
\definecolor{facecolor}{rgb}{0.35,0.35,0.35}
\fill[fill=facecolor, draw=black, shift={(0,0,0)}]
(0, 0, 0) -- (0, 0, 1) -- (1, 0, 1) -- (1, 0, 0) -- cycle;
\fill[fill=facecolor, draw=black, shift={(0,0,1)}]
(0, 0, 0) -- (0, 0, 1) -- (1, 0, 1) -- (1, 0, 0) -- cycle;
\end{tikzpicture}}\,, \
\myvcenter{\begin{tikzpicture}
[x={(-0.173205cm,-0.100000cm)}, y={(0.173205cm,-0.100000cm)}, z={(0.000000cm,0.200000cm)}]
\definecolor{facecolor}{rgb}{0.800,0.800,0.800}
\fill[fill=facecolor, draw=black, shift={(-1,0,0)}]
(0, 0, 0) -- (0, 0, 1) -- (1, 0, 1) -- (1, 0, 0) -- cycle;
\fill[fill=facecolor, draw=black, shift={(1,-1,1)}]
(0, 0, 0) -- (0, 1, 0) -- (0, 1, 1) -- (0, 0, 1) -- cycle;
\definecolor{facecolor}{rgb}{0.35,0.35,0.35}
\fill[fill=facecolor, draw=black, shift={(0,0,0)}]
(0, 0, 0) -- (0, 0, 1) -- (1, 0, 1) -- (1, 0, 0) -- cycle;
\fill[fill=facecolor, draw=black, shift={(0,0,1)}]
(0, 0, 0) -- (0, 0, 1) -- (1, 0, 1) -- (1, 0, 0) -- cycle;
\end{tikzpicture}}
\qquad
\Sigma_2 :
\myvcenter{\begin{tikzpicture}
[x={(-0.173205cm,-0.100000cm)}, y={(0.173205cm,-0.100000cm)}, z={(0.000000cm,0.200000cm)}]
\definecolor{facecolor}{rgb}{0.800,0.800,0.800}
\fill[fill=facecolor, draw=black, shift={(0,0,0)}]
(0, 0, 0) -- (0, 0, 1) -- (1, 0, 1) -- (1, 0, 0) -- cycle;
\fill[fill=facecolor, draw=black, shift={(-1,1,1)}]
(0, 0, 0) -- (0, 0, 1) -- (1, 0, 1) -- (1, 0, 0) -- cycle;
\definecolor{facecolor}{rgb}{0.35,0.35,0.35}
\fill[fill=facecolor, draw=black, shift={(0,0,0)}]
(0, 0, 0) -- (0, 1, 0) -- (0, 1, 1) -- (0, 0, 1) -- cycle;
\fill[fill=facecolor, draw=black, shift={(0,0,1)}]
(0, 0, 0) -- (0, 1, 0) -- (0, 1, 1) -- (0, 0, 1) -- cycle;
\end{tikzpicture}}\,, \
\myvcenter{\begin{tikzpicture}
[x={(-0.173205cm,-0.100000cm)}, y={(0.173205cm,-0.100000cm)}, z={(0.000000cm,0.200000cm)}]
\definecolor{facecolor}{rgb}{0.800,0.800,0.800}
\fill[fill=facecolor, draw=black, shift={(0,-1,0)}]
(0, 0, 0) -- (0, 1, 0) -- (0, 1, 1) -- (0, 0, 1) -- cycle;
\fill[fill=facecolor, draw=black, shift={(-1,1,1)}]
(0, 0, 0) -- (0, 0, 1) -- (1, 0, 1) -- (1, 0, 0) -- cycle;
\definecolor{facecolor}{rgb}{0.35,0.35,0.35}
\fill[fill=facecolor, draw=black, shift={(0,0,0)}]
(0, 0, 0) -- (0, 1, 0) -- (0, 1, 1) -- (0, 0, 1) -- cycle;
\fill[fill=facecolor, draw=black, shift={(0,0,1)}]
(0, 0, 0) -- (0, 1, 0) -- (0, 1, 1) -- (0, 0, 1) -- cycle;
\end{tikzpicture}}
\qquad
\Sigma_3 :
\myvcenter{\begin{tikzpicture}
[x={(-0.173205cm,-0.100000cm)}, y={(0.173205cm,-0.100000cm)}, z={(0.000000cm,0.200000cm)}]
\definecolor{facecolor}{rgb}{0.800,0.800,0.800}
\fill[fill=facecolor, draw=black, shift={(-1,1,1)}]
(0, 0, 0) -- (1, 0, 0) -- (1, 1, 0) -- (0, 1, 0) -- cycle;
\fill[fill=facecolor, draw=black, shift={(0,0,0)}]
(0, 0, 0) -- (1, 0, 0) -- (1, 1, 0) -- (0, 1, 0) -- cycle;
\definecolor{facecolor}{rgb}{0.35,0.35,0.35}
\fill[fill=facecolor, draw=black, shift={(0,0,0)}]
(0, 0, 0) -- (0, 1, 0) -- (0, 1, 1) -- (0, 0, 1) -- cycle;
\fill[fill=facecolor, draw=black, shift={(0,1,0)}]
(0, 0, 0) -- (0, 1, 0) -- (0, 1, 1) -- (0, 0, 1) -- cycle;
\end{tikzpicture}}\,, \
\myvcenter{\begin{tikzpicture}
[x={(-0.173205cm,-0.100000cm)}, y={(0.173205cm,-0.100000cm)}, z={(0.000000cm,0.200000cm)}]
\definecolor{facecolor}{rgb}{0.800,0.800,0.800}
\fill[fill=facecolor, draw=black, shift={(-1,1,1)}]
(0, 0, 0) -- (1, 0, 0) -- (1, 1, 0) -- (0, 1, 0) -- cycle;
\fill[fill=facecolor, draw=black, shift={(0,0,-1)}]
(0, 0, 0) -- (0, 1, 0) -- (0, 1, 1) -- (0, 0, 1) -- cycle;
\definecolor{facecolor}{rgb}{0.35,0.35,0.35}
\fill[fill=facecolor, draw=black, shift={(0,0,0)}]
(0, 0, 0) -- (0, 1, 0) -- (0, 1, 1) -- (0, 0, 1) -- cycle;
\fill[fill=facecolor, draw=black, shift={(0,1,0)}]
(0, 0, 0) -- (0, 1, 0) -- (0, 1, 1) -- (0, 0, 1) -- cycle;
\end{tikzpicture}}\,.
\]
The patterns that appear in dark gray are forbidden by Lemma~\ref{lemm:forbidden},
so this case is settled.

The last two possibilities are
$f' \cup g' = [\mathbf 0, 3]^\star \cup [(1, 1, -1), 1]^\star$
or
$f' \cup g' = [\mathbf 0, 3]^\star \cup [(1, 1, -1), 2]^\star$.
Below (in light gray) are all the possible preimages $f \cup g$
(which are the same for the two possibilities),
and in dark gray is shown their only possible completion $X$ within a stepped plane:
\[
\Sigma_1 :
\myvcenter{\begin{tikzpicture}
[x={(-0.173205cm,-0.100000cm)}, y={(0.173205cm,-0.100000cm)}, z={(0.000000cm,0.200000cm)}]
\definecolor{facecolor}{rgb}{0.800000,0.800000,0.800000}
\fill[fill=facecolor, draw=black, shift={(0,0,0)}]
(0, 0, 0) -- (0, 1, 0) -- (0, 1, 1) -- (0, 0, 1) -- cycle;
\definecolor{facecolor}{rgb}{0.800000,0.800000,0.800000}
\fill[fill=facecolor, draw=black, shift={(-1,-1,1)}]
(0, 0, 0) -- (0, 1, 0) -- (0, 1, 1) -- (0, 0, 1) -- cycle;
\definecolor{facecolor}{rgb}{0.35,0.35,0.35}
\fill[fill=facecolor, draw=black, shift={(-1,-1,1)}]
(0, 0, 0) -- (1, 0, 0) -- (1, 1, 0) -- (0, 1, 0) -- cycle;
\fill[fill=facecolor, draw=black, shift={(-1,0,1)}]
(0, 0, 0) -- (1, 0, 0) -- (1, 1, 0) -- (0, 1, 0) -- cycle;
\end{tikzpicture}}\,, \
\myvcenter{\begin{tikzpicture}
[x={(-0.173205cm,-0.100000cm)}, y={(0.173205cm,-0.100000cm)}, z={(0.000000cm,0.200000cm)}]
\definecolor{facecolor}{rgb}{0.800000,0.800000,0.800000}
\fill[fill=facecolor, draw=black, shift={(0,0,0)}]
(0, 0, 0) -- (0, 1, 0) -- (0, 1, 1) -- (0, 0, 1) -- cycle;
\definecolor{facecolor}{rgb}{0.800000,0.800000,0.800000}
\fill[fill=facecolor, draw=black, shift={(-2,-1,1)}]
(0, 0, 0) -- (1, 0, 0) -- (1, 1, 0) -- (0, 1, 0) -- cycle;
\definecolor{facecolor}{rgb}{0.35,0.35,0.35}
\fill[fill=facecolor, draw=black, shift={(-1,-1,1)}]
(0, 0, 0) -- (1, 0, 0) -- (1, 1, 0) -- (0, 1, 0) -- cycle;
\fill[fill=facecolor, draw=black, shift={(-1,0,1)}]
(0, 0, 0) -- (1, 0, 0) -- (1, 1, 0) -- (0, 1, 0) -- cycle;
\end{tikzpicture}}
\qquad
\Sigma_2 :
\myvcenter{\begin{tikzpicture}
[x={(-0.173205cm,-0.100000cm)}, y={(0.173205cm,-0.100000cm)}, z={(0.000000cm,0.200000cm)}]
\definecolor{facecolor}{rgb}{0.800000,0.800000,0.800000}
\fill[fill=facecolor, draw=black, shift={(0,0,0)}]
(0, 0, 0) -- (0, 0, 1) -- (1, 0, 1) -- (1, 0, 0) -- cycle;
\definecolor{facecolor}{rgb}{0.800000,0.800000,0.800000}
\fill[fill=facecolor, draw=black, shift={(-1,-1,1)}]
(0, 0, 0) -- (0, 0, 1) -- (1, 0, 1) -- (1, 0, 0) -- cycle;
\definecolor{facecolor}{rgb}{0.35,0.35,0.35}
\fill[fill=facecolor, draw=black, shift={(-1,-1,1)}]
(0, 0, 0) -- (1, 0, 0) -- (1, 1, 0) -- (0, 1, 0) -- cycle;
\fill[fill=facecolor, draw=black, shift={(0,-1,1)}]
(0, 0, 0) -- (1, 0, 0) -- (1, 1, 0) -- (0, 1, 0) -- cycle;
\end{tikzpicture}}\,, \
\myvcenter{\begin{tikzpicture}
[x={(-0.173205cm,-0.100000cm)}, y={(0.173205cm,-0.100000cm)}, z={(0.000000cm,0.200000cm)}]
\definecolor{facecolor}{rgb}{0.800000,0.800000,0.800000}
\fill[fill=facecolor, draw=black, shift={(0,0,0)}]
(0, 0, 0) -- (0, 0, 1) -- (1, 0, 1) -- (1, 0, 0) -- cycle;
\definecolor{facecolor}{rgb}{0.800000,0.800000,0.800000}
\fill[fill=facecolor, draw=black, shift={(-1,-2,1)}]
(0, 0, 0) -- (1, 0, 0) -- (1, 1, 0) -- (0, 1, 0) -- cycle;
\definecolor{facecolor}{rgb}{0.35,0.35,0.35}
\fill[fill=facecolor, draw=black, shift={(-1,-1,1)}]
(0, 0, 0) -- (1, 0, 0) -- (1, 1, 0) -- (0, 1, 0) -- cycle;
\fill[fill=facecolor, draw=black, shift={(0,-1,1)}]
(0, 0, 0) -- (1, 0, 0) -- (1, 1, 0) -- (0, 1, 0) -- cycle;
\end{tikzpicture}}
\qquad
\Sigma_3 :
\myvcenter{\begin{tikzpicture}
[x={(-0.173205cm,-0.100000cm)}, y={(0.173205cm,-0.100000cm)}, z={(0.000000cm,0.200000cm)}]
\definecolor{facecolor}{rgb}{0.800000,0.800000,0.800000}
\fill[fill=facecolor, draw=black, shift={(0,0,0)}]
(0, 0, 0) -- (1, 0, 0) -- (1, 1, 0) -- (0, 1, 0) -- cycle;
\definecolor{facecolor}{rgb}{0.800000,0.800000,0.800000}
\fill[fill=facecolor, draw=black, shift={(-1,1,-1)}]
(0, 0, 0) -- (1, 0, 0) -- (1, 1, 0) -- (0, 1, 0) -- cycle;
\definecolor{facecolor}{rgb}{0.35,0.35,0.35}
\fill[fill=facecolor, draw=black, shift={(-1,1,-1)}]
(0, 0, 0) -- (0, 0, 1) -- (1, 0, 1) -- (1, 0, 0) -- cycle;
\fill[fill=facecolor, draw=black, shift={(0,1,-1)}]
(0, 0, 0) -- (0, 0, 1) -- (1, 0, 1) -- (1, 0, 0) -- cycle;
\end{tikzpicture}}\,, \
\myvcenter{\begin{tikzpicture}
[x={(-0.173205cm,-0.100000cm)}, y={(0.173205cm,-0.100000cm)}, z={(0.000000cm,0.200000cm)}]
\definecolor{facecolor}{rgb}{0.800000,0.800000,0.800000}
\fill[fill=facecolor, draw=black, shift={(0,0,0)}]
(0, 0, 0) -- (1, 0, 0) -- (1, 1, 0) -- (0, 1, 0) -- cycle;
\definecolor{facecolor}{rgb}{0.800000,0.800000,0.800000}
\fill[fill=facecolor, draw=black, shift={(-1,1,-2)}]
(0, 0, 0) -- (0, 0, 1) -- (1, 0, 1) -- (1, 0, 0) -- cycle;
\definecolor{facecolor}{rgb}{0.35,0.35,0.35}
\fill[fill=facecolor, draw=black, shift={(-1,1,-1)}]
(0, 0, 0) -- (0, 0, 1) -- (1, 0, 1) -- (1, 0, 0) -- cycle;
\fill[fill=facecolor, draw=black, shift={(0,1,-1)}]
(0, 0, 0) -- (0, 0, 1) -- (1, 0, 1) -- (1, 0, 0) -- cycle;
\end{tikzpicture}}\,.
\]
Now, we have $X \subseteq A$ because Condition~\ref{defi:anneauprop4}
for $A$ and $P$ would fail otherwise ($f$ and $g$ cannot touch).
However, this contradicts the fact that strongly $\LFS$-connected.
Indeed, $X \in \Ledge$ but there cannot exist a pattern $Y \in \LFS$ such that $X \subseteq Y \subseteq A$
because then we must have $Y = \myvcenter{%
\begin{tikzpicture}
[x={(-0.173205cm,-0.100000cm)}, y={(0.173205cm,-0.100000cm)}, z={(0.000000cm,0.200000cm)}]
\definecolor{facecolor}{rgb}{0.800,0.800,0.800}
\fill[fill=facecolor, draw=black, shift={(0,0,0)}]
(0, 0, 0) -- (0, 0, 1) -- (1, 0, 1) -- (1, 0, 0) -- cycle;
\fill[fill=facecolor, draw=black, shift={(0,0,0)}]
(0, 0, 0) -- (1, 0, 0) -- (1, 1, 0) -- (0, 1, 0) -- cycle;
\fill[fill=facecolor, draw=black, shift={(1,0,0)}]
(0, 0, 0) -- (0, 0, 1) -- (1, 0, 1) -- (1, 0, 0) -- cycle;
\end{tikzpicture}\,}%
$, $\myvcenter{%
\begin{tikzpicture}
[x={(-0.173205cm,-0.100000cm)}, y={(0.173205cm,-0.100000cm)}, z={(0.000000cm,0.200000cm)}]
\definecolor{facecolor}{rgb}{0.800,0.800,0.800}
\fill[fill=facecolor, draw=black, shift={(0,0,0)}]
(0, 0, 0) -- (0, 0, 1) -- (1, 0, 1) -- (1, 0, 0) -- cycle;
\fill[fill=facecolor, draw=black, shift={(1,0,0)}]
(0, 0, 0) -- (1, 0, 0) -- (1, 1, 0) -- (0, 1, 0) -- cycle;
\fill[fill=facecolor, draw=black, shift={(0,0,0)}]
(0, 0, 0) -- (1, 0, 0) -- (1, 1, 0) -- (0, 1, 0) -- cycle;
\end{tikzpicture}\,}%
$ or $\myvcenter{%
\begin{tikzpicture}
[x={(-0.173205cm,-0.100000cm)}, y={(0.173205cm,-0.100000cm)}, z={(0.000000cm,0.200000cm)}]
\definecolor{facecolor}{rgb}{0.800,0.800,0.800}
\fill[fill=facecolor, draw=black, shift={(0,0,0)}]
(0, 0, 0) -- (0, 1, 0) -- (0, 1, 1) -- (0, 0, 1) -- cycle;
\fill[fill=facecolor, draw=black, shift={(0,1,0)}]
(0, 0, 0) -- (1, 0, 0) -- (1, 1, 0) -- (0, 1, 0) -- cycle;
\fill[fill=facecolor, draw=black, shift={(0,0,0)}]
(0, 0, 0) -- (1, 0, 0) -- (1, 1, 0) -- (0, 1, 0) -- cycle;
\end{tikzpicture}\,}%
$,
so $Y$ must overlap with $f$ or $g$, which is impossible because $f$ and $g$ are not in $A$.
\end{proof}

\begin{proof}[Proof of Proposition~\ref{prop:pn_union}]
Let $\bfv \in \Fthree$.
To prove the proposition, it is enough to prove that $\cup_{n=0}^\infty P_n = \bfGa_\bfv$,
thanks to Remark~\ref{rq:planes}.
Let $P \in \bfGa_\bfv$ be a finite pattern.
The \emph{combinatorial radius} of $P$ is defined to be
the length of the smallest path of edge-connected unit faces from the origin to $\bfGa_\bfv \setminus P$.

Now, since $\bfv \in \Fthree$,
$\SFSiii$ occurs infinitely many often in the sequence $(\Sigma_i)_{i \in \bbN}$
of the dual substitutions associated with the $\bfF$-expansion of $\bfv$.
Hence, we can apply Proposition~\ref{prop:annulus} to prove that there exists $k \in \bbN$
such that for all $\ell \geq 0$, the pattern
$A_\ell = \Sigma_1 \cdots\Sigma_{k+\ell}(\mcU) \setminus \Sigma_1 \cdots \Sigma_\ell(\mcU)$
is an $\LFS$-annulus of $\Sigma_1 \cdots \Sigma_\ell(\mcU)$.

By Condition~\ref{defi:anneauprop1} of Definition~\ref{defi:annulus},
the pattern $A_\ell \cup \Sigma_1 \cdots \Sigma_\ell(\mcU)$ is simply connected for all $\ell \geq 0$,
so its combinatorial radius increases at least by $1$ when $\ell$ is incremented,
thanks to Conditions~\ref{defi:anneauprop3}~\ref{defi:anneauprop4}.
This proves the required property.
\end{proof}


\section{Connectedness at the critical thickness}\label{sec:main}
We can now prove our main result, namely the  characterization of  the normal vectors $\bfv$ for which the
arithmetical discrete plane is $2$-connected at the connecting thickness $\Omega(\bfv)$.

\subsection{Technical lemmas}
In the present section, we provide technical properties for  dimensions $d=3,2,1$
which will allow us to  use argumentsof  dimension  reduction.
Let us first provide several notation and definitions.
Let $d \geq 1$ be an integer and let ${\mathcal O^+_d}= \{ \bfv \in \bbZ^d \mid 0 \leq \bfv_1 \leq \dots \leq \bfv_d\}$.
The notion of $2$-connectedness extends  in a natural way to $(d-1)$-connectedness in $\bbZ^d$,
which  induces  the  corresponding  notion of  connecting thickness $\Omega(\bfv)$
for $ \bfv \in {\mathbb Z}^d$, as well as the notion of arithmetical discrete line in $\bbZ^2$.

\begin{ntt}[Extension of  $\gcd$]
Let $(\alpha,\beta) \in \bbR^2_+$ such that $\dim_\bbQ(\bfv_1,\bfv_2) = 1$.
There exists $\gamma \in \bbR_+$ such that $(\gamma \alpha, \gamma \beta) \in \bbN^2$ and we set
$\gcd( \alpha,\beta ) = \dfrac{\gcd(\gamma \alpha, \gamma \beta) }{\gamma}$.
One checks that this definition does not depend on the choice of $\gamma$.
\end{ntt}



Lemma~\ref{lem::projection} allows us to restrict further the domain of investigated normal vectors $\bfv$
to the ones with only non-zero coordinates.

\begin{lemm}\label{lem::projection}
Let $\bfv \in \Otpd{d}$ and $\omega \in \bbR$. Assume that $\bfv_1=0$
and let $\bfv' =(\bfv_2,\dots,\bfv_d)\in \Otpd{d-1}$.
Then $\frakP(\bfv,\omega)$ is $(d-1)$-connected in $\bbZ^d$ if and only if
$\frakP(\bfv',\omega)$ is $(d-2)$-connected in $\bbZ^{d-1}$.
Consequently, $\Omega(\bfv) = \Omega(\bfv')$.
\end{lemm}

\begin{proof}
Let $\frakP=\frakP(\bfv,\omega)$ and $\frakP'=\frakP(\bfv',\omega)$. We have
\begin{eqnarray*}
  \frakP &=& \{\bfx \in \bbZ^d \mid 0 \leq \langle \bfv,\bfx\rangle < \omega \}\\
  &=& \{(\bfx_1,\bfx') \in \bbZ^d \mid 0 \leq \langle \bfv',\bfx'\rangle < \omega \}\\
  &=& \bbZ \times \frakP',
\end{eqnarray*}
by writing $\bfx=(\bfx_1,\bfx')$ with $ \bfx_1 \in \bbZ$ and $\bfx' \in \frakP'$.

Assume that $\frakP'$ is $(d-2)$-connected, and let $\bfx,\bfy \in \frakP$.
We write $\bfx=(\bfx_1,\bfx')$ and $\bfy=(\bfy_1,\bfy')$, where $\bfx_1,\bfy_1 \in \bbZ$ and $\bfx',\bfy' \in \frakP'$.
There exists a $(d-2)$-connected path $(\bfz' _{1},\dots,\bfz' _{n})$ from $\bfx'$ to $\bfy'$ in $\frakP'$.
We assume that $\bfy_1\geq \bfx_1$ (otherwise, we exchange $\bfx$ and $\bfy$).
The path
$((\bfx_1,\bfz' _{1}),\dots,(\bfx_1,\bfz'_{n}),(\bfx_1+1,\bfz' _{n}),(\bfx_1+2,\bfz' _{n}),\dots,(\bfy_1,\bfz' _{n}))$
is a $(d-1)$-connected path between $\bfx$ and $\bfy$ in $\frakP$.

Assume now that $\frakP$ is connected and let $\bfx',\bfy' \in \frakP'$.
We have $(0,\bfx') \in \frakP$ and $(0,\bfy') \in \frakP$.
There exists a $(d-1)$-connected path $(\bfz _{1},\dots,\bfz _{n})$ between $(0,\bfx')$ and $(0,\bfy')$ in $\frakP$.
For each $i=1,\dots,n$, write $\bfz_i = (\bfz_{i,1},\bfz'_i)$.
Then the sequence $(\bfz'_1,\dots,\bfz'_n)$ is a $(d-2)$-connected path between $\bfx'$ and $\bfy'$ in $\frakP'$.
\end{proof}

Let us deal with the special case where the normal vector $\bfv$ has exactly one non-zero coordinate.
We characterize completely the set of thicknesses $\omega$ for which $\frakP(\bfv,\omega)$ is  connected.
\begin{lemm}\label{lem::d=1}
Let $d=2,3$.
Let $\bfv \in \Otpd{d}$ and $\omega \in \bbR$.
If $\bfv$ has exactly one non-zero coordinate, then $\frakP(\bfv,\omega)$ is $(d-1)$-connected
as soon as $\omega > 0$. Consequently, $\Omega(\bfv) = 0$.
\end{lemm}
\begin{proof}
The only non-zero coordinate of $\bfv$ is $\bfv_d$.
According to Lemma~\ref{lem::projection}, $\frakP(\bfv,\omega)$ is $(d-1)$-connected if and only if
$\frakP(\bfv_d,\omega)$ is $0$-connected.
Now, $\frakP(\bfv_d,\omega) = \{x \in \bbZ \mid 0 \leq \bfv_d\,x < \omega\}$ is an interval of $\bbZ$.
It is $(d-1)$-connected  as soon as it is  not empty  (by definition) which means ${\omega} > 0$.
\end{proof}

\begin{lemm}[An upper-bound for $\Omega(\bfv)$]
\label{lem::omega-upper-bound}
Let $\bfv \in \Otp$ and $\omega \in \bbR$.
We assume $\bfv \neq 0$. We set $\xi(\bfv) = \min\{|\bfv_i| \mid \bfv_i \neq 0\}$.
If $\omega \geq \|\bfv\|_\infty+\xi(\bfv)$, then $\frakP(\bfv,\omega)$ is $2$-connected.
Consequently, $\Omega(\bfv) \leq \|\bfv\|_\infty + \xi(\bfv)$.
\end{lemm}

\begin{proof}
Assume first $d=1$.
Since $\bfv \neq \vect{0}$, we have $v_1 > 0$.
One has $\frakP(\bfv,\omega)$  $0$-connected  as soon as  $\omega > 0$, hence,
in particular, if $\omega \geq\xi(\bfv)+\|\bfv\|_\infty = 2\,v_1$.

Let us now assume that $d=2$ or $3$  and that the result holds for all $d' <d$.
We prove this result by induction on $d$.

If $\bfv_1=0$ then we set $\bfv' = (\bfv_2,\ldots,\bfv_d)$ and in this case,
thanks to Lemma~\ref{lem::projection}, $\frakP(\bfv,\omega)$ is $(d-1)$-connected in $\bbZ^d$
if $\frakP(\bfv',\omega)$ is $(d-2)$-connected in $\bbZ^{d-1}$.
Then we get the result by the induction hypothesis because
$\xi(\bfv') = \xi(\bfv)$ and $\|\bfv'\|_\infty = \|\bfv\|_\infty$.

If $\bfv_1 > 0$ then we set $\bfv'=(\bfv_1,\ldots,\bfv_{d-1})$.
We have $\xi(\bfv') = \xi(\bfv) = \bfv_1$ and $\|\bfv'\|_\infty = \bfv_{d-1} \leq \bfv_d = \|\bfv\|_\infty$.
For all $\bfz \in \bbZ$, let
$\frakP_\bfz = \{(\bfx_1,\ldots,\bfx_{d-1},\bfz)\mid (\bfx_1,\ldots,\bfx_{d-1}) \in \bbZ^{d-1}, \ 0\leq \bfv_1\,\bfx_1+\cdots+\bfv_{d-1}\,\bfx_{d-1}+\bfv_d\,\bfz < \omega\}$.
We have $\frakP(\bfv,\omega) = \bigcup_{\bfz\in\bbZ} \frakP_\bfz$.

If $\frakP_\bfz$ is connected in $\bbZ^d$, then the projection $\frakP'_\bfz$ of $\frakP_\bfz$
on the $d-1$ first coordinates is  connected.
If $\omega\geq \xi(\bfv)+\|\bfv\|_\infty$ then $\omega \geq \xi(\bfv')+\|\bfv'\|_\infty$,
and, by the induction hypothesis, $\frakP'_\bfz$ is connected; therefore, $\frakP_\bfz$ is connected.
We are left to prove that the $\frakP_\bfz$'s are adjacent, that is,
\[
\forall \bfz \in \bbZ,~\exists \bfx_1,\dots,\bfx_{d-1} \in\bbZ,~(\bfx_1,\dots,\bfx_{d-1},\bfz) \in \frakP_\bfz, \  (\bfx_1,\dots,\bfx_{d-1},\bfz+1) \in \frakP_{\bfz+1}
\]
or equivalently
\begin{equation}\label{eqn::dn+d1}
  \forall \bfz \in \bbZ,~\exists \bfx_1,\dots,\bfx_{d-1} \in\bbZ,~0 \leq \bfv_1\,\bfx_1+\cdots+\bfv_{d-1}\,\bfx_{d-1}+\bfv_d\,\bfz < \omega - \bfv_d.
\end{equation}
We distinguish two cases according to whether $\dim_\bbQ(\bfv_1,\ldots,\bfv_{d-1}) = 1$
or $\dim_\bbQ(\bfv_1,\ldots,\bfv_{d-1}) \geq 2$.
\begin{itemize}
\item
If $\dim_\bbQ(\bfv_1,\ldots,\bfv_{d-1}) = 1$ then let $\gamma = \gcd(\bfv_1,\ldots,\bfv_{d-1})$.
We have $\bfv_1\,\bbZ+\cdots+\bfv_{d-1}\,\bbZ = \gamma\,\bbZ$ and Condition~(\ref{eqn::dn+d1}) is equivalent to
\[
  \forall \bfz \in \bbZ,~\exists u \in\bbZ,~0 \leq \gamma\,u+\bfv_d\,z < \omega - \bfv_d
  \iff
  \forall \bfz \in \bbZ,~\omega-\bfv_d > (v_d\,\bfz) \bmod \gamma,
\]
which is satisfied as soon as $\omega \geq \bfv_d+\gamma$,
and especially if $\omega \geq \bfv_d+\bfv_1$ because $\bfv_1 \geq \gamma$.
\item
If $\dim_\bbQ(\bfv_1,\ldots,\bfv_{d-1}) \geq 2$ then we must have $d\geq 3$
and $\bfv_1\,\bbZ+\cdots+\bfv_{d-1}\,\bbZ$ is dense in $\bbR$.
Condition~(\ref{eqn::dn+d1}) is satisfied as soon as $\omega-\bfv_d > 0$,
and especially if $\omega \geq \bfv_d+\bfv_1$ because $\bfv_1>0$.
\end{itemize}
\end{proof}

\begin{lemm}\label{cor::not_in_F3}
Let $\bfv \in \Otp$ with $\dim_\bbQ(\bfv_1,\bfv_2,\bfv_3)>1$.
If $\bfv_1+\bfv_2 \leq \bfv_3$ and $\bfv^{(n)}_1>0$ for all $n \in \bbN$, then $\Omega(\bfv) = \| \bfv \|_\infty$.
In particular, the arithmetical discrete plane $\frakP( \bfv,\Omega(\bfv))$ is not $2$-connected.
\end{lemm}

\begin{proof}
For all $n \in \bbN^\star$, we set $\omega^{(n)} = \sum^{n-1}_{i=0}{v^{(i)}_1}$.
Then, for all $n \in \bbN^\star$, $\Omega(\bfv)=\Omega(\bfv^{(n)} )+ \omega^{(n)}$.
We set $\xi(\bfv) = \min\{|\bfv_i| \mid \bfv_i \neq 0\}$.
One has by~\cite[Lemma12]{JT09}  together with Lemma~\ref{lem::omega-upper-bound}
\[
\|\bfv\|_\infty \leq \Omega(\bfv) \leq \|\bfv\|_\infty + \xi(\bfv),
\]
and thus
\begin{equation*}
\| \bfv^{(n)} \|_\infty \leq \Omega(\bfv^{(n)} )\leq \| \bfv^{(n)} \|_\infty + \omega^{(n+1)} - \omega^{(n)},
\end{equation*}
or equivalently
\begin{equation*}
\| \bfv^{(n)} \|_\infty + \omega^{(n)}  \leq \Omega(\bfv)
    = \Omega(\bfv^{(n)}) + \omega^{(n)} \leq \| \bfv^{(n)} \|_\infty + \omega^{(n+1)}.
\end{equation*}
Since $\bfv_1+\bfv_2 \leq \bfv_3$, then $\bfv^{(n)}_1+\bfv^{(n)}_2 \leq \bfv^{(n)}_3$, for all $n \in \bbN^\star$.
It follows that $\lim_{n \to \infty}{ \bfv^{(n)}  }= (0,0,\bfv_3-(\bfv_1+\bfv_2))$
and $\lim_{n \to \infty}{  \omega^{(n)}   }= \bfv_1+\bfv_2$.
The non-$2$-connectedness of $\frakP(\vect{v},\Omega(\bfv))$ follows from the fact that
$\frakP(\vect{v}^{(n)}, \| \bfv^{(n)} \|_\infty )$ is not $2$-connected by~\cite{JT09}.
\end{proof}

Now it becomes natural to investigate the critical thickness of
normal vectors $\bfv$ for which $\dim_\bbQ(\bfv_1,\bfv_2,\bfv_3)=2$, $\bfv_1+\bfv_2 \leq \bfv_3$
and there exists $n_0 \in \bbN$ such that $\bfv^{(n_0)}_1=0$.
Note that in that case, $\dim_\bbQ(\bfv_1,\bfv_2)=1$.
Indeed  $\bfv_1+\bfv_2 \leq \bfv_3$ implies  that $\bfv^{(n)}_1+\bfv^{(n)}_2 \leq \bfv^{(n)}_3$ for all $n \in \bbN$.

\begin{lemm}\label{cor::dim2}
Let $\bfv \in \Otp$ with $\dim_\bbQ(\bfv_1,\bfv_2,\bfv_3)=2$ and $\bfv_1+\bfv_2 \leq \bfv_3$.
Let $n_0 \in \bbN$  be such that $\bfv^{(n_0)}_1 = 0$.
The arithmetical discrete plane $\frakP( \bfv,\Omega(\bfv) )$ is $2$-connected
and $\Omega(\bfv) = \bfv_3+\gcd(\bfv_1,\bfv_2) $.
\end{lemm}

\begin{proof}
According to  Lemma~\ref{lem::projection}, the arithmetical discrete plane $\frakP(\bfv,\Omega(\bfv))$ is $2$-connected
if  $\frakP(\bfv^{(n_0)},\Omega(\bfv^{(n_0)}))$ is $1$-connected in $\bbZ^2$.
In fact,  $\frakP(\bfv^{(n_0)},\Omega(\bfv^{(n_0)}))$ is a standard discrete line in $\bbZ^2$.
Thus, $\Omega(\bfv^{(n_0)})) = \| \bfv^{(n_0)} \|_1$ and $\frakP(\bfv^{(n_0)},\Omega(\bfv^{(n_0)}))$
is $1$-connected, which yields $\frakP(\bfv,\Omega(\bfv))$ $2$-connected.

One checks that, for all $n \in \bbN$,
\[
\bfv_3-(\bfv_1+\bfv_2)
    = \bfv^{(n)}_3-(\bfv^{(n)}_1+\bfv^{(n)}_2)
    = \bfv^{(n_0)}_3-\bfv^{(n_0)}_2,
\]
and
\[
\| \bfv \|_1-2\Omega(\bfv)
    = \| \bfv^{(n)} \|_1-2\Omega(\bfv^{(n)})
    = \| \bfv^{(n_0)} \|_1-2\Omega(\bfv^{(n_0)})
    = - \| \bfv^{(n_0)} \|_1.
\]
Hence, $\Omega(\bfv) = \dfrac{\| \bfv \|_1 + \| \bfv^{(n_0)} \|_1}{2}$.
It remains to express $\| \bfv^{(n_0)} \|_1$ in terms of $\bfv_1$, $\bfv_2$ and $\bfv_3$.
Since $\bfv_1+\bfv_2 \leq \bfv_3$, we deduce that
$\bfv^{(n)}_1+\bfv^{(n)}_2 \leq \bfv^{(n)}_3$ for all $n \in \bbN$, $\bfv^{(n_0)}_2 = \gcd (\bfv_1,\bfv_2)$
and $\bfv^{(n_0)}_3 = \bfv_3-(\bfv_1+\bfv_2)-\gcd (\bfv_1,\bfv_2)$.
It follows that $\Omega(\bfv) =  \bfv_3+\gcd(\bfv_1,\bfv_2)$.
\end{proof}


\subsection{Critical connectedness}

\begin{theo}
\label{thm::main}
Let $\bfv \in \Otp$ with  $\bfv \neq 0$.
The arithmetical discrete plane $\frakP(\bfv,\Omega(\bfv))$ is $2$-connected if and only if
one of the following  two conditions holds:
\begin{enumerate}
\item either $\bfv \in \Fthree$;
\item or there exists $n \in \bbN$ such that $\bfv^{(n)}_1 = 0$ with $\dim_\bbQ(\bfv^{(n)}_2,\bfv^{(n)}_3 ) = 2$.
\end{enumerate}
\end{theo}

\begin{proof}
Let $\bfv \in \Otp$.
\begin{itemize}
\item
We first  suppose that $\bfv \in \Fthree$ and let $\bfx \in \frakP(\bfv,\Omega(\bfv))$.
Thanks to Theorem~\ref{theo:algo_cc} we have $\Omega(\bfv) = \|\bfv\|_1/2$.
If $\| \bfv \|_\infty \leq \langle\bfx,\bfv \rangle < \|\bfv\|_1/2$, then
$\| \bfv \|_\infty - \bfv_1 \leq \langle\bfx-\bfe_1,\bfv \rangle
    < \|\bfv\|_1/2 - \bfv_1
    < \|\bfv \|_\infty$, $\bfx-\bfe_1 \in \frakP(\bfv, \| \bfv \|_\infty)$.
In other words, an element $\bfx$ of $\frakP(\bfv,\Omega(\bfv))$
either belongs to $\frakP(\bfv,\| \bfv \|_\infty)$ or is $2$-adjacent to an element of $\frakP(\bfv,\| \bfv \|_\infty)$.

Now, given $\bfy \in \frakP(\bfv,\Omega(\bfv))$, both $\bfx$ and $\bfy$ belong,
or are adjacent to $\frakP(\bfv,\| \bfv \|_\infty)$, so they are $2$-connected in $\frakP(\bfv,\Omega(\bfv))$ because:
    \begin{itemize}
    \item
    $\frakP(\bfv,\| \bfv \|_\infty) \subseteq \cup_{n=0}^\infty \bfT_n$,
    thanks to Propositions~\ref{prop:pn_tn} and~\ref{prop:pn_union},
    \item
    $\cup_{n=0}^\infty \bfT_n$ is $2$-connected: it is an increasing union of sets $\bfT_n$
    which are $2$-connected thanks to Proposition~\ref{prop::Tn_conn},
    \item
    $\cup_{n=0}^\infty \bfT_n \subseteq \frakP(\bfv,\Omega(\bfv))$, thanks to Proposition~\ref{prop::T_n}.
    \end{itemize}
\item
Assume now $\bfv  \not \in \Fthree$ and  there exists $n \in \bbN$ such that $\bfv^{(n)}_1 = 0$
with $\dim_\bbQ(\bfv^{(n)}_2,\bfv^{(n)}_3 ) = 2$.
The arithmetical discrete plane $\frakP(\bfv,\Omega(\bfv))$ is $2$-connected
if so is $\frakP(\bfv^{(n)},\Omega(\bfv^{(n)}))$, by Theorem~\ref{theo:algo_cc}.
We conclude by noticing that $\frakP(\bfv^{(n)},\Omega(\bfv^{(n)})$ is $2$-connected thanks to Lemma~\ref{cor::dim2}.
\end{itemize}
We now prove the converse implication.
We thus assume  that the  arithmetical discrete plane $\frakP(\bfv,\Omega(\bfv))$ is $2$-connected.
\emph{A priori}, several cases occur.
\begin{enumerate}
\item
Suppose $\dim_{\bbQ}(\bfv_1,\bfv_2,\bfv_3)=1$.
We first assume $\bfv \in \bbZ^3$ and $\gcd(\bfv_1,\bfv_2,\bfv_3)=1$.
Let $n_1 \in \bbN$ be  such that $\bfv^{(n_1)} _1 = 0$,
and let $n_2 \in \bbN$ such that $\bfv^{(n_2)} _2=0$.
Then, $\frakP(\bfv,\Omega(\bfv))$   $2$-connected  implies that $\frakP(\bfv^{(n)} _3,\Omega(\bfv^{(n_)} _3))$
is connected for $n \geq n_2$, by Theorem~\ref{theo:algo_cc} together with  Lemma~\ref{lem::projection}.
Now, from Lemma~\ref{lem::d=1}, $\Omega(\bfv^{(n)} _3) = 0$ and $\frakP(\bfv^{(n)} _3,0)$ is empty,
hence not connected (by definition), a contradiction.
This also implies   that  we cannot  have $2$-connectedness if $\dim_{\bbQ}(\bfv_1,\bfv_2,\bfv_3)=1$
(even if $\bfv \not \in \bbZ^3$   and $\gcd(\bfv_1,\bfv_2,\bfv_3)\neq 1$).
\item
Suppose $\dim_{\bbQ}(\bfv_1,\bfv_2,\bfv_3)>1$ with  $\bfv \not \in \Fthree$.
Moreover, suppose $\bfv^{(n)}_1 >0$ for all $n \in \bbN$.
Let $n \in \bbN$ such that $\bfv^{(n)} _1 + \bfv^{(n)} _2 \leq \bfv^{(n)} _3$.
One has  $\frakP(\bfv^{(n)},\Omega(\bfv^{(n)}))$ $2$-connected by Theorem~\ref{theo:algo_cc}.
According to  Lemma~\ref{cor::not_in_F3}, $\Omega(\bfv^{(n)}) = \| \bfv^{(n)} \|_\infty$
and the plane  $\frakP(\bfv^{(n)},\Omega(\bfv^{(n)}))$ is not $2$-connected, a contradiction.
\end{enumerate}
By Proposition~\ref{prop:F3}, the two remaining possible cases  are either $\bfv \in \Fthree$,
or $\dim_\bbQ(\bfv^{(n)}_2,\bfv^{(n)}_3 ) = 2$ and there exists $n \in \bbN$ such that $\bfv^{(n)}_1 = 0$.
\end{proof}



\section*{Acknowledgment}
The authors warmly thank \'Eric Domenjoud and Laurent Vuillon for many helpful discussions.
This research has been partly driven by some computer experiments using the Sage software~\cite{Sage}.
This  work is  supported by \emph{ANR KIDICO}, through contract \texttt{ANR-2010-BLAN-0205},
and by \emph{ANR Dyna3S}, through contract \texttt{ANR-13-BS02-0003}.

\bibliographystyle{amsalpha}
\bibliography{biblio}


\end{document}